%% file: main.tex
\documentclass[twoside,leqno]{article}

\usepackage[letterpaper]{geometry}

\usepackage{siamproceedings}

\usepackage{mathtools}
\usepackage{amssymb}
\usepackage{bm}
\usepackage{thmtools}
\usepackage{adjustbox}
\usepackage{comment}
\usepackage{colortbl}

\usepackage[svgnames,table,dvipsnames]{xcolor}
\usepackage{nicematrix}

\usepackage[round]{natbib}
\setcitestyle{notesep={; }}
\usepackage[inline]{enumitem}
\usepackage[subrefformat=parens]{subcaption}
\usepackage{doi}

\definecolor[named]{ACMBlue}{cmyk}{1,0.1,0,0.1}
\definecolor[named]{ACMYellow}{cmyk}{0,0.16,1,0}
\definecolor[named]{ACMOrange}{cmyk}{0,0.42,1,0.01}
\definecolor[named]{ACMRed}{cmyk}{0,0.90,0.86,0}
\definecolor[named]{ACMLightBlue}{cmyk}{0.49,0.01,0,0}
\definecolor[named]{ACMGreen}{cmyk}{0.20,0,1,0.19}
\definecolor[named]{ACMPurple}{cmyk}{0.55,1,0,0.15}
\definecolor[named]{ACMDarkBlue}{cmyk}{1,0.58,0,0.21}
\hypersetup{
  colorlinks,
  linkcolor=ACMPurple,
  citecolor=ACMPurple,
  urlcolor=ACMDarkBlue,
  filecolor=ACMDarkBlue,
}

\usepackage{newfloat}
\DeclareFloatingEnvironment[name=Listing,within=section]{listing}
\usepackage{tikz}
\usepackage{tikz-qtree}
\usetikzlibrary{arrows}
\usetikzlibrary{backgrounds}
\usetikzlibrary{calc}
\usetikzlibrary{decorations.pathreplacing}
\usetikzlibrary{intersections}
\usetikzlibrary{fit}
\usetikzlibrary{patterns}
\usetikzlibrary{positioning}
\usetikzlibrary{shapes}
\usetikzlibrary{shadows}
\usetikzlibrary{shapes.geometric}
\usetikzlibrary{tikzmark}

\algnewcommand{\IfThenElse}[3]{% \IfThenElse{<if>}{<then>}{<else>}
  \algorithmicif\ #1\ \algorithmicthen\ #2\ \algorithmicelse\ #3}

\makeatletter
  \AddToHook{env/algorithmic/begin}{\def\@currentcounter{ALG@line}}
\makeatother

\newcommand{\crefrangeconjunction}{--}
\makeatletter
  \DeclareRobustCommand{\labelcrefrange}[2]{\@crefrangenostar{labelcref}{#1}{#2}}
\makeatother
\crefname{line}{line}{lines}
\Crefname{line}{Line}{Lines}
\crefalias{ALG@line}{line}
\newcommand{\crefnameof}[1]{\csname cref@#1@name\endcsname}

\newsiamremark{remark}{Remark}
\newsiamthm{conjecture}{Conjecture}

\DeclarePairedDelimiter{\set}{\lbrace}{\rbrace}
\DeclarePairedDelimiter{\ceil}{\lceil}{\rceil}
\DeclarePairedDelimiter{\floor}{\lfloor}{\rfloor}
\DeclarePairedDelimiter{\abs}{\lvert}{\rvert}
\DeclarePairedDelimiter{\setr}{\lparen}{\rparen}
\DeclarePairedDelimiter{\setb}{\lbrack}{\rbrack}

\usepackage[most]{tcolorbox}
\makeatletter
\newtcolorbox{quotebox}[1][]{
  colback=white,          % Background color
  colframe=gray!70,       % Frame color (used for left bar)
  left=5pt,               % Space between bar and text
  right=5pt,              % Space on the right
  top=0pt,
  bottom=0pt,
  sharp corners,          % Square corners
  boxrule=0pt,            % No border line around
  borderline west={2pt}{0pt}{gray!70}, % Left bar
  enhanced,
  grow to left by=-2\parindent,
  after app={\@endparenv},
  #1
}
\makeatother

\makeatletter
\newtcolorbox{thmbox}[1][]{%
  colback=white,          % Background color
  colframe=black,         % Frame color (used for left bar)
  after app={\@endparenv},
  #1
}
\makeatother

\input{commands}  % Custom commands
\input{patches}   % SIAM Patches

\newcommand{\fullcitation}{This manuscript appears as: Thomas~L. Draper and Feras~A. Saad. 2026.
Efficient Online Random Sampling via Randomness Recycling.
In \textit{Proceedings of the 2026 Annual ACM-SIAM Symposium on Discrete Algorithms}.
Society for Industrial and Applied Mathematics,
Philadelphia, PA, 2473--2511.
\href{https://doi.org/10.1137/1.9781611978971.89}{doi:\nolinkurl{10.1137/1.9781611978971.89}}}

\begin{document}

\newcommand\relatedversion{}
\renewcommand\relatedversion{\thanks{\fullcitation}}

\title{\Large Efficient Online Random Sampling via Randomness Recycling\relatedversion}
\author{%
  Thomas L.~Draper\footnotemark[2]
  \and
  Feras A.~Saad%
  \thanks{Computer Science Department, Carnegie Mellon University (\email{tdraper@cmu.edu}, \email{fsaad@cmu.edu}).}
  }

% \fancyfoot[R]{\scriptsize{Copyright \textcopyright\ 2026 by SIAM\\
% Unauthorized reproduction of this article is prohibited}}

\maketitle
\date

\input{abstract}

\clearpage
\bgroup
\hypersetup{hidelinks}
\noindent\textbf{\large Contents}
\let\Contentsline\contentsline
\renewcommand\contentsline[3]{\Contentsline{#1}{#2}{}}
% \usepackage{tocloft}
% \renewcommand{\cftdot}{} %empty {} for no dots. you can have any symbol
% inside. For example put {\ensuremath{\ast}} and see what happens.

\tableofcontents
\egroup
\clearpage

\input{sec-intro}

\input{sec-random-states}
\input{sec-analysis}
\input{sec-uniform}
\input{sec-nonuniform}
\input{sec-related}
\input{sec-remarks}

\section*{Acknowledgements}
The authors thank the anonymous referees for helpful suggestions in improving
the manuscript.
The authors also thank David G.~Harris for observing that
the space lower bound in \cref{conj:entropy-cost-tight}
need not depend on $d$.
This material is based upon work supported by the National Science
Foundation under Grant No.~2311983. Any opinions, findings, and conclusions
or recommendations in this material are those of the authors and do not
necessarily reflect the views of the funding agencies.

\renewcommand{\bibsection}{\origsection*{References}}
\bibliography{main}

\appendix

% Make algorithms numbered in appendix.
% https://tex.stackexchange.com/questions/118632/change-caption-number-of-an-algorithm
\renewcommand{\thealgorithm}{\thesection\arabic{algorithm}}
\setcounter{algorithm}{0}
\input{appx-baselines}

\end{document}

%% file: commands.tex
\newcommand\Tstrut{\rule{0pt}{2.25ex}}

\newcommand{\clrA}{yellow}
\newcommand{\clrB}{blue!30!white}
\newcommand{\clrC}{green!30!white}
\newcommand{\clrD}{red!30!white}

\newcommand{\cA}[1][\mathrm{A}]{\cellcolor{\clrA}#1}
\newcommand{\cB}[1][\mathrm{B}]{\cellcolor{\clrB}#1}
\newcommand{\cC}[1][\mathrm{C}]{\cellcolor{\clrC}#1}
\newcommand{\cD}[1][\mathrm{D}]{\cellcolor{\clrD}#1}

\newcommand{\randS}{Z}
\newcommand{\randM}{M}
\newcommand{\randMin}{\randM_{\min}}
\newcommand{\setp}{\twoheadleftarrow}

\newcommand{\Nat}{\mathbb{N}}

\newcommand{\Rat}{\mathbb{Q}}
\newcommand{\expect}[1]{\mathbb{E}\left[{#1}\right]}

\newcommand{\entropy}[1]{H\setr*{#1}}
\newcommand{\entropyRV}[1]{H\setb*{#1}}

\newcommand{\ky}[1][]{\textrm{KY}\ifthenelse{\equal{#1}{}}{}{\left(#1\right)}}

\newcommand{\twodots}{\mathinner{\ldotp\ldotp}}
\newcommand{\dotgiven}{\,\cdot\,|\,}

\newcommand{\defeq}{\coloneqq}
\newcommand{\eqdef}{\eqqcolon}

\newcommand{\Alphabet}{\mathcal{X}}
\newcommand{\DAlphabet}{\Delta\Alphabet}

\newcommand{\AState}{\mathcal{S}}
\newcommand{\astate}{s}
\newcommand{\Astate}{S}
\newcommand{\Astates}{\bm{\Astate}}

\newcommand{\dist}{p}               % fixed distribution
\newcommand{\Dist}{P}               % random distribution
\newcommand{\dists}{\bm{\dist}}     % sequence of fixed distributions
\newcommand{\Dists}{\bm{\Dist}}     % sequence of random distributions

\newcommand{\flip}{c}               % fixed flip
\newcommand{\Flip}{C}               % random flip
\newcommand{\flips}{\bm{\flip}}     % sequence of fixed flips
\newcommand{\Flips}{\bm{\Flip}}     % sequence of random flips

\newcommand{\nflip}{t}                % fixed flip counter
\newcommand{\nFlip}{T}                % random flip counter
\newcommand{\iflip}{v}                % fixed flip incremental counter
\newcommand{\iFlip}{V}                % random flip incremental counter
\newcommand{\tosses}[2]{\nflip_{#1}(#2)}

\newcommand{\Alg}{f}

\newcommand{\Prob}{\mathbb{P}}

\newcommand{\bX}{\bm{X}}

\makeatletter
\newcommand\approxsim{\mathpalette\@approxsim\relax}
\newcommand\@approxsim[2]{%
  \mathrel{%
    \ooalign{%
      $\m@th#1\sim$\cr
      \hidewidth$\m@th#1:$\hidewidth\cr
    }%
  }%
}
\makeatother

\newcommand{\Distribution}[1]{\mathrm{#1}}
\newcommand{\Uniform}{\Distribution{Uniform}}
\newcommand{\Discrete}{\Distribution{Discrete}}
\newcommand{\Bernoulli}{\Distribution{Bernoulli}}

%% file: patches.tex
\usepackage{xparse}

\let\origsection\section
\let\origsubsection\subsection
\let\origsubsubsection\subsubsection
\let\origparagraph\paragraph
\let\origsubparagraph\subparagraph

\RenewDocumentCommand{\section}{ s o m }{%
  \IfBooleanTF{#1}
    {\origsection*{#3.}}% starred: print only
    {\IfNoValueTF{#2}
       {\origsection[#3]{#3.}}% no optional arg: TOC gets clean title, heading gets dotted
       {\origsection[#2]{#3.}}% optional arg present: TOC uses it, heading gets dotted
    }%
}

\RenewDocumentCommand{\subsection}{ s o m }{%
  \IfBooleanTF{#1}
    {\origsubsection*{#3.}}
    {\IfNoValueTF{#2}
       {\origsubsection[#3]{#3.}}
       {\origsubsection[#2]{#3.}}
    }%
}

\RenewDocumentCommand{\subsubsection}{ s o m }{%
  \IfBooleanTF{#1}
    {\origsubsubsection*{#3.}}
    {\IfNoValueTF{#2}
       {\origsubsubsection[#3]{#3.}}
       {\origsubsubsection[#2]{#3.}}
    }%
}

\RenewDocumentCommand{\paragraph}{ s o m }{%
  \IfBooleanTF{#1}
    {\origparagraph*{#3.}}
    {\IfNoValueTF{#2}
       {\origparagraph[#3]{#3.}}
       {\origparagraph[#2]{#3.}}
    }%
}

\RenewDocumentCommand{\subparagraph}{ s o m }{%
  \IfBooleanTF{#1}
    {\origsubparagraph*{#3.}}
    {\IfNoValueTF{#2}
       {\origsubparagraph[#3]{#3.}}
       {\origsubparagraph[#2]{#3.}}
    }%
}

\renewenvironment{abstract}{\begin{@abssec}{\abstractname}}{\end{@abssec}}

\makeatletter
\renewcommand{\tableofcontents}{\@starttoc{toc}}
\makeatother

\crefformat{section}{#2\S#1#3}
\Crefformat{section}{#2\S#1#3}
\crefmultiformat{section}{#2\S#1#3}{ and~#2\S#1#3}{, #2\S#1#3}{, and~#2\S#1#3}
\Crefmultiformat{section}{#2\S#1#3}{ and~#2\S#1#3}{, #2\S#1#3}{, and~#2\S#1#3}
\crefrangeformat{section}{#3\S#1#4\crefrangeconjunction{}#5\S#2#6}

%% file: abstract.tex
%!TEX root=./main.tex
\begin{abstract}
This article studies the fundamental problem
of using i.i.d.~coin tosses from an entropy source
to efficiently generate random variables $X_i \sim \Dist_i$ $(i \ge 1)$,
where $(\Dist_1, \Dist_2, \dots)$ is a random sequence
of rational discrete probability
distributions subject to an \textit{arbitrary} stochastic process.
Our method achieves an amortized expected entropy cost within
$\varepsilon > 0$ bits of the information-theoretically optimal
\citeauthor{shannon1948} lower bound
using $O(\log(1/\varepsilon))$ space.
This result holds both pointwise in terms of the
Shannon information content conditioned on $X_i$ and $\Dist_i$,
and in expectation to obtain a rate of
$\mathbb{E}[H(P_1) + \dots + H(P_n)]/n + \varepsilon$
bits per sample as $n \to \infty$ (where $H$ is the Shannon entropy).
The combination of space, time, and entropy properties of our method
improves upon the \citet{knuth1976} entropy-optimal algorithm and
\citet{han1997} interval algorithm for online sampling, which require
unbounded space.
It also uses exponentially less space than the more specialized methods of
\citet{Kozen2022} and \citet{shao2025} that generate i.i.d.~samples from a
fixed distribution.
Our online sampling algorithm rests on a powerful algorithmic technique called
\textit{randomness recycling}, which reuses a fraction of the random
information consumed by a probabilistic algorithm to reduce its amortized
entropy cost.

On the practical side, we develop randomness recycling techniques to
accelerate a variety of prominent sampling algorithms,
which include uniform sampling, inverse transform sampling, lookup-table
sampling, alias sampling, and discrete distribution generating (DDG) tree
sampling.
We show that randomness recycling enables state-of-the-art runtime
performance on the Fisher-Yates shuffle when using a cryptographically
secure pseudorandom number generator, and that it reduces the entropy cost of
discrete Gaussian sampling.
Accompanying the manuscript is a performant software library in the C
programming language.
\end{abstract}

%% file: sec-intro.tex
%!TEX root=./main.tex

\section{Introduction}
\label{sec:intro}

Let $\Alphabet = \set{\alpha_1, \alpha_2, \dots, \alpha_l}$ be a finite alphabet of $l \ge 1$
distinct symbols and
\begin{align}
\DAlphabet \defeq \set*{\dist: \Alphabet \to [0,1] \cap \Rat \,\middle|\, \sum_{i = 1}^{l} \dist(\alpha_i) = 1} \subset [0,1]^{\Alphabet}
\end{align}
be the set of all rational probability distributions on $\Alphabet$.
Suppose there is a $\DAlphabet$-valued random sequence
$\Dists \defeq (\Dist_i)_{i\ge 1}$ of probability distributions,
subject to any stochastic process, that we aim to sample from.
This article studies the fundamental problem of using
i.i.d.~unbiased coin tosses
$\Flips \defeq (\Flip_i)_{i \ge 1}$
from an entropy source
to generate an output sequence $\bX \defeq (X_i)_{i\ge1}$ of
random variables $X_i \sim \Dist_i$, in the following setting.
\begin{thmbox}[%
  enhanced,
  attach boxed title to top center={yshift=-3mm,yshifttext=-1mm},
  colbacktitle=white,
  boxed title style={size=small,colframe=white},
  coltitle=black,
  title={\textbf{Problem Statement}: Online Sampling a Dynamic Sequence of Probability Distributions}
  ]

Goal: Design an \textit{online random sampling} algorithm which, given
i.i.d.~fair coins and access to a dynamic sequence of probability
distributions, executes the following loop:
\begin{quotebox}
\textbf{for} $i = 1, 2, 3, \dots$
\begin{enumerate}[topsep=2pt, noitemsep, label=(S\arabic*),itemindent=10pt]
\item \label{step:observe} Receive the next target distribution $\Dist_i$
\item \label{step:generate} Generate the next random variable $X_i \sim \Dist_i$
\end{enumerate}
\end{quotebox}
\end{thmbox}
The efficiency of online random sampling algorithm is measured in terms of
its space, time, and entropy complexities, which are limited
computational resources to be conserved.
The random sequence $\Dists$ of target distributions
can itself follow any law over $(\DAlphabet)^\Nat$,
such as deterministic, i.i.d., exchangeable, stationary, Markov, etc.
For example, each $\Dist_i$ may depend on the previous distributions
$\Dist_{<i} \defeq (\Dist_j)_{j<i}$, on the previous generated variables $X_{<i}$,
or on other sources of exogenous randomness or nondeterminism.
The described setup is very general and covers most
computational settings that involve sequentially generating discrete random
variables, such as randomized algorithms, Markov chain Monte Carlo, or
probabilistic programs for stochastic simulation.

\subsection{Computational Model}
\label{sec:intro-model}

\input{fig-ors}

We formalize the space, time, and entropy complexity of
online sampling algorithms that use i.i.d.~coin tosses
and propagate some auxiliary state that across each round of
\labelcref{step:generate,step:observe}.
\Cref{lst:ors} defines an online random sampling algorithm in the
special case that $\Dists = \dists$ is a \textit{fixed} $\DAlphabet$-valued
sequence of probability distributions.
The correctness requirement~\cref{eq:ors-correct} implies that
when the online random sampling algorithm is instead given a \textit{random}
$\DAlphabet$-valued sequence
$\Dists$ over the same probability space $(\Omega, \mathcal{E}, \Prob)$,
its output $X_i$ at each round $i \ge 1$ is correctly
drawn from the dynamically provided target distribution
\begin{align}
\Prob(X_i = x_i \mid X_{<i}=x_{<i}, \Dist_{\le i} = \dist_{\le i}) = \dist_i(x_i),
\label{eq:ors-correct-rv}
\end{align}
subject to a natural independence constraint between
the coin process $\Flips$ and distribution process $\Dists$:
\begin{alignat}{2}
\Dist_i \perp \Flips~\big|~X_{<i}. \label{eq:ors-fresh}
\end{alignat}
\Cref{eq:ors-fresh} states that $\Dist_i$ cannot depend on fresh coins to
be consumed by the sampler in future rounds, and that it cannot depend on
the internal history of the sampler other than through its generated
outputs $X_{<i}$.

\subsubsection{Entropy Cost}
In \cref{lst:ors}, the counters $\nflip_n \defeq \iflip_1 + \dots + \iflip_n$ of
consumed coin tosses define the \textit{entropy cost} of an online
sampling algorithm, i.e., the total
number of tosses consumed to generate outputs $(x_1, \dots, x_n)$.
For a probabilistic run of online random sampling
with random coin and distribution
sequences $(\Flips, \Dists)$,
we notate the random entropy cost explicitly
as $\nFlip_n \equiv \tosses{n}{\Flips, \Dists}$.
For a probabilistic run where the distribution sequence $\dists$ is fixed,
we write the entropy cost $\tosses{n}{\Flips, \dists}$.
For a fixed distribution sequence $\dists$,
the \textit{expected entropy cost}
$\expect{\tosses{n}{\Flips, \dists}}$
is the entropy cost averaged over random coin tosses.

\subsubsection{Shannon's Fundamental Rate}

Suppose momentarily that $\dists$ is a fixed
distribution sequence and that we aim to generate an independent output
sequence $\bX = (X_i \sim \dist_i)_{i \ge 1}$.
\citet[Theorem 2.2]{knuth1976} show that the expected number of coin tosses
$t_n(\Flips, \dists)$ needed to generate $(X_1,\dots,X_n)$,
conditioned on any positive probability event $\set{X_1 = x_1, \dots, X_n = x_n}$
is at least the Shannon \textit{information content} (or surprisal)
of the observed symbols, which implies the following fundamental lower bounds
for every integer $n \ge 1$:
\begin{align}
\expect{\tosses{n}{\Flips, \dists} \mid X_1 = x_1, \dots, X_n = x_n}
&\geq \sum_{i=1}^{n}\log\left(\frac{1}{\dist_i(x_i)} \right),
&&
\label{eq:shannon-entropy-surprisal}
\\
\expect{\tosses{n}{\Flips, \dists}}
&\geq \sum_{i=1}^{n}\entropy{\dist_i}
&&\mbox{where } \entropy{\dist} \defeq \sum_{\alpha \in \Alphabet} \log \left(\frac{1}{\dist(\alpha)}\right) \dist(\alpha).
\label{eq:shannon-entropy}
\end{align}
\Cref{eq:shannon-entropy} states that the expected number of coin tosses
is lower bounded by the sum of \textit{Shannon entropies} $\entropy{\dist}$,
which comports with the source coding theorem of \citet{shannon1948}.

\paragraph{Complexity}
The rates~\cref{eq:shannon-entropy-surprisal,eq:shannon-entropy}
capture entropy-cost lower bounds of online random sampling using
i.i.d.~coin tosses in terms of information content and Shannon entropy.
The entropy complexity of an online sampling algorithm is defined as its
entropy cost $\tosses{n}{\cdot}$.
Its space complexity is captured by the auxiliary
states $\Astates = (\Astate_i)_{i \ge 1}$ propagated
across rounds, and any temporary space used by $\Alg$
within each round.
Its (expected) time complexity is that of the function $\Alg$.
We seek space- and time-efficient online
sampling algorithms whose amortized entropy cost
$\expect{\tosses{n}{\Flips,\Dists}}/n$ as $n$ grows large is as close as possible to the
information-theoretic lower bound.
We assume each target distribution $\dist_i$ is given
either as a (possibly sparse) list of integer weights $(w_1, \dots, w_l)$
where $\dist_i(\alpha_k) = w_k / (w_1 + \dots + w_l)$,
or as an efficiently computable cumulative distribution
function $k \mapsto \sum_{j \le k}\dist_i(\alpha_j)$.

\begin{remark}
We will reserve the notation $\entropy{\dist}$ to denote the entropy of a probability
distribution $\dist \in \DAlphabet{}$, and instead use square braces for
the entropy of a random variable,
i.e., $\entropyRV{X} \defeq -\sum_{\alpha \in \Alphabet} \log (\Prob(X=\alpha)) \Prob(X=\alpha)$.
\end{remark}

\begin{remark}
For a random distribution $\Dist$,
the entropy $\entropy{\Dist} \defeq \omega \mapsto \entropy{\Dist(\omega)}$
is itself a random variable,
which is of course distinct from the entropy
$\entropyRV{\Dist} = -\sum_{p \in \DAlphabet}\log(\Prob(\Dist=p))\Prob(\Dist=p)$
of the random element $\Dist$ (an irrelevant quantity in this article).
\end{remark}

\begin{remark}
When applying the bounds from \cref{eq:shannon-entropy-surprisal,eq:shannon-entropy}
to a random distribution sequence $\Dists$,
it is essential to consider
the surprisals and entropies of the random distributions
$(\Dist_1, \dots, \Dist_n)$,
and \textit{not} the joint entropy $\entropyRV{X_1, \dots, X_n}$ of the output sequence.
This distinction is subtle, but obvious from a simple example.
Suppose $\Dist_1 = \delta_{\alpha_1}$ or $\Dist_1 = \delta_{\alpha_2}$ with
equal probability, i.e., $\Dist_1$ is always realized as a degenerate distribution.
Any sampling algorithm can generate $X_1 \sim \Dist_1$ using zero coin
tosses for any realization of $\Dist_1$, whereas $\entropyRV{X_1} = 1$ bit.
\end{remark}

\begin{remark}
\label{remark:output-sequence-process}
A common case for online sampling is generating an output sequence
$\bX \defeq (X_i)_{i \ge 1}$
subject to an arbitrary stochastic process whose distribution
is not influenced by any external randomness.
In particular, each random distribution $\Dist_i$
is itself fully determined by $X_{<i}$, allowing us to write
$\Dists = \set{p_i(\dotgiven{} X_1, \dots, X_{i-1})}_{i \ge 1}$,
where
$p_i(\dotgiven X_1, \dots, X_{i-1}) \defeq \Prob(X_i = \cdot \mid X_1, \dots, X_{i-1})$
is the conditional distribution of $X_i$ given all previously generated outputs for $i > 1$.
In this special case, the bound~\cref{eq:shannon-entropy} on the expected entropy cost
of generating $(X_1,\dots,X_n)$ becomes the more familiar entropy of the output sequence:
\begin{align}
\expect{\tosses{n}{\Flips, \Dists}}
  \ge \expect{\sum_{i=1}^{n} \entropy{p_i(\dotgiven X_1, \dots, X_{i-1})}}
  = \entropyRV{X_1, \dots, X_n}.
  \label{eq:entropy-cost-output-entropy}
\end{align}
\end{remark}

\subsection{Main Result}
\label{sec:intro-result}

The main contribution of this article is an online random
sampling algorithm whose amortized entropy cost
can be made arbitrarily close to the Shannon bound in the limit, using
bounded auxiliary space across rounds.
Recalling the lower bound~\cref{eq:shannon-entropy-surprisal}
from \citet[Theorem 2.2]{knuth1976},
the precise result on entropy and space complexity is as follows,
where
$\DAlphabet_d \defeq \set{ \dist \in \DAlphabet{} \mid d \ge \min_{k}\set{k \in \Nat \mid \forall \alpha \in \Alphabet.~k\dist(\alpha) \in \Nat} }$
is the set of rational probability distributions over $\Alphabet$ whose probabilities
have common denominator at most $d$.

\begin{thmbox}
\begin{restatable}{theorem}{EntropyCost}
\label{theorem:entropy-cost}
For any $\varepsilon > 0$ and $d \ge 1$,
there exists an online random sampling algorithm
using a sequence $\Flips$ of i.i.d.~coin tosses such
that, for every distribution sequence $\dists \in (\DAlphabet_d)^{\Nat}$,
the entropy cost of generating
an output sequence $\bX = (X_i \sim \dist_i)_{i \ge 1}$ satisfies
\begin{align}
\expect{\tosses{n}{\Flips, \dists} \mid X_1=x_1, \dots, X_n=x_n}
  < \sum_{i=1}^{n}\log \left( \frac{1}{\dist_i(x_i)} \right) + \varepsilon{n} + W_{d,\varepsilon}
&& (n \ge 1),
\label{eq:main-thm-cost}
\end{align}
where $\set{X_1=x_1, \dots, X_n=x_n}$ is any positive probability event
and $W_{d,\varepsilon} \sim \log(d/\varepsilon)$ as $d/\varepsilon \to \infty$.
The algorithm uses auxiliary space of at most $2W_{d,\varepsilon}$ bits
across rounds and $O(W_{d,\varepsilon})$ temporary space \mbox{per round.}
\end{restatable}
\end{thmbox}

This strong result establishes an online sampling algorithm whose entropy cost is
$\varepsilon$-close to the information-content lower
bound~\cref{eq:shannon-entropy-surprisal} for any $n$,
\textit{uniformly in the distribution sequence} $\dists$.
In particular, \cref{theorem:entropy-cost} implies
the expected entropy cost of the online sampling algorithm is
close to the Shannon-entropy lower bound~\cref{eq:shannon-entropy}:
\begin{align}
\expect{\tosses{n}{\Flips, \dists}} < \sum_{i=1}^{n} \entropy{\dist_i} + \varepsilon{n} + W_{d,\varepsilon}
\label{eq:main-thm-cost-shannon}
&&(n \ge 1).
\end{align}
\Cref{eq:main-thm-cost,eq:main-thm-cost-shannon} are the analogues of
\cref{eq:shannon-entropy-surprisal,eq:shannon-entropy}
when the target sequence $\dists$ is fixed.
If the distribution sequence $\Dists$ follows an arbitrary
$\DAlphabet{}_d$-valued stochastic process subject to \cref{eq:ors-fresh},
then \cref{theorem:entropy-cost} also implies the following analogues
of \cref{eq:shannon-entropy-surprisal,eq:shannon-entropy}:
\begin{align}
\expect{\tosses{n}{\Flips, \Dists} \mid X_{\le n}=x_{\le n}, \Dist_{\le n}=\dist_{\le n}}
  &< \sum_{i=1}^{n}\log \left( \frac{1}{\dist_i(x_i)} \right) + \varepsilon{n} + W_{d,\varepsilon}
&& (n \ge 1),
\label{eq:main-thm-cost-random}
\\
\expect{\tosses{n}{\Flips, \Dists}}
  &< \sum_{i=1}^{n} \expect{\entropy{\Dist_i}} + \varepsilon{n} + W_{d,\varepsilon}
&& (n \ge 1).
\label{eq:main-thm-cost-random-shannon}
\end{align}

We conjecture that the space-entropy combination in
\cref{theorem:entropy-cost} is the best possible,
up to constant factors and dependence on $d$.

\begin{thmbox}
\begin{restatable}{conjecture}{EntropyCostTight}
\label{conj:entropy-cost-tight}
Any online random sampling algorithm that generates exact samples from a sequence of
arbitrary discrete distributions,
within $\varepsilon > 0$ of the information-theoretically
optimal entropy rate,
using a stream of i.i.d.~coin tosses as the entropy source,
requires $\Omega(\log(1/\varepsilon))$ bits of space
for an auxiliary state that is carried over between rounds.
\end{restatable}
\end{thmbox}

For the conditional distribution sequence considered in
\cref{remark:output-sequence-process}, \cref{eq:main-thm-cost-random-shannon}
gives the following corollary.

\begin{thmbox}
\begin{corollary}
\label{corollarly:entropy-cost}
Let $\Dists \defeq \set{p_i( \dotgiven X_1, \dots, X_{i-1})}_{i \ge 1}$
denote a sequence of conditional distributions
for random variables $\bX \defeq (X_i)_{i \ge 1}$,
whose probabilities have common denominator $d \ge 1$.
For any $\varepsilon > 0$,
there exists an online random sampling algorithm that generates $\bX$
whose entropy cost satisfies
\begin{align}
\expect{\tosses{n}{\Flips, \Dists}} < \entropyRV{X_1,\dots,X_n} + \varepsilon{n} + W_{d,\varepsilon}
&& (n \ge 1).
\label{eq:corollary-cost}
\end{align}
\end{corollary}
\end{thmbox}

\begin{remark}
In the special case of an i.i.d.~sequence with common distribution
$\Dist$, \cref{theorem:entropy-cost} furnishes, for any $\varepsilon > 0$, an
algorithm with amortized entropy cost asymptotically less than
$\entropy{\Dist}+\varepsilon$ tosses per output.
\end{remark}

The algorithm witnessing \cref{theorem:entropy-cost} enjoys several
desirable theoretical properties.
\begin{itemize}[wide, leftmargin=*, noitemsep]
  \item It works in an online setting, where the distributions
    $\Dists$ are given sequentially and follow a random process.

  \item It guarantees an $\varepsilon > 0$ bound on the amortized
    entropy consumption using $O(\log(1/\varepsilon))$ space, which remains
    bounded as the number $n$ of generated samples grows. The methods of
    \citet{knuth1976,han1997} achieve $\varepsilon = 0$ entropy loss;
    however, they require unbounded memory even for distributions in
    $\DAlphabet_d$, i.e., distributions
    whose probabilities have common denominator bounded by $d \ge 1$.

  \item Compared to the finite-memory implementation of the
    \citet{han1997} interval (arithmetic coding) method from \citet{uyematsu2003},
    our method
    \begin{itemize}[noitemsep, nosep]
      \item produces exact samples from the target distributions, not approximate samples;

      \item uses efficient integer arithmetic, not expensive arbitrary-precision arithmetic.
    \end{itemize}

  \item Compared to the more specialized sampling algorithms of
    \citet{Kozen2022} and \citet{shao2025} that sample a \textit{fixed}
    i.i.d.~sequence $\Dists$, our method obtains an
    exponential improvement, using $O(\log(1/\varepsilon))$ rather than
    $O(1/\varepsilon)$ space to achieve a desired $\varepsilon > 0$ entropy
    bound.
\end{itemize}

On the practical side, we show that our method delivers the following improvements.
\begin{itemize}[wide,leftmargin=*, noitemsep]

  \item It enables speedups to generating random permutations using the
  \citet{fisher1953} shuffle, compared to the state-of-the-art method of
  \citet{brackett2025} (\cref{fig:shuffle}).

  \item It enables speedups to several prominent sampling algorithms for
  discrete probability distributions
  (\cref{fig:benchmark-uniform,fig:benchmark-uniform-full,fig:benchmark})
  and reduces the entropy cost of discrete Gaussian sampling
  (\cref{fig:benchmark-gaussian}).
\end{itemize}

\subsection{Existing Approaches}
\label{sec:intro-existing}

To give further context to our main result, we discuss the strengths and
limitations of existing online random sampling algorithms.
For this discussion it is sufficient to focus on the setting from
\cref{remark:output-sequence-process,corollarly:entropy-cost},
where $\Dists = \set{p_i( \dotgiven X_1, \dots, X_{i-1})}_{i \ge 1}$ is a random sequence
of conditional distributions that are determined solely
by the previously generated outputs.
\Cref{table:compare} summarizes the entropy loss and space and time
complexity of existing baseline algorithms as well as our method (refer to
\cref{appx:baselines} for technical details and concrete implementations of
the baselines).

\subsubsection{Entropy-Optimal Sampling for One Distribution}
\label{sec:intro-existing-single}
For generating a single discrete random variable $X \sim \dist$,
\citet{knuth1976} describe an ``entropy-optimal'' algorithm
with expected entropy cost
\begin{align}
\expect{\tosses{1}{\Flips, \dist}}
= \sum_{\alpha \in \Alphabet} \sum_{i=0}^{\infty} (\lfloor2^{i} p(\alpha) \rfloor \bmod 2)i 2^{-i}
< \entropyRV{X} + 2,
\label{eq:ky-loss}
\end{align}
based on the binary expansions of the $p(\alpha)$,
which is the best possible~\citep[Theorem 2.2]{knuth1976}.
This result shows that the \citeauthor{shannon1948} rate of
$\entropyRV{X}$ tosses per output may not be achievable for a single output.
(The expression $\expect{\tosses{1}{\Flips, \dist}}$ is a minor abuse of
notation: it indicates the expected entropy cost of generating $X_1$ in the
first round of online sampling, and so only one distribution $\dist$ is
notated instead of a sequence $\dists$.)

The usual way to implement the \citeauthor{knuth1976}
method is to explicitly construct an entropy-optimal
``discrete distribution generating'' (DDG)
tree for $\dist$ as a preprocessing step
(e.g., \citet[Algorithm~1]{roy2013}; \citet[Algorithms 5 and 6]{saad2020popl}),
which can require exponential space in the number of bits needed
to encode $\dist$ \citep[Theorem 3.6]{saad2020popl}.
Alternatively, using the cumulative distribution function of $\dist$,
\citet[Algorithm 1]{saad2025} give a linear-space
entropy-optimal algorithm that avoids preprocessing a DDG tree.
However, this method incurs a runtime overhead of computing binary
expansions of target probabilities during sampling.

\input{fig-compare}

\subsubsection{Online Entropy-Suboptimal Sampling}
\label{sec:intro-existing-suboptimal}

For the problem of generating an output sequence given by conditional
distributions specified in an online fashion,
the most straightforward application of \cref{sec:intro-existing-single}
is to use a sequence of entropy-optimal samplers
that target each conditional distribution:
\begin{align}
X_1 &\overset{\ky}{\sim} p_1,
&& X_2 \overset{\ky}{\sim} p_2( \dotgiven X_1),
&&\dots,
&&X_n \overset{\ky}{\sim} p_n( \dotgiven X_1, \dots, X_{n-1}).
\label{eq:baseline-ky-chain-rule-gen}
\end{align}
The expected entropy cost of generating
$(X_1, \dots, X_n)$ using this dynamic approach satisfies
\begin{equation}
\expect{\tosses{n}{\Flips, \Dists}}
< \sum_{i=1}^{n}\left(\entropyRV{X_i \mid X_1,\dots,X_{i-1}}+2\right)
= \entropyRV{X_1,\dots,X_n} + 2n.
\label{eq:baseline-ky-chain-rule-entropy}
\end{equation}
The space complexity remains
bounded as $n$ grows large, since at each step the conditional
distribution $p_i(x_i \mid X_1=x_1, \dots, X_{i-1}=x_{i-1})$
is finitely encodable using at most $d\ceil{\log{d}}$ bits.
The space complexity of generating each $X_i$
may be exponential or linear as described before,
depending on the amount of preprocessing.

The term $2n$ in \cref{eq:baseline-ky-chain-rule-entropy} indicates a
significant waste of random bits, and
the cost bound $\entropyRV{X_1,\dots,X_n}/n + 2$
does not tend to the optimal rate.
\Citet[\S2.3]{devroye2020} describe a ``randomness extraction''
procedure that recycles the random bits in this sequential approach
using the interval method of \citet{han1997} (described in the next paragraph),
which claims (cf.~\cref{footnote:randomness-extraction}) to achieve
% an asymptotically optimal rate
\begin{equation}
\expect{\tosses{n}{\Flips, \Dists}} < \entropyRV{X_1,\dots,X_n} + 2 + o(n).
\end{equation}
However, the space complexity is no longer bounded as the algorithm
requires arbitrary-precision arithmetic.
Randomness extraction uses operations with big integers of unbounded size
as $n \to \infty$ \textit{even if} the rational probabilities have bounded
denominator $d$, compromising performance.

\subsubsection{Online Entropy-Optimal Sampling}
\label{sec:intro-existing-optimal}

The entropy loss of $<2$ bits per sample in \cref{eq:baseline-ky-chain-rule-entropy}
arises because each $X_i$ is treated completely separately;
batched sampling or storing more information across rounds
can drive this loss to zero.
Consider first an offline setting where $n$ is a prespecified value,
and each $\Dist_i$ is a deterministic function of the previous outputs
$X_{<i}$ as in \cref{remark:output-sequence-process}.
Then we can generate $(X_1, \dots, X_n)$
by first reconstructing the joint distribution
$\dist^{(n)}$ from the conditionals and then sampling
\begin{align}
(X_1, \dots, X_n) \overset{\ky}{\sim} p^{(n)}(x_1,\dots,x_n) \defeq
  \prod_{i=1}^{n}p_i(x_i\mid x_1, \dots, x_{i-1}),
\end{align}
where $X \overset{\ky}{\sim} \dist$ means that ``$X$ is generated from $\dist$
using the \citeauthor{knuth1976} method''
as described in \cref{sec:intro-existing-single}.
As $p^{(n)}$ is a bona fide discrete distribution
over $\Alphabet^n$, the expected entropy cost
\begin{align}
\expect{\tosses{1}{\Flips, \dist^{(n)}}} < \entropyRV{X_1,\dots,X_n} + 2
\label{eq:ky-batched}
\end{align}
achieves the best possible rate for each batch size $n$.
The expected amortized entropy cost
$\entropyRV{X_1,\dots,X_n}/n + 2/n$
tends to the Shannon-optimal
rate of
$\entropyRV{X_1,\dots,X_n}/n$ tosses per output as $n \to \infty$.
The disadvantage is that this method is offline.
Further, constructing the joint distributions $p^{(n)}$ requires
exponential space (up to $\abs{\Alphabet}^n$) and its probabilities can
have a large denominator (up to $d^n$).

Fortunately, \citet[p.~384]{knuth1976} show that it is possible to
entropy-optimally sample from $p^{(n)}$ in an online fashion,
without explicitly constructing a DDG tree
or even the joint distribution.
\Citeauthor{knuth1976} show that their tree representations for
entropy-optimal samplers can be dynamically refined step-by-step,
as new conditional distributions are specified.
In fact, only the leaves corresponding to the previously sampled symbol need to
be refined, so this online refinement algorithm avoids the exponential complexity
of the na\"ive batching method~\cref{eq:ky-batched}.
This method achieves the absolute minimum expected entropy cost over all algorithms.
However, the space and runtime complexity per sample grow unbounded
with the number of samples,
because the algorithm requires the binary expansions of rational numbers with
denominators growing as $d^n$ after $n$ samples.
This same issue arises with other exact online samplers that achieve an
$O(1/n)$ entropy loss per sample.

\subsubsection{Interval Method (Arithmetic Coding)}
\label{sec:intro-existing-interval}

The interval algorithm of \citet{han1997} generates samples from
a discrete distribution by recursively subdividing the
unit interval, analogously to arithmetic coding \citep[\S6.3]{mackay2003}.
It extends to sampling i.i.d.~random variables,
a stationary homogeneous Markov process, or an arbitrary process \citep[\S{V}]{han1997}.
When the entropy source emits fair coin tosses,
\citet[Theorem 3]{han1997}
show that generating $(X_1, \dots, X_n)$ using the interval method has
\begin{align}
\expect{\tosses{n}{\Flips, \Dists}} < \entropyRV{X_1,\dots,X_n} + 3.
\label{eq:baseline-interval-entropy}
\end{align}
The interval method does not require batching and is suitable for online
sampling.
Its cost bound
$\entropyRV{X_1,\dots,X_n}/n + 3/n$ also converges to the asymptotically
optimal rate \citep[Eq.~(5.2)]{han1997}.
The disadvantage is that it requires arbitrary-precision arithmetic,
as with the method of \citet{devroye2020} and the online optimal
method of \citet{knuth1976}.
\Citet{uyematsu2003} give a finite-space implementation of the
interval method that uses the usual arithmetic coding technique of rounding
cumulative probabilities, but this version does not return
exact samples.

\subsection{Overview of Randomness Recycling}
\label{sec:intro-overview}

The entropy bound in \cref{theorem:entropy-cost} is enabled
by a technique called \textit{randomness recycling}.
The key idea is as follows.
The auxiliary-state space
$\AState \defeq \set{(z,m) \mid m \ge 1, 0 \le z < m}$
of our online sampling algorithm
is comprised of pairs denoting a uniform draw $z$ over $\set{0,\dots,m-1}$,
with initial state $\astate_0 = (0,1)$.
When sampling $X_i \sim \Dist_i$ at round $i$,
the sampler is given fresh coin tosses $\Flips_{\ge T_{i-1}}$
and the \textit{uniform state}
$S_{i-1} = (\randS_{i-1},\randM_{i-1})$, which is a pair of
random variables such that
$\randS_{i-1} \mid \randM_{i-1} \sim \Uniform[0,\randM_{i-1})$.
(We only consider discrete random variables, so we write
$\Uniform[0,M)$ as shorthand for $\Uniform(\set{0,\dots,M-1})$.)
This uniform state incorporates fresh coin tosses drawn from the
entropy source and randomness that has been recycled from
generating random variables $X_1, X_2, \dots, X_{i-1}$
in previous rounds.
After $X_i$ is generated, $\Astate_{i-1} =(\randS_{i-1},\randM_{i-1})$
is updated to value $\Astate_i = (\randS_i,\randM_i)$ such that
$\randS_i \mid \randM_i \sim \Uniform[0,\randM_i)$.
This auxiliary state will be used to sample
$X_{i+1} \sim \Dist_{i+1}$ in the next round, in accordance with the per-round map
specification~\cref{eq:ors-map}.

% \interfootnotelinepenalty=10000

To meet the correctness requirement~\cref{eq:ors-correct},
randomness recycling requires the following
invariant\footnote{\label{footnote:randomness-extraction}%
Careful reasoning about this invariant is required for generating correct samples according
to \cref{eq:ors-correct}.
For example, the randomness extraction method of \citet[\S2.3, Algorithms 4 and 5]{devroye2020}
does not maintain an analogous invariant in its implementation
of the recycling rule from the \citet{han1997} interval method.
As a result, it does not produce samples from the correct target distributions.
For example, given target distributions $\Dist_1 = \Dist_2 = \mathrm{Bernoulli}(1/4)$
for $X_1$ and $X_2$ (i.i.d.), \citet[Algorithm 5]{devroye2020}
can (under certain configurations) produce $X_2 \sim \mathrm{Bernoulli}(5/16)$,
because the \textit{value} of a recycled bit is correlated with its
\textit{availability} at a point in time.
}
on the auxiliary state to be maintained at all rounds:
\begin{enumerate}[wide, leftmargin=*, label=(I),]
\item \label{invariant} For all $i \ge 1$,
  the global uniform state $(\randS_i,\randM_i)$ satisfies
  $\randS_i \perp (X_1, \dots, X_i) \mid \randM_i$.
\end{enumerate}
Maintaining this global uniform state makes it possible to substantially
reduce the number fresh coin tosses drawn from the entropy source,
by storing unused information from previous coin tosses whenever possible.
By induction, the invariant \labelcref{invariant} is maintained as long as
$(\randS_0, \randM_0) = (0,1)$ and the sampling function used at each step
$i$ fulfills the following contract:
\begin{enumerate}[wide, leftmargin=*, label=(C),]
\item \label{contract} For all $i \ge 1$,
  if the initial global uniform state $(\randS_{i-1},\randM_{i-1})$ satisfies
  $\randS_{i-1} \perp (X_1, \dots, X_{i-1}) \mid \randM_{i-1}$, then
  the updated global uniform state $(\randS_i,\randM_i)$ satisfies
  $\randS_i \perp (X_1, \dots, X_{i}) \mid \randM_i$.
\end{enumerate}

\paragraph{Workflow}

\input{fig-overview}

\Cref{fig:overview} gives an overview of the ``randomness recycling''
workflow for generating a random sequence $(X_i \sim \Dist_i)_{i \ge 1}$.
The entropy source provides an i.i.d.~sequence $\Flips$ of unbiased coin tosses.
These coin tosses are used to create and refill a
global uniform state $(\randS, \randM)$
so that $\randS \mid \randM \sim \Uniform[0,\randM)$.
To sample $X_i \sim \Dist_i$,
a discrete uniform $U_i \sim \Uniform[0, N_i)$ over an appropriate
range is sampled via \cref{alg:uniform}, which accesses and possibly
refills the uniform state $(\randS, \randM)$.
The output sample $X_i \gets h(U_i)$ is a deterministic function of
$U_i$, where $h$ depends on the specific sampling algorithm used (e.g.,
\cref{alg:inversion,alg:lookup,alg:alias,alg:ddg}).
The map $h$ ensures that,
conditioned on $\set{X_i = x_i}$,
all elements of the preimage
$h^{-1}(x_i) \defeq \set{u \in [0, N_i) \mid f_i(u) = x_i}$
are equally likely.
We thus extract a uniform state $(\randS', \randM')$ where
$\randM' = \abs{h^{-1}(x_i)}$ is the cardinality of the preimage and
$\randS' \defeq r(U_i) \in [0, \randM')$ is obtained using a
\textit{randomness recycling rule} for $h$ that guarantees $\randS' \perp X_i \mid \randM'$.
The final step is to call \Call{Recycle}{} (\cref{alg:recycle}) to
merge $(\randS', \randM')$ into the global uniform state
$(\randS, \randM)$.

\paragraph{Organization}

\Cref{sec:random-states} introduces uniform and nonuniform random states
and describes how they can be manipulated into new random states
through information-preserving transformations.
\Cref{sec:analysis} proves \cref{theorem:entropy-cost}, by demonstrating
an online random sampling algorithm that leverages randomness recycling.
\Cref{sec:rr-uniform,sec:rr-general} investigate further applications
of randomness recycling for accelerating samplers for uniform
and general discrete distributions (\cref{table:algorithms}).
\Cref{sec:related} discusses related work.
\Cref{sec:remarks} concludes the paper with remarks and a conjecture.

\input{fig-algorithms}

%% file: fig-ors.tex
%!TEX root=./main.tex

\begin{listing}[!t]
\caption{Formal specification of an online random sampling algorithm.}
\label{lst:ors}

\begin{thmbox}
An \textit{online random sampling algorithm}
over a finite alphabet $\Alphabet$ and a countable
auxiliary state space $\AState$ with initial state $\astate_0 \in \AState$
is a partial computable function
\begin{align}
\Alg &: \set{0,1}^{\Nat} \times \DAlphabet \times \AState \rightharpoonup \Alphabet \times \AState.
\label{eq:ors-signature}
\end{align}
For each $(\dist, \astate) \in \DAlphabet \times \AState$, the set
$\set{\flips \mid (\flips, \dist, \astate) \notin \mathrm{dom}(\Alg)}$ of ``non-halting''
streams $\flips$ has Lebesgue measure zero
(i.e., $\Alg$ halts almost everywhere in its first argument).
Online sampling operates as follows.

\begin{itemize}[wide=0pt,leftmargin=*]

\item \textbf{Per-Round Map.}
Let $\dists = (\dist_i)_{i \ge 1}$ be a sequence of target distributions
and $\flips = (\flip_i)_{i \ge 1}$ a sequence of coin tosses.
For each round $i \ge 1$, the $i$th output symbol $x_i$
and updated auxiliary state $\astate_i$ at the end of the round $i$
are obtained by iterating
\begin{align}
(x_i, \astate_i) &\defeq \Alg(\flips_{>\nflip_{i-1}}, p_i, \astate_{i-1})
&& (i=1,2,\dots).
\label{eq:ors-map}
\end{align}
Here, $\nflip_{i-1} \defeq \iflip_1 + \dots + \iflip_{i-1}$ is the total
number coin tosses consumed in rounds $1$ through $i-1$
and $\flips_{>\nflip_{i-1}} \defeq (\flip_j)_{j > \nflip_{i-1}}$
is the suffix of $\flips$, which contains
only fresh tosses.
Formally, the number $\iflip \ge 0$ of coin tosses consumed on a given input
$(\flips, \dist, \astate) \in \mathrm{dom}(\Alg)$
is the smallest number such that $\Alg (\flips, \dist, \astate) = \Alg(\flips', \dist, \astate)$
for every coin toss sequence $\flips'$ whose first $\iflip$ values agree with those of $\flips$.

\item \textbf{Probabilistic Run.}
Let $\Flips \sim \mathrm{Uniform}(\set{0,1}^\Nat)$ be an i.i.d.~sequence of coin tosses
which lives on a standard probability space $(\Omega, \mathcal{E}, \Prob)$.
Let $\Astate_0 = \astate_0$ be a degenerate random variable.
The random output symbol sequence $\bX \defeq (X_i)_{i \ge 1}$
and auxiliary state sequence $\Astates \defeq (\Astate_i)_{i \ge 1}$
are obtained by iterating
\begin{align}
(X_i, \Astate_i) &\defeq \Alg(\Flips_{>\nFlip_{i-1}}, \dist_i, \Astate_{i-1})
&& (i = 1, 2, \dots).
\label{eq:ors-run}
\end{align}
where $\nFlip_{i-1} = \iFlip_0 + \dots + \iFlip_{i-1}$ is now the random
total number of coin tosses consumed so far.

\item \textbf{Correctness Requirement.}
For any fixed sequence of target distributions $\dists \defeq (\dist_i)_{i \ge 1}$,
% and for all $n \geq 1$ and all $x_1,\dots,x_n \in \Alphabet$,
the random output sequence $\bX$ from a probabilistic run must satisfy
\begin{align}
\Prob(X_1=x_1, \dots, X_n=x_n) = \prod_{i=1}^{n} \dist_i(x_i)
&&(n \ge 1;\ x_1,\dots,x_n \in \Alphabet).
\label{eq:ors-correct}
\end{align}

\end{itemize}
\end{thmbox}
\end{listing}

%% file: fig-compare.tex
%!TEX root=./main.tex

\begin{table}[t]
\centering
\caption{Comparison of online random sampling algorithms for generating
a sequence $(X_1, \dots, X_n)$ of $n \ge 1$ discrete random variables,
given conditional distributions $p_i( \dotgiven{} X_1,\dots, X_{i-1})$
over $K$ outcomes
whose probabilities have common denominator at most $d$.
The online optimal method and interval method have an entropy loss that tends
to zero as $n \to \infty$, but their worst-case space complexity is unbounded,
even for fixed $n$; furthermore, the expected space complexity grows with $n$
and is thus also unbounded as $n \to \infty$.
The online suboptimal method uses bounded space but its amortized entropy loss
is arbitrary close to $2$ bits in the worst case.
Randomness recycling achieves an entropy loss of $\varepsilon > 0$ bits
as $n \to \infty$, using bounded space of $O(\log(d/\varepsilon))$ bits.
}
\label{table:compare}
\begin{adjustbox}{max width=\linewidth}
\begin{tabular*}{\linewidth}{|l@{\extracolsep{\fill}}l|}
\hline\hline
~                         & \textbf{Online Suboptimal Method}~\citep[\cref{alg:ky-ddg}]{knuth1976}  \\\cline{2-2}
\textrm{Amortized Entropy Loss}     & $<2$  \rule{0pt}{2.5ex} \\
\textrm{Time Complexity}  & $O\left(n K \log(d) \log(K)\right)$ \\
                          % implicit traversal loop: K \log(d) \log(K)
                          % times n  \\
\textrm{Space Complexity} & $O\left(K \log(d)\right)$
                          % weight copies: K \log(d)
\\[5pt]
~                         & \textbf{Online Optimal Method}~\citep[\cref{alg:ky-dg}]{knuth1976} \rule{0pt}{2.5ex} \\ \cline{2-2}
\textrm{Amortized Entropy Loss}     & $<2/n$ \rule{0pt}{2.5ex} \\
\textrm{Time Complexity}  & $O\left(n^2 K \log(d) \log(K) + n^3 K \log^2(d) + n^3 \log(n) \log^2(d)\right)$ \\
                           % modular exponentiation: n^2 \log(n) \log^2(d)
                           % weight update: K n^2 \log^2(d)
                           % implicit traversal loop: n K \log(K) \log(d)
                           % times n
\textrm{Space Complexity} & $O\left(n K \log(d)\right)$ (Expected)
                          % auxiliary state: n \log(d)
                          % updated weight copies: n K \log(d)
\\[5px]
~                         & \textbf{Interval Method}~\citep[\cref{alg:hh-interval}]{han1997} \\ \cline{2-2}
\textrm{Amortized Entropy Loss}     & $< 3/n$  \rule{0pt}{2.5ex} \\
\textrm{Time Complexity}  & $O\left(n K \log(d) + n^2 \log^2(d) \log(K)\right)$ \\
                          % prefix sums: K \log(d)
                          % binary search (with multiplication at each step): n \log^2(d) \log(K)
                          % times n
\textrm{Space Complexity} & $O\left((n+K) \log(d)\right)$ (Expected)
                          % auxiliary state: n \log(d)
                          % prefix sums: K \log(d)
\\[5pt]
~                         & \textbf{Randomness Recycling} (\cref{table:algorithms}; \cref{alg:loop}) \\ \cline{2-2}
\textrm{Amortized Entropy Loss}     & $< \varepsilon + O(\log(d/\varepsilon))/n$ \rule{0pt}{2.5ex} \\
\textrm{Time Complexity}  & $O\left(n K \log(d) + n \log(d) \log(d/\varepsilon)\right)$ \\
                          % linear search: K \log(d)
                          % uniform and recycle: \log(d) \log(d/\varepsilon)
                          % times n \\
\textrm{Space Complexity} & $O\left(\log(d/\varepsilon)\right)$
                          % auxiliary state and partial sum variables: \log(d/\varepsilon)
                          % \rule{0pt}{3ex} \rule[-1.75ex]{0pt}{0ex}
\\ \hline \hline
\end{tabular*}
\end{adjustbox}
\end{table}

%% file: fig-overview.tex
%!TEX root=./main.tex

\begin{figure}[t]
\centering
\begin{tikzpicture}
\node[
  name=entropy,
  draw=black,
  label={above:
    \begin{tabular}{c}
    \textbf{Entropy Source} \\
    $\Flips\sim\Uniform(\set{0,1}^\Nat)$
    \end{tabular}},
  minimum height=.5cm,
  minimum width=4cm,
  pattern=crosshatch dots,
  ]{};

\node[
  name=uniform-state,
  below=1.5 of entropy,
  draw=black,
  label={[name=uniform-state-label]below:
    \begin{tabular}{c}
    \textbf{Uniform State} $(\randS,\randM)$ \\
    $\randS \mid \randM \sim \Uniform[0,\randM)$
    \end{tabular}},
  ]{$
    \begin{NiceArray}[first-row,hvlines,corners]{ccccc}
    \CodeBefore
      % \EmptyRow{2}
    \Body
    \Block{1-5}{\randM \in [1,2^W)} \\
    1 & 2 & \Block[transparent, tikz={pattern = north west lines, pattern color = gray!50!white}]{1-1}{3} & \dots & 2^W-1 \Tstrut \\
    \Block[draw=white,borders={top,bottom}]{1-5}{\randS \in [0,\randM)} \rule{0pt}{4ex} \\
    0 & 0 & 0 & \Block{3-1}{\ddots} & 0 \\
      & 1 & \Block[transparent, tikz={pattern = north west lines, pattern color = gray!50!white}]{1-1}{1} & ~ & 1 \\
      &   & 2 &   & 2 \\
      &   &   & \ddots & \vdots \\
      &   &   &   & 2^{W}-2 \Tstrut \\
    \end{NiceArray}
  $};

\draw[-latex,thick] (entropy.south) -- node[pos=0.5,label={[align=center]left:\Call{Refill}{}\\(\cref{alg:refill})}]{} (uniform-state.north);

\node[
  name=uniform-sample,
  inner sep=0pt,
  draw=black,
  right=3 of uniform-state,
  label={[name=uniform-sample-label]below:
    \begin{tabular}{c}
    \textbf{Uniform Sample} \\
    $U_i \sim \Uniform[0,N_i)$
    \end{tabular}},
  ]{$
    \begin{NiceArray}[hvlines]{c}
    \Block[transparent, fill=green!50!white]{1-1}{0} \\
    1 \\
    \Block[transparent, fill=green!50!white]{1-1}{2} \\
    \vdots \\
    U_i-1 \\
    \Block[transparent, fill=green!50!white]{1-1}{U_i} \\
    \Block[transparent, fill=green!50!white]{1-1}{U_i+1} \\
    \vdots \\
    N_i-1
    \end{NiceArray}
  $};

\node[
  name=finv,
  right=0.25 of uniform-sample.east,yshift=-1cm,
  minimum width=1cm, minimum height=.5cm,fill=green!50!white]{$h^{-1}(X_i)$};

\draw[-latex,thick]
  (uniform-state.east)
  to
  node[
    pos=0.5,
    label={[align=center]above:\Call{Uniform}{$N_i$}\\[2pt](\cref{alg:uniform}\textsuperscript{\textdagger})}]{}
  (uniform-sample.west);

\node[
  name=nonuniform-sample,
  inner sep=0pt,
  draw=black,
  right=3.5 of uniform-sample,
  label={[name=nonuniform-sample-label]below:
    \begin{tabular}{c}
    \textbf{Nonuniform Sample} \\
    $X_i \sim P_i$
    \end{tabular}},
  ]{$
    \begin{NiceArray}[hvlines]{c}
    0      \\
    1      \\
    \vdots \\
    \Block[transparent, fill=gray!25!white]{1-1}{h(U_i)\eqdef X_i} \\
    \vdots
    % N_i-1
    \end{NiceArray}
  $};

\draw[-latex,thick]
  (uniform-sample)
  to
  node[
    pos=0.5,
    % label={[align=center]above:General Samplers\\(\cref{alg:inversion,alg:lookup,alg:alias,alg:ddg})}]{}
    label={[align=center]above:\Call{Inversion}{$\Dist_i$}\\[2pt](\cref{alg:inversion}\textsuperscript{\textdaggerdbl})}]{}
  (nonuniform-sample);

\coordinate[name=snuffle,at=($(entropy.south)!0.5!(uniform-state.north)$)];
\coordinate[name=buffle,at=(snuffle -| nonuniform-sample.north)];
\draw[latex-,thick,rounded corners=4mm]
  ([xshift=1cm]uniform-state.north)
  |-
  (buffle)
  node[
    pos=0.75,
    label={[align=center]above:
      \Call{Recycle}{$\randS',\randM'$} \\
      $\randS'\defeq r_i(U_i)$ \\
      $\randM'\defeq \abs{h^{-1}(X_i)}$\\
      (\cref{alg:recycle})
      }
    ]{}
  |-
  (nonuniform-sample.north);

\node[rectangle, anchor=north east, at=(current bounding box.north east),xshift=-1cm] {$i =1,2,\dots,$};

\node[font=\scriptsize,align=left,at=(uniform-sample-label -| nonuniform-sample-label)]
  {
  \textsuperscript\textdagger~or \cref{alg:uniform-widening,alg:uniform-lemire,alg:uniform-brackett} \\
  \textsuperscript\textdaggerdbl~or \cref{alg:lookup,alg:alias,alg:ddg}
  };

\end{tikzpicture}

\caption{Online random sampling using randomness recycling.
The sampling algorithm is dynamically given a random sequence
$\Dists \defeq (\Dist_i)_{i\ge 1}$ of probability distributions
and access to i.i.d.~coin tosses $\Flips \defeq (\Flip_i)_{i \ge 1}$,
and generates output sequence $\bX = (X_i)_{i\ge1}$
such that $X_i \sim \Dist_i$ for $i \ge 1$.}
\label{fig:overview}
\end{figure}

%% file: fig-algorithms.tex
%!TEX root=./main.tex

\begin{table}[t]
\centering
\caption{Randomness recycling strategies developed in this work for uniform and general sampling algorithms.}
\label{table:algorithms}
\begin{tabular}{|ll|}\hline
\textbf{Algorithm} & \textbf{Reference} \\ \hline
\rowcolor{gray!20!white}
{\textit{Uniform Distributions}} & ~ \\
\;Division Method & \cref{alg:uniform} \\
\;Machine Word Division Method & \cref{alg:uniform-widening} \\
\;\Citet{lemire2019} Method & \cref{alg:uniform-lemire} \\
\;\Citet{brackett2025} Method & \cref{alg:uniform-brackett} \\
\rowcolor{gray!20!white}
{\textit{General Distributions}} & ~ \\
\;Inversion Sampling & \cref{alg:inversion} \\
% \;Bernoulli Sampling & \cref{alg:bernoulli} \\
\;Lookup-Table Sampling \citep{devroye1986} & \cref{alg:lookup} \\
\;Alias Sampling \citep{walker1977} & \cref{alg:alias} \\
\;DDG Sampling \citep{knuth1976,saad2020fldr} & \cref{alg:ddg} \\ \hline\hline
\end{tabular}
\end{table}

%% file: sec-random-states.tex
%!TEX root=./main.tex

\section{Random States}
\label{sec:random-states}

We first describe the auxiliary-state space
$\AState = \set{(z,m) \mid m \ge 1, 0 \le z < m}$ used by our online
sampling algorithm, as well as information-preserving operations on
the elements of this space.

\begin{definition}
A \textit{uniform random state} $(\randS,\randM)$ is any pair of discrete
random variables such that $M \ge 1$ and $\randS \mid \randM \sim \Uniform[0,\randM)$.
\end{definition}

\begin{definition}
A \textit{nonuniform random state} $(\randS,\mathbf{W})$
is a pair of discrete random variables such
that $\mathbf{W} \defeq (W_0,\dots,W_{n-1})$ is a list of $n \ge 1$
positive random variables
and
$\randS \mid \mathbf{W} \sim \Discrete(\mathbf{W})$,
i.e., $\Prob(\randS=i \mid \mathbf{W}) = W_i/\sum_{j=0}^{n-1}{W_j}$ for $i \in [0,n)$.
\end{definition}

Randomness recycling leverages several properties about merging and
splitting both uniform (\cref{sec:random-states-uniform})
and nonuniform \cref{sec:random-states-nonuniform}
random states.
Merging and splitting are information-preserving transformations
(bijections) of random variables that perfectly invert one another.
The general idea is that is if $X$ is a discrete random element
with probability mass function $p_X$
and $f$ is a bijection, the transformation $Y = f(X)$ preserves Shannon
information content (surprisals) in the sense that for all $(x,y)$ such that
$y = f(x)$, we have
\begin{align}
p_Y(y) = p_X(x) \implies
\log\left(\frac{1}{p_Y(y)}\right) = \log\left(\frac{1}{p_X(x)}\right).
\label{eq:information-content-preserved}
\end{align}

\subsection{Merging and Splitting Uniform Random States}
\label{sec:random-states-uniform}

\subsubsection{Merging Two Uniform States into a Uniform State}

We first describe how to \textit{merge} two uniform states $(\randS,\randM)$ and
$(\randS',\randM')$ to obtain a new uniform state
over a larger range $\randM \randM'$.

\begin{proposition}
\label{prop:merge-uniform}
For any integers $m, m' \ge 1$,
\begin{align*}
\randS \sim \mathrm{Uniform}[0,m), \quad
\randS' \sim \mathrm{Uniform}[0,m'), \quad
\randS \perp \randS'
\implies
\randS + \randS'm \sim \mathrm{Uniform}[0,mm').
% &\qedhere
\end{align*}
\end{proposition}

\begin{proof}
The map
$[0,m) \times [0,m') \ni (x,x') \mapsto x+x'm \in \set{0,1,\dots,mm'-1}$
is a bijection, so each outcome is equally likely with probability $1/(mm')$.
\end{proof}

\input{alg-recycle-refill}

Our randomness recycling algorithms will propagate a uniform state
$\Astate = (\randS,\randM)$
in a way that maintains the invariant~\labelcref{invariant}.
To ensure that the space complexity remains bounded,
we fix an upper bound $W \ge 1$ on the size of any uniform state,
so that $\Prob(\randM \in [1, 2^W)) = 1$ and $\Prob(\randS \in [0, 2^W-1)) = 1$.
The variables $(\randS, \randM)$ are each stored as one $W$-bit word in
memory.
Rather than explicitly accept and return a state in the functional
style~\cref{eq:ors-map}, we instead manipulate a global random state
using in-place updates.

The \Call{Recycle}{} procedure shown in \cref{alg:recycle}
uses \cref{prop:merge-uniform}
to recycle a new uniform state $(\randS',\randM')$
into the global uniform state $(\randS,\randM)$.
The notation $\randS \setp \dots$ and $\randM \setp \dots$
on \cref{algline:recycle-update-Z,algline:recycle-update-m}
means that $(\randS,\randM)$ is being updated in place.
The global uniform state is initialized
as $\astate_0 = (\randS,\randM) = (0,1)$ in
\cref{algline:recycle-init-Z,algline:recycle-init-m}, which is the only
valid deterministic choice.

It will also be necessary
to recycle a \textit{fresh} uniform state $(\randS',\randM')$ that
is drawn directly from the entropy source
into $(\randS,\randM)$.
We call this operation $\Call{Refill}{}$, as shown in \cref{alg:refill}.
The interface to the entropy stream $\Flips = (\Flip_i)_{i \ge 1}$
is by denoted $\Call{Flip}{k}$,
which returns off the next $k$ i.i.d.~fair coin tosses subject
to the fresh coin guarantee~\cref{eq:ors-fresh}.
The number $k$ of tosses computed on \cref{algline:refill-fresh-k}
ensures that the updated uniform state $(\randS,\randM)$ does not
overflow beyond $W$-bit words.
The coin toss counters $\iflip_i$ and $\nflip_i$ from \cref{lst:ors} are obtained by
keeping track of the calls to $\Call{Flip}{k}$ on \cref{algline:refill-fresh-k}
across the rounds of online sampling.

\subsubsection{Splitting a Uniform State into Two Uniform States}
The merge operation \cref{prop:merge-uniform} preserves information
and has the following inverse.

\begin{proposition}
\label{prop:split-uniform}
For any integers $m, m' \ge 1$,
\begin{equation*}
\randS \sim \mathrm{Uniform}[0,mm')
\implies
\left\lbrace\begin{aligned}
\floor{\randS / m} &\sim \mathrm{Uniform}[0,m') \\
(\randS \bmod m) &\sim \mathrm{Uniform}[0,m) \\
\floor{\randS / m} &\perp (\randS \bmod m).
\end{aligned}\right.
\end{equation*}
\end{proposition}

\begin{proof}
Immediate from the arguments in \cref{sec:rr-uniform}.
\end{proof}
\Cref{prop:split-uniform} shows how to split a uniform random state (over a composite
integer range with known factorization) into two smaller
independent uniform random states, using a single integer division.

\subsection{Merging and Splitting Nonuniform Random States}
\label{sec:random-states-nonuniform}

We next describe generalizations of
\cref{prop:merge-uniform,prop:split-uniform} for nonuinform random states,
which will be used by our recycling algorithms.

\subsubsection{Merging a Nonuniform and Uniform State into a Uniform State}

\begin{proposition}
\label{prop:merge-nonuniform}
For any positive integers $\mathbf{w} \defeq w_0, \dots, w_{n-1}$
with sum $m$,
\begin{equation*}
\randS \sim \mathrm{Discrete}(\mathbf{w}), \quad
(\randS' \mid \randS) \sim \mathrm{Uniform}[0, w_\randS)
\implies \left(\randS' + \sum_{i=0}^{\randS-1}w_i\right) \sim \mathrm{Uniform}[0,m).
% \qedhere
\end{equation*}
\end{proposition}

\begin{proof}
The map
$\cup_{i=0}^{n-1}(\set{i}\times[0,w_i)) \ni (x,x') \mapsto x' + \sum_{i=0}^{x-1}w_i \in [0,m)$
is a bijection, so each outcome is equally likely with probability
$(w_i / m) (1 / w_i) = 1/m$.
\end{proof}

\begin{remark}
\Cref{prop:merge-uniform} is a special case of \cref{prop:merge-nonuniform}
with $\mathbf{w} = (m', m', \dots, m')$.
\end{remark}

\subsubsection{Splitting a Uniform State into a Nonuniform and Uniform State}
The inverse of \cref{prop:merge-nonuniform} is as follows.

\begin{proposition}
\label{prop:split-nonuniform}
For any positive integers $\mathbf{w} \defeq w_0, \dots, w_{n-1}$ with sum $m$,
\setlength{\belowdisplayskip}{0pt}
\begin{equation*}
\randS \sim \mathrm{Uniform}[0,m), \quad
\randS \equiv \sum_{i=0}^{Y-1}w_i + X\ (\mbox{where } X \in [0,w_Y))
\implies
\left\lbrace\begin{aligned}
Y &\sim \mathrm{Discrete}(\mathbf{w}) \\
(X \mid Y) &\sim \mathrm{Uniform}[0, w_Y).
\end{aligned}\right.
\end{equation*}
\end{proposition}

\begin{proof}
Let $W_i \defeq w_0 + \dots + w_i$ denote the prefix sums of $\mathbf{w}$
for $i \in [-1,n)$.
The event $\set{Y=i}$ holds if and only if $Z \in [W_{i-1}, W_i)$
which has probability $(W_i-W_{i-1})/m = w_i/m$.
Since $X = Z - W_{Y-1}$, conditioning on the event $\set{Y=i}$ gives
$Z \mid \set{Y=i} \in [W_{i-1}, W_i)$ and $X \mid \set{Y=i} \in [0, W_{i}-W_{i-1} \equiv w_i)$
which are both uniformly distributed by the uniformity of $Z$.
\end{proof}

\Cref{prop:split-nonuniform} says that a uniform random state can be split
into \begin{enumerate*}[label=(\roman*)]
\item a nonuniform random state following \textit{any} given finite discrete rational distribution; and
\item a leftover uniform random state;
\end{enumerate*}
as long as the original uniform range
matches the common denominator of the target discrete distribution.
It confirms the merge operation in \cref{prop:merge-nonuniform} is
reversible and justifies the recycling rule for the inversion method
presented in the next section.

%% file: alg-recycle-refill.tex
%!TEX root=./main.tex
\begin{listing}[t]
\begin{minipage}{\linewidth}
\input{alg-recycle}
\vspace{-.75cm}
\end{minipage}
\begin{minipage}{\linewidth}
\input{alg-refill}
\end{minipage}%
\end{listing}

%% file: alg-recycle.tex
%!TEX root=./main.tex
\begin{algorithm}[H]
\caption{Recycling a uniform state into the global uniform state}
\label{alg:recycle}
\begin{algorithmic}[1]
\LComment{Global variables: word size $W$; uniform state $S = (\randS, \randM)$}
\State \textbf{int} $W \gets W_{d,\varepsilon}$ \Comment{cf.~\cref{eq:W-d-epsilon}}
\State \textbf{mutable int} $\randS \gets 0$ \Comment{$W$-bit integer} \label{algline:recycle-init-Z}
\State \textbf{mutable int} $\randM \gets 1$ \Comment{$W$-bit integer} \label{algline:recycle-init-m}
\Require{Uniform state $(\randS',\randM')$ such that $\randS' \mid \randM' \sim \Uniform[0,\randM')$
and $\randS' \perp \randS \mid \randM, \randM'$.}
\Ensure{Update the global uniform state $(\randS, \randM)$ such that $\randS \mid \randM,\randM' \sim \Uniform[0,\randM \randM')$ (cf.~\cref{prop:merge-uniform}).}
\Procedure{Recycle}{$\randS',\randM'$}
\State $\textbf{update}~\randS \setp \randS + \randS' \randM$ \label{algline:recycle-update-Z}
\State $\textbf{update}~\randM \setp \randM \randM'$          \label{algline:recycle-update-m}
\EndProcedure
\end{algorithmic}
\end{algorithm}

%% file: alg-refill.tex
%!TEX root=./main.tex
\begin{algorithm}[H]
\caption{Refilling the global uniform state with fresh uniform bits from the entropy source}
\label{alg:refill}
\begin{algorithmic}[1]
\Require{Read access to global variables $W$ and $\randM$ from \cref{alg:recycle}.}
\Ensure{
  Update the global uniform state $(\randS, \randM)$ so that
  $\randM \in [2^{W-1}, 2^{W})$ and
  $\randS \mid \randM \sim \Uniform[0,\randM)$.
}
\Procedure{Refill}{{}}
\State $k \gets W - \ceil{\log_{2}(\randM+1)}$
  \Comment{number of fresh bits to append to $Z$}
  \label{algline:refill-fresh-k}
\State $(\flip_0, \dots, \flip_{k-1}) \gets \Call{Flip}{k}$
  \Comment{obtain $k$ i.i.d.~bits from entropy source}
  \label{algline:refill-fresh-flips}
\State $\randS' \gets \flip_{k-1} \cdot 2^{k-1} + \dots + \flip_1 \cdot 2^1 + \flip_0 \cdot 2^0$
  \Comment{compute random $k$-bit integer}
  \label{algline:refill-fresh-state}
\State $\randM' \gets 2^{k}$
  \Comment{upper bound on the value of the random integer}
  \label{algline:refill-fresh-bound}
\State \Call{Recycle}{$\randS'$,$\randM'$}
  \Comment{merge fresh uniform state into global state (\cref{alg:recycle})}
\EndProcedure
\end{algorithmic}
\end{algorithm}

%% file: sec-analysis.tex
%!TEX root=./main.tex

\section{Proof of Main Theorem}
\label{sec:analysis}

\subsection{A Recycler for Uniform Distributions}

We now use the ideas from the previous section to develop a
randomness recycler for sampling discrete uniforms.
Suppose our goal is to generate
$X \sim \Uniform[0,n)$ over a range $n \leq \randM$,
where $(\randS, \randM)$ denotes the global uniform state maintained
in \cref{alg:recycle}.
Applying \cref{prop:split-nonuniform} with the weights
$\mathbf{w} = (\randM-n, n)$ furnishes
a Bernoulli sample $Y \sim \Discrete(\randM-n,n)$
and a uniform sample $X \sim \Uniform[0,w_Y)$.
A straightforward accept-reject approach is as follows.
\begin{enumerate}[wide,label=(R\arabic*),leftmargin=*]

\item\label{step:unif-rej-r-1} If $Y=1$, then $X \sim \Uniform[0,n)$ is
  accepted and returned as the desired sample.

\item\label{step:unif-rej-r-2} If $Y=0$, then $X \sim \Uniform[0,\randM-n)$, which is not
  over the desired range, is rejected and $(X, \randM-n)$
  is used as the new global uniform state.
\end{enumerate}

\paragraph{Information Loss}
In the accept and reject cases, we must discard the Bernoulli outcome
$Y$ and waste the corresponding $H_{\rm b}(n/\randM)$ bits of information
because it is in general correlated with the output sample $X$,
which violates the contract \labelcref{contract}.
More importantly, we cannot recycle $Y$ back into the
global uniform state using \cref{prop:merge-nonuniform}
unless we recycle it in both the accept and reject cases.
But recycling in the accept case requires consuming the desired
$X \sim \Uniform[0, n)$, which effectively undoes the sampling operation.
Discarding a Bernoulli outcome $Y$ in this way is the only operation
where we lose information and is fundamentally irreversible.

\paragraph{Reducing Rejection Probability}
If $\randM \approx n$, then
\labelcrefrange{step:unif-rej-r-1}{step:unif-rej-r-2}
gives a reasonably time- and entropy-efficient
sampler for $X \sim \mathrm{Uniform}[0,n)$
(cf.~\cref{remark:uniform-power-of-two}).
On the other hand, if $\randM \gg n$ then the rejection probability
$(\randM-n)/\randM$ is very high.
To decrease the rejection probability, we can multiply the accept weight
and use \cref{prop:split-uniform} in the accept case to split the uniform state
into two smaller uniform states.
The resulting sampler operates as follows:
\begin{enumerate}[wide,label=(S\arabic*),leftmargin=*]
\item\label{step:unif-rej-s-1}
  Generate
  $\hat{X} \sim \Uniform[0,\floor{\randM/n} n)$ using
  \labelcrefrange{step:unif-rej-r-1}{step:unif-rej-r-2}.

\item\label{step:unif-rej-s-2}
  Apply \cref{prop:split-uniform}
  (with $Z = \hat{X}$, $m=n$, $m'=\floor{\randM/n}$)
  to obtain $X \defeq (\hat{X} \bmod n) \sim \Uniform[0,n)$
  as the desired sample.
  Recycle the leftover
  uniform state $\floor{\hat{X}/n} \sim \Uniform[0,\floor{\randM/n})$
  into the global uniform state using \cref{alg:recycle}.
\end{enumerate}
The correctness of \labelcref{step:unif-rej-s-1} follows from the fact
$\floor{\randM/n} n \le \randM$, so
\labelcrefrange{step:unif-rej-r-1}{step:unif-rej-r-2}
can be applied.

\paragraph{Implementation}

We now give a concrete algorithm that
fuses \labelcrefrange{step:unif-rej-s-1}{step:unif-rej-s-2} together
to arrive at a simple and efficient algorithm for uniform sampling.
The goal is to generate a discrete uniform
$U \sim \Uniform[0,n)$ for some $n \in [1,\randM]$,
using the global uniform state $(\randS,\randM)$.
Consider the quotients and remainders
of $\randS = q_\randS n + r_\randS$ and $\randM = q_\randM n + r_\randM$
modulo $n$,
whose possible values as $\randS$ ranges over $[0,\randM)$ are shown below.
\begin{align*}
\NiceMatrixOptions{custom-line={letter = I, tikz = {line width=2pt, blue }}}
\begin{adjustbox}{max width=\linewidth}
$\begin{NiceArray}[first-col, first-row, hvlines]{ccccIccccIcIccccIcccc}
\CodeBefore
  \rowcolors{2}{yellow}{yellow}
  \cellcolor{red!30!white}{2-14,2-15,2-16,2-17}
  \cellcolor{green!30!white}{3-1,3-2,3-3,3-4}
  \cellcolor{blue!30!white}{3-5,3-6,3-7,3-8}
  \cellcolor{cyan!30!white}{3-10,3-11,3-12,3-13}
  \cellcolor{DarkOliveGreen!30!white}{3-14,3-15,3-16,3-17}
\Body
~
\\
\randS &
0 & 1 & \dots & n-1
&
n & n+1 & \dots & 2n-1
&
\dots
&
\Block{1-4}{\dots} & ~ & ~ & ~
&
q_\randM n & q_\randM n+1 & \dots & \randM -1
\\
q_\randS &
\Block{1-4}{0} & & &
&
\Block{1-4}{1} & & &
&
\cellcolor{white} \dots
&
\Block{1-4}{q_\randM -1} & & &
&
\Block{1-4}{q_\randM \defeq \floor{\randM/n}} & & &
\\
r_\randS &
0 & 1 & \dots & n-1
&
0 & 1 & \dots & n-1
&
\cellcolor{white} \dots
&
0 & 1 & \dots & n-1
&
0 & 1 & \dots & r_\randM -1
\CodeAfter
  \tikz \draw[latex-latex] ([yshift=.25cm] 1-|1) -- ([xshift=-.05cm,yshift=.25cm] 1-|14) node[fill=white,pos=0.5,inner ysep=0pt,font=\bfseries]{Accept};
  \tikz \draw[latex-latex] ([xshift=.05cm,yshift=.25cm] 1-|14) -- ([yshift=.25cm] 1-|18) node[fill=white,pos=0.5,inner ysep=0pt,font=\bfseries]{Reject};
\end{NiceArray}$
\end{adjustbox}
\end{align*}
The first row shows the possible values of $\randS \in [0,\randM)$.
The second and third
rows show the corresponding random quotient $q_\randS = \floor{\randS/n}$
and remainder $r_\randS = \randS \bmod {n}$, respectively.
\Cref{alg:uniform} uses these properties as follows:

\begin{itemize}[wide,leftmargin=*]
\item Consider the event $\set{q_Z < q_\randM}$.
Here, $\randS$ is equally likely to be any one of the elements whose
quotient is less than $q_\randM$.
Further, among all elements with quotient $q_\randS$,
$\randS$ is equally likely to have any remainder
$r_\randS \in \set{0,\dots,n-1}$.
It follows that $q_\randS \mid \set{q_\randS < q_\randM} \sim \Uniform[0,q_\randM)$
and $r_\randS \mid \set{q_\randS < q_\randM} \sim \Uniform[0,n)$.

\item Consider the event $\set{q_\randS = q_\randM}$. In this case, $\randS$ falls
in the final (pink) segment above.
It follows that the remainder $r_\randS \mid \set{q_\randS = q_\randM} \sim \Uniform[0,r_\randM)$,
where $r_\randS = r_\randM$ is impossible because $\randS < \randM$ surely.
\end{itemize}

\input{alg-uniform}

\begin{remark}
If $n$ divides $\randM$, then $\Prob(q_\randS = q_\randM) = 0$,
i.e., $\Call{Uniform}{n}$ never rejects.
The if-else statement in \cref{alg:uniform} can be seen as
branching into either the accept  case $\randS \sim \Uniform[0, q_\randM n)$,
where $\randS \in [0,q_\randM n)$ ensures termination,
or the rejection case $\randS - q_\randM n \in [0,r_\randM)$, which requires
drawing more coin tosses from the entropy source.
\end{remark}

\begin{remark}
\label{remark:uniform-refill}
In \cref{algline:uniform-refill}, \Call{Refill}{} ensures that
the upper bound $\randM$ of the uniform state lies in the range $[2^{W-1}, 2^W)$.
While refilling the uniform state is not necessary when $n \leq \randM$,
the algorithm is most entropy efficient when the rejection probability
$r_\randM / \randM \le (n-1)/\randM$ is near zero.
To ensure efficient amortized entropy usage, the call to \Call{Refill}{}
in \cref{alg:uniform} minimizes the
rejection probability $r_\randM / \randM$ by enlarging $\randM$
to the range $[2^{W-1},2^W)$, where $2^{W-1}$ is much larger than any
integer $n$ we typically intend to sample from.
For example, if we only intend to sample from uniforms over $n \in (1, 2^{32})$,
then setting $W = 64$ bits ensures a vanishingly small rejection
probability of less than $2^{-31} < 5\times 10^{-10}$.
Each time an accept/reject decision is made
(i.e., the if-else branch in \cref{alg:uniform}), exactly
$H_{\rm b}(r_\randM / \randM)$ bits of entropy are wasted.
This accept/reject information must be discarded to maintain
the independence invariant \labelcref{invariant}.
The fresh coin tosses created in the call to \Call{Refill}{} are recovered
to the maximum extent possible in
\cref{algline:uniform-recycle-accept,algline:uniform-recycle-reject} of
\cref{alg:uniform} while maintaining the invariants.
\end{remark}

\subsection{A Recycler for General Distributions}

The uniform recycler in the previous section can be used
to sample from any rational discrete distribution.
Consider first generating
$X \sim \Bernoulli(a/b)$ with rational weight $a/b$.
We draw $U \sim \Uniform[0,b)$ and set
$X \gets 1$ if $U < a$, and $X \gets 0$ if $a \le U$.
The randomness recycling rule is based on extracting
the following uniform states:
\begin{itemize}[wide,leftmargin=*]
\item If  $\set{U < a}$, then $U \mid \set{U<a} \sim \Uniform[0,a)$.

\item If  $\set{a \le U}$, then $U \mid \set{a \le U} \sim \Uniform[a,b)$.
\end{itemize}
\begin{align*}
\NiceMatrixOptions{custom-line={letter = I, tikz = {line width=2pt, blue }}}
\begin{NiceArray}[first-col,hlines,vlines]{ccccIcccc}
U & \cA[0] & \cA[1] & \dots & \cA[a-1]
&
\cD[a] & \cD[a+1] & \dots & \cD[b-1]
\CodeAfter
  \UnderBrace[shorten,yshift=4pt]{1-1}{1-4}{X=1}
  \UnderBrace[shorten,yshift=4pt]{1-5}{1-8}{X=0}
\end{NiceArray}\,.
\\[5pt]
\end{align*}

This method admits a direct generalization to sampling
from $\dist \defeq (a_0,\dots,a_{n-1})/A$, where
$A \defeq a_0 + \dots + a_{n-1} > 0$.
Let $F \defeq (A_0, A_1, \dots, A_{n-1}\equiv A)$ denote the cumulative
probabilities, i.e.,
the prefix sums of the integers that define $\dist$, which
are either constructed from the integers $(a_0, \dots, a_{n-1})$
or given as a function $i \mapsto F_i$.
A uniform variate $U \sim \Uniform[0,A)$ is first generated and
then $X \gets \min\set{ i \in [0,n) \mid U < A_i}$
is returned as the sample from $\dist$.
In this setup, the event $\set{X = i}$ is equivalent to the event
$U \in [A_{i-1}, A_i)$, where $A_{-1} \defeq 0$.
The uniformity of $U$ ensures that
$U \mid \set{X=i} \sim \Uniform[A_{i-1},A_i)$, which
gives a uniform random state to recycle.

\begin{align*}
\NiceMatrixOptions{custom-line={letter = I, tikz = {line width=2pt, blue }}}
\begin{NiceArray}[first-col,hvlines]{ccccIccccIcIcccc}
U
&
\cA[0] & \cA[1] & \cA[\dots] & \cA[A_0-1]
&
\cB[A_0] & \cB[A_0+1] & \cB[\dots] & \cB[A_1-1]
&
\cC[\dots]
&
\cD[A_{n-2}] & \cD[A_{n-2}+1] & \cD[\dots] & \cD[A_{n-1}-1]
\\
F
&
\Block{1-4}{A_0} & ~ & ~ & ~
&
\Block{1-4}{A_1} & ~ & ~ & ~
&
\dots
&
\Block{1-4}{A_{n-1}} & ~ & ~ &
\CodeAfter
  \OverBrace[shorten,yshift=4pt]{1-1}{1-4}{a_0}
  \OverBrace[shorten,yshift=4pt]{1-5}{1-8}{a_1}
  \OverBrace[shorten,yshift=4pt]{1-10}{1-13}{a_{n-1}}
  \UnderBrace[shorten,yshift=4pt]{1-1}{2-4}{X=0}
  \UnderBrace[shorten,yshift=4pt]{1-5}{2-8}{X=1}
  \UnderBrace[shorten,yshift=4pt]{1-10}{2-13}{X=n-1}
\end{NiceArray}
\\[5pt]
\end{align*}
% [3, 2 ,3]
% [3, 5, 8]
% 0, 1, 2 -> 0
% 3, 4    -> 1
% 5, 6, 7 -> 2
%
The top array (which has $A_{n-1}$ elements) shows all possible values
of $U\sim \Uniform[0,A_{n-1})$.
The bottom array (which has $n$ elements) shows the cumulative probabilities
$F$ that determine $X$ from $U$.
\Cref{alg:inversion} shows a generic inversion sampler, which
recycles randomness from the leftover uniform random state in
\cref{prop:split-nonuniform}.
The search on \cref{algline:inversion-search} can be performed in
several different ways, depending on the usage pattern of the sampler.
\begin{itemize}[wide,leftmargin=*]
\item Using binary search on the cumulative probabilities
when they are either given as input or computed as a preprocessing step.
This method is ideal for generating multiple samples from the same
distribution, since its space is linear and search time is logarithmic.
\item Using a linear scan over the array of (cumulative) probabilities.
This method is ideal when just a single sample from the distribution is
required, as it does not require preprocessing.
\item Using a lookup table (\cref{sec:rr-general-lookup}).
This method is ideal when $A$ is small, and it delivers constant search time.
\end{itemize}

\input{alg-inversion}

\subsection{Entropy Cost of Uniform Sampling}
\label{sec:analysis-uniform}

We first analyze the entropy loss that occurs in a single invocation of
$\Call{Uniform}{}$ (\cref{alg:uniform}).
For any invocation of $\Call{Uniform}{n}$, the call to $\Call{Refill}{}$
on \cref{algline:uniform-refill}
ensures that the global uniform state $(\randS, \randM)$ satisfies
$\randM \geq \randMin \defeq 2^{W-1}$.
The rejection probability is thus $(\randM \bmod n) / \randM \leq (n-1)/\randM \leq (n-1)/\randMin$.
Recall from \cref{sec:rr-uniform} that entropy is only ever lost when
discarding the information of a $\Bernoulli((\randM \bmod n) / \randM)$
coin toss in the accept/reject decision.
The total entropy loss from a given invocation is then the expected number
of trials times the expected entropy of the discarded Bernoulli.

We make this intuition precise.
Consider a generic invocation of $U \sim \Call{Uniform}{n}$, given a
random state $(\randS^{-}, \randM^{-})$ such that
$\randS^{-} \mid \randM^{-} \sim \Uniform[0,\randM^{-})$,
where $n \in [1, \randMin]$.
Let $(\randS^+, \randM^+)$ be the final random state at termination of the
call.
Let $\iFlip$ denote the (random) number of coin tosses
requested from the entropy source, across all invocations of $\Call{Flip}{}$
(through \cref{algline:refill-fresh-flips} of
\hyperref[alg:refill]{\Call{Refill}{}}) made by the call $\Call{Uniform}{n}$.
With these notations, the information loss given initial state
$(\randS^{-}, \randM^{-})$ is
\begin{align}
(\log(\randM^{-}) + \iFlip) - (\log(n) + \log(\randM^{+}))
\label{eq:entropy-loss-uniform}
\end{align}
where $\log(\randM^{-}) + \iFlip$ is the information content consumed
and $\log(n) + \log(\randM^{+})$ is the information content produced
(all input and output variables in \cref{alg:uniform} are uniform).

\begin{proposition}
\label{prop:entropy-loss-uniform}
The entropy loss of any invocation
$\Call{Uniform}{n}$ of \cref{alg:uniform} satisfies
\begin{equation}
% 0
% \leq
\expect{(\log(\randM^{-}) + \iFlip) - (\log(n) + \log(\randM^{+})) \mid \randM^{-}, U}
\leq
\frac{\randMin}{\randMin - n + 1} H_{\rm b}\left( \frac{n-1}{\randMin} \right),
\label{eq:entropy-loss-uniform-bound}
\end{equation}
where the expectation is taken over
the random coin tosses from the entropy source $\Call{Flip}{}$.
\end{proposition}

\begin{proof}
\newcommand{\randMi}[1][j]{\randM^-_{#1}}
\newcommand{\randMm}[1][j]{\randM^\star_{#1}}
\newcommand{\randMo}[1][j]{\randM^+_{#1}}
\newcommand{\randSi}[1][j]{\randS^-_{#1}}
\newcommand{\randSm}[1][j]{\randS^\star_{#1}}
\newcommand{\randSo}[1][j]{\randS^+_{#1}}
This entropy loss bound is justified informally
in \cref{sec:rr-uniform}.
Here, we prove the bound using a formal analysis of
\cref{alg:uniform}.
Let $J \ge 1$ denote the random number of loops in the invocation
$\Call{Uniform}{n}$ until a sample is returned.
In each loop iteration $j$,
the global uniform state is manipulated in two places:
the call to $\Call{Refill}{}$ on \cref{algline:uniform-refill}
at the start of the loop; and
the recycling on
\cref{algline:uniform-recycle-accept} or \cref{algline:uniform-recycle-reject}
at the end the of the loop.
To track the state, we define for $j \in [1,J]$
\begin{itemize}[wide,noitemsep]
\item $(\randSi, \randMi)$ is the global uniform state at the start of iteration $j$;
\item $(\randSm, \randMm)$ is the global uniform state after the call to $\Call{Refill}{}$;
\item $(\randSo, \randMo)$ is the global uniform state at the end of iteration $j$.
\end{itemize}
These quantities satisfy
\begin{align}
\randMi[1] = \randM^{-}
&&
\randMi[j+1] = \randMo[j]\,\, (j \in [1,J-1])
&&
\randMo[J] = \randM^{+}.
\label{eq:entropy-loop-properties}
\end{align}
Let $\iFlip_j$ for $j \in [1,J]$ denote the number of fresh coin tosses
at iteration $j$ of $\Call{Uniform}{}$, through
\cref{algline:refill-fresh-k} of $\Call{Refill}{}$.
\Cref{prop:merge-uniform} implies the information is preserved:
\begin{align}
\log(\randMm) = \log(\randMi) + \iFlip_j && (j \in [1,J]).
\label{eq:entropy-mid-in-toss}
\end{align}

The accept-reject step at iteration $j$ is
\begin{enumerate}[wide=2\parindent,widest={(Accept)},leftmargin=*]
\item[(Accept)] If $\randSm < n \floor{\randMm/n}$ then
$\randMo = \floor{\randMm/n}$; and the loop terminates
and a sample from $\Uniform[0,n)$ is successfully returned.
\item[(Reject)] Else $\randMo = \randMm \bmod n$; and the loop continues.
\end{enumerate}
It follows that the output information $O_j$ at the end of iteration $j$ satisfies
\begin{equation}
O_j = \left\lbrace\begin{array}{lll}
  \log(\floor{\randMm/n}) + \log(n)
    & \mbox{with probability } \displaystyle \frac{n \floor{\randMm/n}}{\randMm}
    &\mbox{(accept case)}
  \\
  \log(\randMm \bmod n)
    & \mbox{with probability } \displaystyle 1-\frac{n \floor{\randMm/n}}{\randMm}
    & \mbox{(reject case).}
\end{array}\right.
\label{eq:expected-output-entropy-step-j}
\end{equation}

Applying \cref{eq:entropy-mid-in-toss}, the expected entropy loss at iteration $j$ is then
\begin{align}
\expect{ (\log(\randMi) + \iFlip_j) - O_j \mid M^{-}} = \expect{(\log(\randMm) - O_j) \mid \randM^{-}}.
\end{align}
Applying \cref{eq:expected-output-entropy-step-j} gives
\begin{align}
&\expect{\log(\randMm) - O_j \mid \randM^{-}}
\notag
\\
&=
\expect{
  \log(\randMm)
  -
  \frac{n \floor{\randMm/n}}{\randMm}
  \left(\log(\floor{\randMm/n})+ \log(n) \right)
  -
  \frac{\randMm \bmod n}{\randMm}
  \left(\log(\randMm \bmod n)\right)
  ~\Big\vert~ \randM^{-}
  }
\\
&=
\expect{
\frac{n \floor{\randMm/n}}{\randMm}
  \left(\log\left(\frac{\randMm}{n\floor{\randMm/n}} \right)\right)
  +
  \frac{\randMm \bmod n}{\randMm}
  \left(\log\left(\frac{\randMm}{\randMm \bmod n} \right)\right)
  ~\Big\vert~\randM^{-}
  }
\\
&=
\expect{H_{\rm b} \left( \frac{\randMm \bmod n}{\randMm} \right)~\Big\vert~\randM^{-}}
\\
&\le
H_{\rm b} \left( \frac{(n-1)}{\randMin} \right).
\label{eq:iteration-bound-entropy-loss}
\end{align}

We now argue that the sum of these entropy losses over the $J$
random iterations precisely corresponds
to the total entropy loss \cref{eq:entropy-loss-uniform}, as follows:
\begin{align}
&\sum_{j=1}^{J}\left(\left(\log(\randMi) + \iFlip_j\right) - O_j\right) \notag \\
&= \log(\randMi[1]) + \left(\iFlip_1 + \dots \iFlip_j\right) + \sum_{j=2}^J(\log(\randMi)-\log(\randMo[j-1])) - (\log(n) + \log(\randMo[J])) \\
&= \left(\log(\randM^{-}) + \iFlip \right) - (\log(n) + \log(\randM^{+})),
\end{align}
where we have used the properties in \cref{eq:entropy-loop-properties}
and the fact that $O_j$ is precisely in the \textit{reject case}
of \cref{eq:expected-output-entropy-step-j} for $j \in [1,J-1]$
(i.e., $\randMo = \randMm \bmod n = \randMi[j+1]$)
by the definition of $J$.

Finally, since the acceptance probability in \cref{eq:expected-output-entropy-step-j}
is at least $(\randMin -n + 1) / \randMin$, the expected number of iterations is
$\expect{J} \le \randMin/(\randMin -n + 1)$, which combined with
\cref{eq:iteration-bound-entropy-loss} gives the result:
\begin{align}
\expect{\sum_{j=1}^{J}\left( \log(\randMi) + \iFlip_j - O_j\right) ~\Big\vert~\randM^{-}, U}
\label{eq:conditional-entropy-loss-given-U}
&\le
\expect{\sum_{j=1}^{J}H_{\rm b} \left( \frac{(n-1)}{\randMin} \right) ~\Big\vert~\randM^{-}, U }
\\
&=
\expect{J ~\Big\vert~\randM^{-}, U} H_{\rm b} \left( \frac{(n-1)}{\randMin} \right)
\\
&\le
\frac{\randMin}{(\randMin -n + 1)} H_{\rm b} \left( \frac{(n-1)}{\randMin} \right).
\end{align}
In \cref{eq:conditional-entropy-loss-given-U}, conditioning on the returned
uniform $U$ is justified by the following facts:
\begin{itemize}[noitemsep]
\item Conditioned on $\randM^{-}$, the number of loop
      iterations $J$ is independent of the return value $U$.

\item The initial state $\randMi[1] = \randM^{-}$ is vacuously
      conditionally independent of $U$ given $\randM^{-}$.

\item For each $j \in [1,J]$, coin toss counter
      $\iFlip_j = W - (\floor{\log_2(\randMi[j])} + 1)$
      (cf.~\cref{algline:refill-fresh-k} of \cref{alg:refill})
      is determined by $\randMi$.
      For $j \in [1,J]$,
      the output entropy $O_j$ and state bound $\randMo[j]$
      are determined by $\randMi$, $\iFlip_j$, and $J$
      (cf.~\cref{eq:expected-output-entropy-step-j,eq:entropy-mid-in-toss}).
      Lastly, $\randMo[j] = \randMi[j+1]$ by definition for $j \in [1,J-1]$.

\item By induction, for all $j \in [1,J]$, the quantities
      $\randMi$, $\iFlip_j$, and $O_j$, are conditionally
      independent of $U$ given $\randM^{-}$.
      The second point is the base case,
      the third point provides the inductive step,
      and the first point handles the random stopping time $J$
      which is used in each inductive step.
\end{itemize}
\end{proof}

\begin{remark}
In practice, a typical setting to minimize entropy loss
is letting $\randMin=2^{63}$
(i.e., $W=64$; $M$ is guaranteed to fit in an \texttt{unsigned long long int} word)
and $n < 2^{32}$ (guaranteed to fit in an \texttt{unsigned long int} word).
From \cref{prop:entropy-loss-uniform}, these values
guarantee an expected entropy loss that is less than
$2 \times 10^{-8}$ bits,
and a similar bound parametric in $n$ is
$5 \times 10^{-10} \log(n)$ bits.
\end{remark}

\subsection{Entropy Cost of General Sampling}
\label{sec:analysis-nonuniform}

We now analyze the entropy loss that occurs in a single invocation
of $\Call{Inversion}{}$ (\cref{alg:inversion}).
Any invocation $\Call{Inversion}{\dist}$
uses exactly \textit{one call}
$U \gets \Call{Uniform}{n}$
(on \cref{algline:inversion-uniform})
It then uses \cref{prop:split-nonuniform} to reversibly decompose $U$ into
a sample $X$ from the desired target distribution together with a leftover
uniform state $(\randS', \randM')$ that is recycled into the global uniform
state.
Therefore, the information loss should correspond precisely to that of
generating the uniform $U$ in \cref{prop:entropy-loss-uniform}.

\input{fig-inversion}

We make this intuition precise.
Consider a generic invocation of $X \sim \Call{Inversion}{\dist}$, given a
random state $(\randS^{-}, \randM^{-})$ such that
$\randS^{-} \mid \randM^{-} \sim \Uniform[0,\randM^{-})$.
Let $n \in [1, \randMin]$ denote the argument to $\Call{Uniform}{}$
on \cref{algline:inversion-uniform}.
Write $(\randS^\star, \randM^\star)$ to be the uniform state after
\cref{algline:inversion-uniform}
and let $\iFlip$ denote the total number of coin tosses requested by $\Call{Uniform}{}$.
Let $(\randS', \randM')$ be as defined in \cref{algline:inversion-recycle}
and write
$(\randS^+, \randM^+)$ for the final random state after
\cref{algline:inversion-recycle-call}, which is the result of applying
\cref{prop:merge-uniform} to
$(\randS^\star, \randM^\star)$ and $(\randS', \randM')$.
\Cref{fig:inversion-information} shows the information flow in
$\Call{Inversion}{}$ and \cref{prop:entropy-loss-nonuniform} formalizes
the entropy loss bound with respect to this diagram.

\begin{proposition}
\label{prop:entropy-loss-nonuniform}
The entropy loss of any invocation $\Call{Inversion}{\dist}$ of \cref{alg:inversion}
satisfies
\begin{equation}
\expect{(\log(\randM^{-}) + \iFlip) - \left(\log\left(\frac{1}{\dist(X)}\right) + \log(\randM^{+})\right)~\Bigg\vert~\randM^{-}, U, X}
\leq
\frac{\randMin}{\randMin - n + 1} H_{\rm b}\left( \frac{n-1}{\randMin} \right),
\label{eq:entropy-loss-nonuniform-bound}
\end{equation}
where the expectation is taken over
the random coin tosses from the entropy source $\Call{Flip}{}$.
\end{proposition}

\begin{proof}
From \cref{prop:split-nonuniform}, the variable $U$ is reversibly split
into $(X, \randS')$, which via~\cref{eq:information-content-preserved}
preserves the information content pointwise:
\begin{align}
\log{n} = \log\left( \frac{1}{\dist(x)} \right) + \log(\randM').
\end{align}
From \cref{prop:merge-uniform}, the variables $((\randS^\star,\randM^\star), (\randS',\randM'))$
are reversibly merged into $(\randS^+, \randM^+)$, which gives
\begin{align}
\log(\randM^+) = \log(\randM^\star) + \log(\randM').
\end{align}
These equalities establish that $\log{n} + \log{M^\star} = \log(1/\dist(x)) + \log{M^+}$,
where the expression on the left-hand-side is precisely the output entropy in
\cref{eq:entropy-loss-uniform-bound}.
The conclusion follows from \cref{prop:entropy-loss-uniform}, where the
additional condition on $X$ in \cref{eq:entropy-loss-nonuniform-bound} is
permitted as $X$ is a \textit{deterministic} function of $U$
(cf.~\cref{algline:inversion-search} of \cref{alg:inversion}).
\end{proof}

\subsection{Entropy Cost of Generating a Random Sequence}
\label{sec:analysis-seq}

Toward proving \cref{theorem:entropy-cost},
we analyze the overall entropy cost of using $\Call{Inversion}{}$
to generate an output sequence $\bX = (X_i)_{i \ge 1}$,
given an arbitrary distribution sequence $\dists = (p_i)_{i \ge 1}$.

\begin{proposition}
\label{prop:entropy-waste-total}
The entropy cost
of a sequence of calls of $\Call{Inversion}{\dist_i}$
to generate $(X_i \sim \dist_i)_{i \ge 1}$ satisfies
\begin{align}
\expect{\tosses{k}{\Flips, \dists} \mid X_1=x_1, \dots, X_k = x_k} <
\sum_{i=1}^{k}\log\left(\frac{1}{\dist_i(x_i)}\right)
+
k \frac{\randMin}{\randMin-d+1} H_{\rm b}\left( \frac{d-1}{\randMin} \right) + W,
\label{eq:entropy-waste-total}
\end{align}
where $\set{X_1 = x_1, \dots, X_k = x_k}$ is a positive probability
event and $d \in [1, \randMin]$ is an upper bound on the argument
$n$ in any invocation of $\Call{Uniform}{n}$
(i.e., $\dist_{\leq k} \in (\DAlphabet_d)^k$).
\end{proposition}

\begin{proof}
Write $\tosses{k}{\Flips, \dists} = \iFlip_1 + \dots + \iFlip_k$ as the sum
of coin tosses used in each round and denote the conditioning
event as $E \defeq \set{X_1 = x_1, \dots, X_k = x_k}$.
Then
\begin{align}
&\expect{\tosses{k}{\Flips, \dists} \mid E}
= \expect{\iFlip_1 + \dots \iFlip_k \mid E}
= \sum_{i=1}^{k}\expect{\iFlip_i \mid E} \\
&= \sum_{i=1}^{k}\expect{ \expect{\iFlip_i \mid \randM_i^{-}, U_i, X_i=x_i} \mid E}
  \label{eq:entropy-waste-total-pf-tower-property} \\
&\le \sum_{i=1}^{k}
  \left[\log\left(\frac{1}{\dist(x_i)}\right)
  +
  \expect{\expect{\log(\randM^{+}_{i}) -\log(\randM^{-}_{i}) \mid \randM^{-}_{i}, X_i=x_i} \mid E}
  +
  \frac{\randMin}{\randMin - n + 1} H_{\rm b}\left( \frac{n-1}{\randMin} \right)
  \right]
  \label{eq:entropy-waste-total-pf-tower-bound}
  \\
&= \sum_{i=1}^{k}\log\left(\frac{1}{\dist_i(x_i)}\right)
  +
  k \frac{\randMin}{\randMin-d+1} H_{\rm b}\left( \frac{d-1}{\randMin} \right)
  +
  \sum_{i=1}^{k}\expect{\expect{\log(\randM^{+}_{i}) -\log(\randM^{-}_{i}) \mid X_i=x_i, \randM^{-}_{i}} \mid E}.
  \label{eq:entropy-waste-total-pf-telescope}
\end{align}
\Cref{eq:entropy-waste-total-pf-tower-property} uses the tower property of conditional
expectation, where the fresh coin guarantee guarantee~\cref{eq:ors-fresh} ensures that
$\iFlip_i$ is conditionally independent of all events in $E$ except
$\set{X_i=x_i}$ given the initial state $\randM^{-}_i$.
\Cref{eq:entropy-waste-total-pf-tower-bound} applies the bound from
\cref{eq:entropy-loss-nonuniform-bound} and monotonicity of conditional
expectation.
We next analyze the term~\cref{eq:entropy-waste-total-pf-telescope}, recalling
that (by definition) $\randM^{-}_1 = \randM_0$ and $\randM^{-}_i = \randM^{+}_{i-1}$
for $i=2,\dots,k$.
The shared terms for $i=1,\dots,k-1$ in the telescoping series are
\begin{align}
& \expect{\expect{\log(\randM^+_i) \mid \randM^-_i, X_i=x_i} \mid E} = \expect{\log(\randM^+_{i}) \mid E} = \expect{\log(\randM^-_{i+1}) \mid E},\\
& \expect{\expect{-\log(\randM^-_{i+1}) \mid \randM^-_{i+1}} \mid E}  = \expect{-\log(\randM^-_{i+1}) \mid E}.
\end{align}
where the first equality follows from the fact that
$\randM^+_{i}$ is independent of $E$ given $\randM^-_i$ and $X_i$.
The surviving terms from the final sum in \cref{eq:entropy-waste-total-pf-telescope} are
\begin{align}
\expect{ \expect{ \log(\randM^+_{k}) \mid \randM^-_{k}, X_k=x_k } - \expect{ \log(\randM^{-}_1) \mid \randM^-_{1}, X_1 = x_1} \mid E} < W,
\label{eq:entropy-waste-total-telescope-final}
\end{align}
where use the fact that $\randM^{-}_1 = \randM_0 = 1$ and at every step $\randM^{+} < 2^W$ surely.
The conclusion~\cref{eq:entropy-waste-total} follows.
\end{proof}

As the bound involving
$\expect{\tosses{k}{\Flips, \dists} \mid X_{\leq k}=x_{\leq k}} - \sum_{i=1}^{k}\log(1/\dist_i(x_i))$
in \cref{eq:entropy-waste-total} holds for \textit{every} distribution sequence
$\dists$, an identical bound on
$\expect{\tosses{k}{\Flips, \Dists} \mid X_{\leq k}=x_{\leq k}, \Dist_{\leq k}=\dist_{\leq k}} - \sum_{i=1}^{k}\log(1/\dist_i(x_i))$
holds for any random distribution sequence $\Dists$ satisfying \cref{eq:ors-fresh},
as noted in \cref{eq:main-thm-cost-random}.

\begin{remark}
The same result holds for any randomness recycling sampler that generates
$X \sim \dist$ by calling $\Call{Uniform}{}$ once, such as
\cref{alg:lookup,alg:alias,alg:ddg} in \cref{sec:rr-general}.
\end{remark}

\subsection[Proof of \crefnameof{theorem}~\ref{theorem:entropy-cost}]{Proof of \cref{theorem:entropy-cost}}
\label{sec:analysis-theorem}

This section proves the main result, which is restated below.

\EntropyCost*

\input{alg-loop}

\begin{proof}
\Cref{alg:loop} is a witness to the online algorithm from the theorem statement,
making repeated use of the inversion sampler in \cref{alg:inversion}.
The $\Call{Recycle}{\randS', \randM'}$, $\Call{Refill}{{}}$,
and $\Call{Uniform}{n}$ methods each use $O(W)$ space.
The global uniform state $(\randS, \randM)$,
which is the only auxiliary state carried over between rounds,
uses $2W$ bits of space.
\Cref{alg:uniform,alg:inversion} will never overflow the
$W$-bit size of the global uniform state
$(\randS, \randM)$.
After calling $\Call{Uniform}{n}$,
the global uniform state has a bound of $\randM \leq (2^W-1)/n$.
Any corresponding recycled uniform state $(\randS', \randM')$ always
satisfies $\randM' \leq n$.
To arrive at the entropy bound in \cref{theorem:entropy-cost},
by \cref{prop:entropy-waste-total} it suffices to find $W \ge 1$ such that
\begin{align}
\frac{1}{1-(d-1)/2^{W-1}} H_{\rm b}\left( \frac{d-1}{2^{W-1}} \right)
\leq \varepsilon.
\end{align}
Let $h(\delta) \defeq H_{\rm b}(\delta) / (1 - \delta)$.
Then $h(\delta) \sim \delta \log(1/\delta)$ as $\delta \to 0$,
so
\begin{align}
W_{d, \varepsilon} \defeq 1 + \ceil{\log((d-1) / h^{-1}(\varepsilon))}
\label{eq:W-d-epsilon}
\end{align}
satisfies $W_{d, \varepsilon} \sim \log(d / \varepsilon)$
as $d / \varepsilon \to \infty$,
which matches the claimed entropy-space tradeoff.
\end{proof}

\begin{remark}
The exact inverse of $h(\delta)$ requires solving a transcendental equation,
which is generally not possible in closed form, but it can be approximated
with a series expansion, which may look roughly like
\begin{equation}
W_{d,\varepsilon} \approx 1+\ceil{\log(d) + \sum_{m=1}^{\log^*(1/\varepsilon)} \log^{\circ{m}}(1/\varepsilon)}.
\end{equation}
\end{remark}

\begin{remark}
We assume the rational target distribution $\dist$ is given
as either a list of integer weights or pointwise access to its cumulative
distribution function, which enables linear or logarithmic time complexity.
A detailed discussion of time complexity lower bounds for random sampling under
different representations of the target distribution---e.g., as probability
mass functions or cumulative probabilities, and possibly as computable
reals---is available in \citet{Trevisan2010,yamakami1999}.
\Citet{bringmann2017} discuss the time complexity of discrete sampling algorithms
in the real RAM model for a static array of sorted or unsorted probabilities.
\end{remark}

%% file: alg-uniform.tex
%!TEX root=./main.tex
\begin{algorithm}
\caption{Uniform sampling with randomness recycling}
\label{alg:uniform}
\begin{algorithmic}[1]
\Require{%
  Read and write access to global variables $(\randS, \randM)$ from \cref{alg:recycle};\\
  Integer $n \in [1, 2^{W-1}]$, where $W$ is the bit width of $\randS$ and $\randM$.
  }
\Ensure{Random sample $U \sim \Uniform[0,n)$}
\Procedure{Uniform}{$n$}
\While{\textbf{true}}
  \State $\Call{Refill}$ \Comment{refill global uniform state (\cref{alg:refill})}
    \label{algline:uniform-refill}
  \State $(q_\randS, r_\randS) \gets \Call{DivMod}{\randS,n}$ \Comment{$\randS = q_\randS n + r_\randS$}
  \State $(q_\randM, r_\randM) \gets \Call{DivMod}{\randM,n}$ \Comment{$\randM = q_\randM n + r_\randM$}
  \If{$q_\randS < q_\randM$}
    \Comment{accept case}
    \State $\textbf{update}~(\randS, \randM) \setp (q_\randS, q_\randM)$ \Comment{recycling}
      \label{algline:uniform-recycle-accept}
    \State \Return $r_\randS$ \Comment{$r_\randS \sim \Uniform[0,n)$}
  \Else
    \Comment{reject case}
    \State $\textbf{update}~(\randS,\randM) \setp (r_\randS, r_\randM)$ \Comment{recycling}
      \label{algline:uniform-recycle-reject}
  \EndIf
\EndWhile
\EndProcedure
\end{algorithmic}
\end{algorithm}

%% file: alg-inversion.tex
%!TEX root=./main.tex
\begin{algorithm}
\caption{Inversion sampling with randomness recycling}
\label{alg:inversion}
\begin{algorithmic}[1]
\Require{Positive integers $a_0, \dots, a_{n-1}$ with prefix sums $A_i = \sum_{j=0}^{i} a_j$
  (or a computable function $i \mapsto A_i$).
}
\Ensure{Random sample $X \sim \Discrete(a_0,\dots,a_{n-1})$}
\Procedure{Inversion}{$a_0, \dots, a_{n-1}$}
% \State Set
%   $A_{-1} \gets 0$
%   and $A_i \gets a_0 + \dots + a_i$ for $i=0,\dots,{n-1}$
%                                           \Comment{prefix sums (if not given)}
\State $U \sim \Call{Uniform}{A_{n-1}}$   \Comment{draw uniform variate (\cref{alg:uniform})}
                                          \label{algline:inversion-uniform}
\State Let $X \gets \min\set{i \in [0,n) \mid U < A_i}$
                                          \Comment{search for sample $X$}
                                          \label{algline:inversion-search}
\State $(\randS',\randM') \gets (U-A_{X-1}, a_X)$
                                          \Comment{extract uniform state}
                                          \label{algline:inversion-recycle}
\State $\Call{Recycle}{\randS', \randM'}$ \Comment{recycle the uniform state (\cref{alg:recycle})}
                                          \label{algline:inversion-recycle-call}
\State \Return $X$
\EndProcedure
\end{algorithmic}
\end{algorithm}

%% file: fig-inversion.tex
%!TEX root=./main.tex
\begin{figure}[h]
\centering
\begin{tikzpicture}[thick]
\tikzset{io/.style={minimum width=1cm, minimum height=.75cm,draw=none}};
\tikzset{op/.style={minimum width=1cm, minimum height=1.5cm,draw=black}};

\node[name=uniform, op, align=center]{$\Call{Uniform}{}$\\\cref{algline:inversion-uniform}};
\node[name=search, right=3 of uniform, op, align=center]{$\Call{Search}{}$\\\cref{algline:inversion-search}};
\node[name=recycle, right=3 of search.south east, anchor=south west, io, draw, align=center]{$\Call{Recycle}{}$\\\cref{algline:inversion-recycle}};

\node[name=initial-state, io, left= 1 of uniform.north west, anchor=north east]{$(\randS^{-}, \randM^{-})$};
\node[name=initial-flips, io, left= 1 of uniform.south west, anchor=south east]{$\Call{Flip}{\iFlip}$};
\draw[-latex] (initial-state.east) -- (initial-state.east -| uniform.west);
\draw[-latex] (initial-flips.east) -- (initial-flips.east -| uniform.west);

\node[name=uniform-u, io, right= .5 of uniform.north east, anchor=north west]{$U$};
\node[name=uniform-state, io, right= .5 of uniform.south east, anchor=south west]{$(\randS^{\star}, \randM^{\star})$};

\draw[-latex] (uniform.east |- uniform-u.west) -- (uniform-u.west);
\draw[-latex] (uniform.east |- uniform-state.west) -- (uniform-state.west);

\draw[-latex] (uniform-u.east) -- (uniform-u.east -| search.west);

\node[name=search-x, io, right= .5 of search.north east, anchor=north west]{$X$};
\node[name=search-state, io, right= .5 of search.south east, anchor=south west]{$(\randS^{'}, \randM^{'})$};

\draw[-latex] (search.east |- search-x.west) -- (search-x.west);
\draw[-latex] (search.east |- search-state.west) -- (search-state.west);

\draw[-latex] (search-state.east) -- (search-state.east -| recycle.west);

\node[name=recycle-state, io, right= .5 of recycle.south east, anchor=south west]{$(\randS^{+}, \randM^{+})$};
\draw[-latex] (recycle.east |- recycle-state.west) -- (recycle-state.west);

\draw[-latex] (uniform-state.south) -- ([yshift=-.4cm]uniform-state.south) -- ([yshift=-.4cm]uniform-state.south -| recycle.south) -- (recycle);
\end{tikzpicture}
\caption{Information flow in $\Call{Inversion}{}$ (\cref{alg:inversion}) for sampling general distributions.}
\label{fig:inversion-information}
\end{figure}

%% file: alg-loop.tex
%!TEX root=./main.tex
\begin{algorithm}
\caption{Online random sampling algorithm witnessing \cref{theorem:entropy-cost}}
\label{alg:loop}
\begin{algorithmic}[1]
\Require{Target distributions $\Dists = (\Dist_i)_{i \ge 1}$,
  with $\Dist_i$ presented at step $i \ge 1$.}
\Ensure{Sequence of output samples $X_i \sim \Dist_i$, for $i \ge 1$.}
\Procedure{RandomSequence}{}
\For{$i \gets 1$ to $\infty$}
  \State \labelcref{step:observe} Receive the next target distribution $\Dist_i$
  \State \labelcref{step:generate} $X_i \gets \Call{Inversion}{P_i}$ \Comment{(\cref{alg:inversion})}
  \State \textbf{yield} $X_i$
\EndFor
\EndProcedure
\end{algorithmic}
\end{algorithm}

%% file: sec-uniform.tex
%!TEX root=./main.tex

\section{Randomness Recyclers for Uniform Distributions}
\label{sec:rr-uniform}

In this section we compare the randomness recycler
\cref{alg:uniform} for discrete uniforms to previous uniform samplers
(\crefrange{sec:rr-uniform-fdr}{sec:rr-uniform-rust}) and then present
algorithmic extensions that leverage the speed of word-level operations on modern CPUs
(\cref{sec:rr-uniform-optimized}).

\subsection{Comparison to the \texorpdfstring{\citeauthor{lumbroso2013}}{Lumbroso} Uniform Sampler}
\label{sec:rr-uniform-fdr}

The Fast Dice Roller (FDR) algorithm from \citet[p.~4]{lumbroso2013}
for generating a discrete uniform
can be understood as a special case of \cref{alg:uniform}.
It uses a different strategy in the call to
$\Call{Refill}{\randS,m}$ on \cref{algline:uniform-refill}
than our strategy described in \cref{remark:uniform-refill}.
The ``randomness recycling'' interpretation of FDR in
\citet{huber2024} shows that the algorithm
refills the uniform random state by repeatedly drawing fresh
random bits until $\randM \geq n$, rather than our approach of always
ensuring $\randM \in [2^{W-1}, 2^W)$.
FDR is entropy optimal in the sense of \citeauthor{knuth1976} when taking
a \textit{single sample}, and coincides with the suboptimal baseline
from \cref{sec:intro-existing-suboptimal} when na\"ively used to generate a
sequence of samples.

\Citet[Section 3.1]{lumbroso2013} observes that,
for generating an i.i.d.~sequence $U_0, U_1, \dots, U_{k-1} \sim \Uniform[0,n)$,
it is possible to first
generate $Y = \Uniform[0,n^k)$ entropy optimally using FDR and then
recover the individual $U_i$ using the decomposition of $Y$ in base-$n$:
\begin{equation}
Y = U_{k-1} n^{k-1} + \dots + U_0 n^0.
\label{eq:lumbroso-nary-base}
\end{equation}
This technique coincides with batched baseline
from~\cref{eq:ky-batched}.
Because each distribution $P_i$ is uniform, the space
complexity of storing $P_{1:k}$ does not grow exponentially.
However, computing the decomposition~\cref{eq:lumbroso-nary-base}
requires integer division operations,
which becomes expensive for large $n^k$.
The method also requires prespecifying $k$ and the
ranges of the $U_i$ beforehand.
Sequential sampling with a fixed batch size $2/\varepsilon$ allows FDR to achieve
the asymptotic entropy rate of $\entropyRV{U_0,\dots, U_{k-1}}/k + \varepsilon$ bits per sample,
but the $O(1/\varepsilon)$ computational overhead exceeds the
$O(\log(1/\varepsilon))$ achieved by \cref{alg:uniform} using randomness recycling.

The analysis of FDR reveals that the \citeauthor{knuth1976} entropy toll of
less than $2$ bits arises from the rejection probability $r_\randM / \randM$
when splitting the uniform state.
In the worst case, this probability approaches a geometric distribution
with parameter $1/2$.
Using a larger value of $\randM \gg n$ reduces this rejection probability,
and in turn the expected entropy cost of a given run.
A larger uniform state requires an upfront investment of entropy, but
amortized over many samples, the total cost quickly falls below that of any
single-sample entropy-optimal sampler.
\Cref{remark:uniform-refill} shows that if $W = 64$ bits and $n < 2^{32}$,
then the rejection probability is less than $2^{-31}$.
The expected amortized entropy cost of sampling each individual uniform is
then bounded by
\begin{equation}
(1+10^{-9})\entropyRV{U} + 64/k \textrm{ bits per sample}.
\label{eq:entropy-cost-bound}
\end{equation}
Thus, \cref{alg:uniform} is a natural way to extend the spirit of batched
generation \`{a} la \citeauthor{lumbroso2013} in an online and space-efficient
manner.

\subsection{Comparison to the \texorpdfstring{\citeauthor{lemire2019}}{Lemire} Uniform Sampler}
\label{sec:rr-uniform-rust}

\input{fig-benchmark-uniform}

Discrete uniform samplers often scale a uniform state over a range $[0, 2^w)$
(where the word size is typically $w = 64$) to a smaller range $[0,n)$
using integer division, with rejection sampling to ensure that the result is exact.
\Citet{lemire2019} describes a fast exact uniform sampler that eliminates
the division in almost all cases by instead using a widening multiplication,
which is efficient on modern CPUs (e.g., x86).
%
% The Rust random library uses this algorithm with a precomputed remainder,
% which speeds up the runtime when the range $n$ is known at compile time
% or when samples are repeatedly drawn from the same range $n$.
%
This method is especially fast when the randomness source is inexpensive.
In cases where the randomness source is expensive, the entropy inefficiency
of the method becomes significant.
\Cref{fig:benchmark-uniform} compares the method of \citeauthor{lemire2019}
to \cref{alg:uniform} and to the Fast Dice Roller from \citeauthor{lumbroso2013}
when using a cryptographically secure PRNG.
The plot of sampling time shows three distinct regimes:
\begin{itemize}[wide,leftmargin=*]
\item For large $n$, near the maximum $64$-bit integer,
each sampler consumes approximately the same amount of entropy,
roughly $64$ bits per sample, so the difference in sampling time is
determined by the other operations.
The widening multiplication of \citeauthor{lemire2019} is faster than both the
divisions in \cref{alg:uniform} and the loop over the bit-length of $n$ in
the FDR sampler, which requires roughly $64$ iterations in this regime.
\item For small $n$, the method of \citeauthor{lemire2019} is slowest because
it still consumes $64$ bits of entropy per sample, which is wasteful.
The FDR sampler can be fastest in this case because it is reasonably entropy-efficient
and only requires a few loop iterations with efficient bitwise operations,
rather than the division in \cref{alg:uniform}.
\item For intermediate $n$, \cref{alg:uniform} is fastest because it is more
entropy efficient than the method of \citeauthor{lemire2019} and even FDR,
without looping over the bit-length of $n$ as FDR does.
\end{itemize}

\subsection{Optimized Uniform Sampling via Widening Multiplication and Batching}
\label{sec:rr-uniform-optimized}

Real-world software libraries for uniform sampling leverage code
optimizations that exploit the speed of word-level operations on modern
CPUs.
We develop additional randomness recycling techniques for uniform
samplers that incorporate two optimizations: widening multiplication and
batched sampling.
We will show how randomness recycling can be used to speed up highly
optimized uniform samplers from the literature that use widening
multiplication and batching as follows.

\begin{itemize}[wide,leftmargin=*]
\item \Cref{alg:recycle-widening,alg:uniform-widening}
adapt \cref{alg:recycle,alg:uniform} by
using specialized recycling rules that use efficient widening
multiplication.

\item \Cref{alg:uniform-lemire} augments the widening-multiplication
method of \citet{lemire2019} discussed in \cref{sec:rr-uniform-rust} with
randomness recycling.

\item \Cref{alg:uniform-brackett} augments the widening-multiplication
and batched-sampling method of \citet{brackett2025} with randomness
recycling.
\end{itemize}

We assume for this section that $W$ is the word size,
the randomness stream is read in $W$-bit words, and $W$-bit integer
operations are efficient (in particular, widening multiplication, which maps
two $W$-bit integers to their $2W$-bit product, stored as two $W$-bit words
that contain the high and low bits of the result, respectively).

\subsubsection{Recycling with Widening Multiplication}
\label{sec:rr-uniform-optimized-widening}

Implementing the uniform-merging map of \cref{prop:merge-uniform} in finite
space as \cref{alg:recycle} could result in overflow if the product $\randM
\randM'$ exceeds $2^W - 1$, so we must ensure a relationship between the global
uniform state bound $M$ and the size $M'$ of any state which we want to recycle.
This requirement restricts algorithms that recycle uniform states and can make
them significantly slower.
For example, the call to $\Call{Refill}{}$ in \cref{alg:uniform} requires
that the randomness stream be shifted by variable numbers of bits,
instead of simply read word by word.
\Cref{alg:recycle-widening} shows how widening multiplication can be used to merge two word-sized
uniform states into a single word-sized uniform state, with an additional independent full word of
i.i.d.~random bits in the case that the product would overflow.

\input{alg-recycle-widening}

\Cref{alg:recycle-widening} can be seen as merging two states
as in $\Call{Recycle}{}$ (\cref{alg:recycle}) and then splitting the result
as in $\Call{Uniform}{2^W}$ (\cref{alg:uniform}),
so the analysis of correctness is similar.
By \cref{prop:merge-nonuniform}, $2^W \randS_{\rm hi} + \randS_{\rm lo}$
is uniformly distributed over $[0, 2^W \randM_{\rm hi} + \randM_{\rm lo})$
Conditioned on the event $\randS_{\rm hi} = \randM_{\rm hi}$, the value
$\randS_{\rm lo}$ is uniformly distributed over $[0, \randM_{\rm lo})$
by \cref{prop:split-nonuniform}.
Similarly, conditioned on the event $\randS_{\rm hi} < \randM_{\rm hi}$,
the value $2^W \randS_{\rm hi} + \randS_{\rm lo}$ is uniformly distributed
over $[0, 2^W \randM_{\rm hi})$.
In this case, \cref{prop:split-uniform} shows that the quotient and remainder
must be distributed as
$\randS_{\rm hi} \mid \randM_{\rm hi} \sim \Uniform[0, \randM_{\rm hi})$ and
$\randS_{\rm lo} \sim \Uniform[0, 2^W)$, respectively, and they are independent.

Equipped with \cref{alg:recycle-widening}, we can implement a simpler variant of
$\Call{Uniform}{}$, as shown in \cref{alg:uniform-widening}.
This variant uses random words directly from the source instead of calling
$\Call{Refill}{}$, and we note that it only accesses the global uniform state
indirectly via $\Call{RecycleWidening}{}$.
\Cref{alg:uniform-widening} and the remaining uniform algorithms in this section
only recycle randomness \textit{into} the global uniform state, and the recycled
randomness is reused only in the form of random words produced by $\Call{RecycleWidening}{}$.
The accept-reject branching in \cref{alg:recycle-widening} can waste up to
$H_{\rm b}(2^W / (2^{W+1} - 1)) \approx 1$ bits of entropy in the worst case,
so samplers like \cref{alg:uniform-widening} do not achieve the
same entropy efficiency as \cref{alg:uniform}.
However, in practice, the entropy waste is not too large, and samplers based on
this recycling method can be significantly faster than those using \cref{alg:uniform},
even when the randomness source is very expensive.

\input{alg-uniform-widening}

\subsubsection{Uniform Sampling with Widening Multiplication and Batching}
\label{sec:rr-uniform-optimized-batching}

\Cref{alg:uniform-widening} benefits significantly from the word-level operations,
but it still requires two divisions.
\Citet{lemire2019} shows how to essentially eliminate the division,
but there is no apparent efficient way to implement a recycling rule;
the leftover randomness from their algorithm is found in the
low bits of the widening multiplication result, and a na\"ive approach to
extract a recyclable uniform state introduces two extra divisions,
entirely nullifying the benefit of fewer divisions.
Instead of recycling after running the algorithm of \citet{lemire2019} verbatim,
we can instead eliminate a division by using a generalization of this method to
batched sampling.
\Citet{brackett2025} show how to generate multiple uniforms by repeated
widening multiplications from a single random word, which is batched
analogously to~\cref{eq:ky-batched}.
We can apply this method to batch together uniforms over the ranges $n$ and
$\floor{2^W / n}$, which are the target uniform and the recyclable state,
respectively.
\Cref{fig:brackett} shows this method, for the example of generating a uniform
of size $n=6$ given a uniform random word of length $W=4$ bits.

\input{fig-brackett}

The rows of this diagram are interpreted as follows:
\begin{itemize}[wide,leftmargin=*,noitemsep]
\item The first row shows the possible values of a uniform random word $U \in [0,16)$.

\item The second and third rows show the target uniform over $n = 6$ outcomes and the
recyclable uniform state over $\floor{2^W / n} = 2$ outcomes, respectively.

\item The fourth row shows the rejection condition
$n \floor{2^W / n} U \bmod 2^W < 2^W \bmod n$.

\item The fifth row shows a novel recycling rule for this method, which is justified
in \cref{prop:lemire-reject-recycle}.
\end{itemize}

The recycling rule in the reject case is not essential, as
noted by \citet[Remark 7]{mennucci2010}, because the rejection probability is
vanishingly small in practice.
Even more significantly, the entropy inefficiency of $\Call{RecycleWidening}{}$
dominates the efficiency gained by recycling the reject case in
\cref{alg:uniform-lemire,alg:uniform-brackett}.

The following result proves the correctness of the recycling rule.

\begin{proposition}
\label{prop:lemire-reject-recycle}
Let $N = n \floor{2^W / n} \in (2^{W-1}, 2^W]$ be the target number of uniform
outcomes including the recyclable state, and let $U \sim \Uniform[0,2^W)$.
Let
$R \defeq \set{u \in [0, 2^W) \mid u N \bmod 2^W < 2^W - N}$
be the set of reject outcomes
and $f(u) \defeq u - \floor{u N / 2^W}$.
Then $[f(U) \mid \set{U \in R}] \sim \Uniform[0, 2^W - N)$.
\end{proposition}

\begin{proof}
First observe that $f(0) = 0$.
Next, note that
% \begin{equation}
$\floor{(x+1) N / 2^W} = \floor{x N / 2^W}
\iff
x N \bmod 2^W < 2^W - N.$
% \end{equation}
Otherwise, $\floor{(x+1) N / 2^W} = \floor{x N / 2^W} + 1$.
Therefore, $f(x + 1) = f(x) + 1$ if and only if $x \in R$;
otherwise, $f(x + 1) = f(x)$.
It follows that
\begin{equation}
f(x) = \abs*{[0, x) \cap R}.
\end{equation}
\Citet[Lemma 4.1]{lemire2019} shows that $|R| = 2^W \bmod N = 2^W - N$,
so $f(x)$ ranges over $[0, 2^W - N)$ as $x$ ranges over $R$.
The result is uniform because each value in $[0, 2^W - N)$ is attained
exactly once by some $x \in R$.
\end{proof}

\input{alg-uniform-brackett}
\input{fig-benchmark-uniform-full}

\Cref{alg:uniform-lemire} shows the implementation of this method.
The exact same method can be used to recycle leftover randomness from
the batched sampler of \citet{brackett2025}, by simply adding one uniform
to the end of the batch, as shown in \cref{alg:uniform-brackett}.
This modification of the batched sampler with randomness recycling
is competitive with the original method, even though it introduces
one extra division per batch and computes the product of the ranges
in the batch every time.

% \subsubsection{Benchmark Comparison}

\Cref{fig:benchmark-uniform-full} shows the improvements enabled
by our randomness recycling techniques
(\crefrange{alg:uniform-widening}{alg:uniform-brackett}) in terms of the
entropy cost and wall-clock sampling time, when used to repeatedly
generate many uniforms over a fixed range $n$ (for $n$ between $2$ and $10^{17}$).
The plot shows that \crefrange{alg:uniform-widening}{alg:uniform-brackett}
deliver performance improvements over both \cref{alg:uniform} (from
\cref{sec:rr-uniform}), the prior methods of
\citet{lemire2019,brackett2025}, and
an optimized version of the FDR \citep{lumbroso2013}.
Particularly notable in \cref{fig:benchmark-uniform-full}
is the runtime performance of \cref{alg:uniform-brackett}, which
consistently retains fast sampling time across all values of $n$.

\subsubsection{Application to Generating Random Permutations}
\label{sec:rr-uniform-optimized-permutations}

\input{fig-shuffle}

\Citet{brackett2025} show that variable-sized batching and mixed-radix
decomposition using the method of \citet{lemire2019} to eliminate divisions
yields a state-of-the art sampler for multiple uniforms $(U_0, U_1, \dots, U_{k-1})$.
They apply their sampler to generate random permutations
via the Fisher-Yates shuffle~\citep{durstenfeld1964,fisher1953},
whose key subroutine is efficiently sampling
$U_i \sim \Uniform[0,k-i)$ for $i=0,\dots,k-1$.

Our methods with randomness recycling require more operations and are therefore
comparatively slower when the entropy source is a fast pseudorandom number
generator (PRNG) that quickly returns (pseudo)random coin tosses.
However, when using a more computationally expensive entropy source such as
a cryptographically secure PRNG that is continuously seeded
from system entropy (e.g., \texttt{/dev/random} on Linux),
the entropy consumption becomes a
significant factor in the overall runtime.
In this setting, our randomness recycling algorithms, which consume almost
the minimum possible expected entropy, can surpass the performance of the
\citeauthor{brackett2025} method.
\Cref{fig:shuffle} shows a comparison of the runtime per shuffled item (in nanoseconds)
using four baseline algorithms (Standard, Division Batched-2, BR Batched-2, BR Batched-6)
for uniform sampling from \citeauthor{brackett2025} with
\cref{alg:uniform,alg:uniform-widening,alg:uniform-lemire,alg:uniform-brackett}
from this article.

%% file: fig-benchmark-uniform.tex
%!TEX root=./main.tex
\begin{figure}[t]
\centering
\includegraphics[width=.5\linewidth]{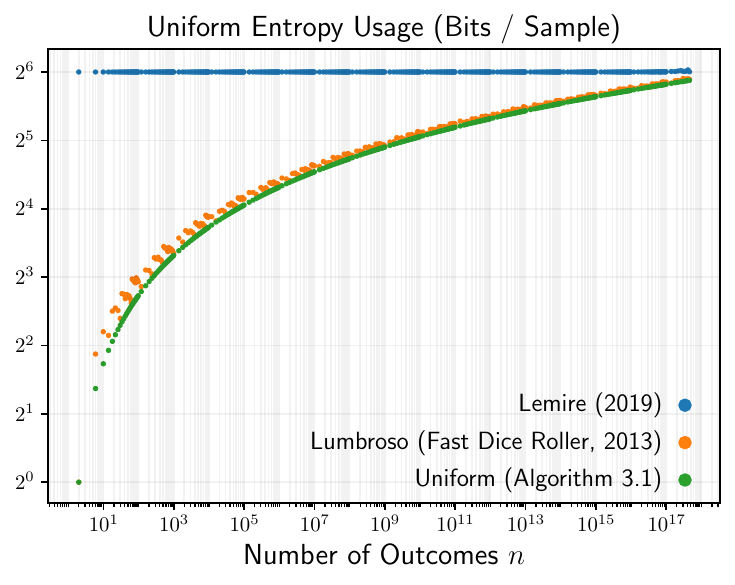}%
\includegraphics[width=.5\linewidth]{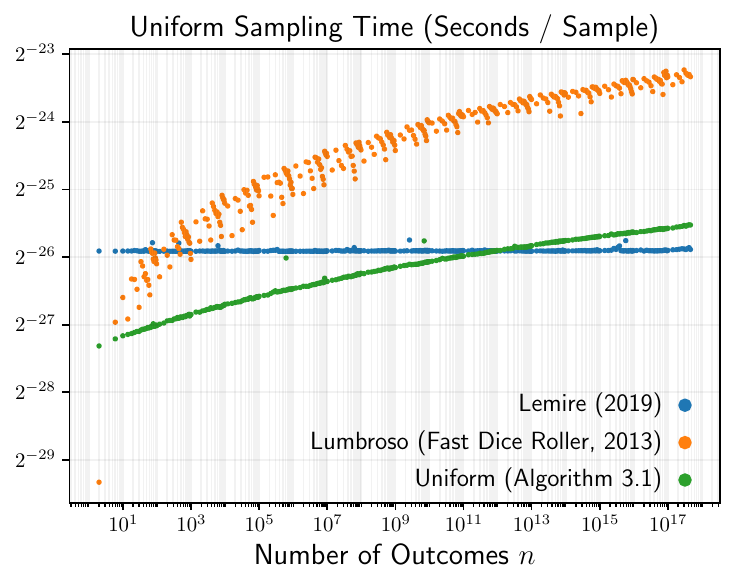}%

\caption{Benchmark comparison of entropy consumption and sampling time
on a range of distribution sizes $n$,
for three uniform samplers:
the Fast Dice Roller of \citet{lumbroso2013},
the method of \citet{lemire2019},
and our uniform sampler with randomness recycling (\cref{alg:uniform}).
Random bits are supplied by $256$-byte buffered requests to \texttt{/dev/random}.}
\label{fig:benchmark-uniform}
\end{figure}

%% file: alg-recycle-widening.tex
%!TEX root=./main.tex
\begin{algorithm}
\caption{Recycling a uniform state into a global uniform state using widening multiplication}
\label{alg:recycle-widening}
\begin{algorithmic}[1]
\Require{%
  Read and write access to global variables $(\randS, \randM)$ from \cref{alg:recycle};
  \\
  Uniform state $(\randS',\randM')$ such that $\randS' \mid \randM' \sim \Uniform[0,\randM')$
  and $\randS' \perp \randS \mid \randM, \randM'$.
  }
\Ensure{Update the global uniform state $(\randS, \randM)$ to incorporate
randomness recycled from $(\randS',\randM')$;
possibly push an independent random uniform word to the entropy stream.}
\Procedure{RecycleWidening}{$\randS',\randM'$}
\State $(\randS_{\rm hi}, \randS_{\rm lo}) \gets \randS + \randS' \otimes \randM$
    \Comment{widening multiplication}
\State $(\randM_{\rm hi}, \randM_{\rm lo}) \gets \randM \otimes \randM'$
    \Comment{widening multiplication}
\If{$\randS_{\rm hi} = \randM_{\rm hi}$}
  \State $\textbf{update}~(\randS, \randM) \setp (\randS_{\rm lo}, \randM_{\rm lo})$
\Else
  \State $\textbf{update}~(\randS, \randM) \setp (\randS_{\rm hi}, \randM_{\rm hi})$
  \State Recycle $\randS_{\rm lo}$ into the entropy source
    \Comment{$\randS_{\rm lo} \sim \Uniform[0,2^W)$}
\EndIf
\EndProcedure
\end{algorithmic}
\end{algorithm}

%% file: alg-uniform-widening.tex
%!TEX root=./main.tex
\begin{algorithm}[t]
\caption{Uniform sampling with $\Call{RecycleWidening}{}$}
\label{alg:uniform-widening}
\begin{algorithmic}[1]
\Require{%
  Integer $n \in [1, 2^W)$, where $W$ is the word size from \cref{alg:recycle}.
  }
\Ensure{Random sample $X \sim \Uniform[0,n)$}
\Procedure{UniformWidening}{$n$}
\While{\textbf{true}}
  \State $X \gets \Call{Flip}{W}$ \Comment{$X \sim \Uniform[0,2^W)$; uniform random word}
  \State $(q_X, r_X) \gets \Call{DivMod}{X,n}$ \Comment{$X = q_X n + r_X$}
  \State $(q_B, r_B) \gets \Call{DivMod}{2^W,n}$ \Comment{${2^W} = q_B n + r_B$}
  \If{$q_X < q_B$}
    \Comment{accept case}
    \State $\Call{RecycleWidening}{q_X, q_B}$ \Comment{(\cref{alg:recycle-widening})}
    \State \Return $r_X$ \Comment{$r_X \sim \Uniform[0,n)$}
  \Else
    \Comment{reject case}
    \State $\Call{RecycleWidening}{r_X, r_B}$ \Comment{(\cref{alg:recycle-widening})}
  \EndIf
\EndWhile
\EndProcedure
\end{algorithmic}
\end{algorithm}

%% file: fig-brackett.tex
%!TEX root=./main.tex
\begin{figure}[H]
\begin{align*}
\begin{adjustbox}{max width=\linewidth}
$\begin{NiceArray}[first-col, hlines,vlines]{w{c}{.225cm}w{c}{.225cm}w{c}{.225cm}w{c}{.225cm}w{c}{.225cm}w{c}{.225cm}w{c}{.225cm}w{c}{.225cm}w{c}{.225cm}w{c}{.225cm}w{c}{.225cm}w{c}{.225cm}w{c}{.225cm}w{c}{.225cm}w{c}{.225cm}w{c}{.25cm}}
\CodeBefore
  \cellcolor{green!30!white}{2-2,2-4,2-7,2-10,2-12,2-15}
  \cellcolor{blue!30!white} {2-3,2-6,2-8,2-11,2-14,2-16}
  % \cellcolor{cyan!30!white}{2-7,2-8}
  % \cellcolor{green!30!yellow}{2-10,2-11}
  % \cellcolor{blue!30!yellow}{2-12,2-14}
  % \cellcolor{cyan!30!yellow}{2-15,2-16}
  \cellcolor{yellow}{3-2,3-3,3-4,3-6,3-7,3-8,3-10,3-11,3-12,3-14,3-15,3-16}
  \cellcolor{red!30!white}{4-1,4-5,4-9,4-13}
  \cellcolor{DarkOliveGreen!30!white}{5-1,5-5,5-9,5-13}
\Body
U &
0 & 1 & 2 & 3 & 4 & 5 & 6 & 7 &
8 & 9 & 10 & 11 & 12 & 13 & 14 & 15
\\
6U \div 16 &
0 & 0 & 0 & 1 & 1 & 1 & 2 & 2 &
3 & 3 & 3 & 4 & 4 & 4 & 5 & 5
\\
% 6U \bmod 16 &
% 0 & 6 & 12 & 2 & 8 & 14 & 4 & 10 &
% 0 & 6 & 12 & 2 & 8 & 14 & 4 & 10
% \\
2 (6U \bmod 16) \div 16 \equiv (12U \div 16) \bmod 2 &
0 & 0 & 1 & 0 & 1 & 1 & 0 & 1 &
0 & 0 & 1 & 0 & 1 & 1 & 0 & 1
\\
2 (6U \bmod 16) \bmod 16 \equiv 12U \bmod 16 &
0 & 12 & 8 & 4 & 0 & 12 & 8 & 4 &
0 & 12 & 8 & 4 & 0 & 12 & 8 & 4
\\
U - (12U \div 16) &
0 & 1 & 1 & 1 & 1 & 2 & 2 & 2 &
2 & 3 & 3 & 3 & 3 & 4 & 4 & 4
\\
\end{NiceArray}$
\end{adjustbox}
\end{align*}
\caption{Illustration of \cref{alg:uniform-lemire} to generate a uniform
with range $n=6$ given a uniform random word $U$ of length $W=4$ bits,
including the operations required for randomness recycling.}
\label{fig:brackett}
\end{figure}

%% file: alg-uniform-brackett.tex
%!TEX root=./main.tex
\begin{listing}[!t]

% LEMIRE
\begin{minipage}[t]{\linewidth}
\begin{algorithm}[H]
\caption{Uniform sampling with widening multiplication and $\Call{RecycleWidening}{}$}
\label{alg:uniform-lemire}
\begin{algorithmic}[1]
\Require{Integer $n \in [1, 2^W)$, where $W$ is the word size from \cref{alg:recycle}.}
\Ensure{Random sample $U \sim \Uniform[0,n)$}
\Procedure{UniformLemire}{$n$}
\State $(q, t) \gets \Call{DivMod}{2^W,n}$ \Comment{${2^W} = q \cdot n + t$}
\While{\textbf{true}}
  \State $X \gets \Call{Flip}{W}$ \Comment{$X \sim \Uniform[0,2^W)$; uniform random word}
  \State $(U, r) \gets X \otimes n$ \Comment{widening multiplication; $U \approxsim \Uniform[0,n)$}
  \State $(U', r') \gets q \otimes r$ \Comment{widening multiplication; $U' \approxsim \Uniform[0,q)$}
  \If{$r' \geq t$}
    \Comment{accept case}
    \State $\Call{RecycleWidening}{U', q}$ \Comment{(\cref{alg:recycle-widening})}
    \State \Return $U$ \Comment{$U \sim \Uniform[0,n)$}
  \Else
    \Comment{reject case}
    \State $(U'', r'') \gets X \otimes (n \cdot q)$ \Comment{widening multiplication}
    \State $\Call{RecycleWidening}{X - U'', t}$ \Comment{(\cref{alg:recycle-widening})}
  \EndIf
\EndWhile
\EndProcedure
\end{algorithmic}
\end{algorithm}
\end{minipage}

% BRACKETT-ROZINSKY
\begin{minipage}[t]{\linewidth}
\begin{algorithm}[H]
\caption{Batched uniform sampling with widening multiplication and $\Call{RecycleWidening}{}$}
\label{alg:uniform-brackett}
\begin{algorithmic}[1]
\Require{Positive integers $n_i > 0$ with $\prod_{i=1}^k n_i < 2^W$, where $W$ is the word size from \cref{alg:recycle}.}
\Ensure{Independent random samples $U_i \sim \Uniform[0,n_i)$ for $i \in \set{1, \dots, k}$}
\Procedure{UniformBrackett}{$n_1, \ldots, n_k$}
\State $n \gets \prod_{i=1}^k n_i$ \Comment{total target outcomes}
\State $(n_{k+1}, t) \gets \Call{DivMod}{2^W, n}$ \Comment{$2^W = t + \prod_{i=1}^{k+1} n_i$}
\While{\textbf{true}}
  \State $X \gets \Call{Flip}{W}$ \Comment{$X \sim \Uniform[0,2^W)$; uniform random word}
  \State $r \gets X$ \Comment{copy of $X$ to allow recycling}
  \For{$i \gets 1, \ldots, k+1$}
    \State $(U_i, r) \gets r \otimes n_i$ \Comment{widening multiplication; $U_i \approxsim \Uniform[0,n_i)$}
  \EndFor
  \If{$r \geq t$}
    \Comment{accept case}
    \State $\Call{RecycleWidening}{U_{k+1}, n_{k+1}}$ \Comment{(\cref{alg:recycle-widening})}
    \State \Return $U_1, \ldots, U_k$ \Comment{$U_i \sim \Uniform[0,n_i)$}
  \Else
    \Comment{reject case}
    \State $(U, r) \gets X \otimes (n \cdot n_{k+1})$ \Comment{widening multiplication}
    \State $\Call{RecycleWidening}{X - U, t}$ \Comment{(\cref{alg:recycle-widening})}
  \EndIf
\EndWhile
\EndProcedure
\end{algorithmic}
\end{algorithm}
\end{minipage}
\end{listing}

%% file: fig-benchmark-uniform-full.tex
%!TEX root=./main.tex
\begin{figure}[!p]
\centering

\includegraphics[width=.925\linewidth]{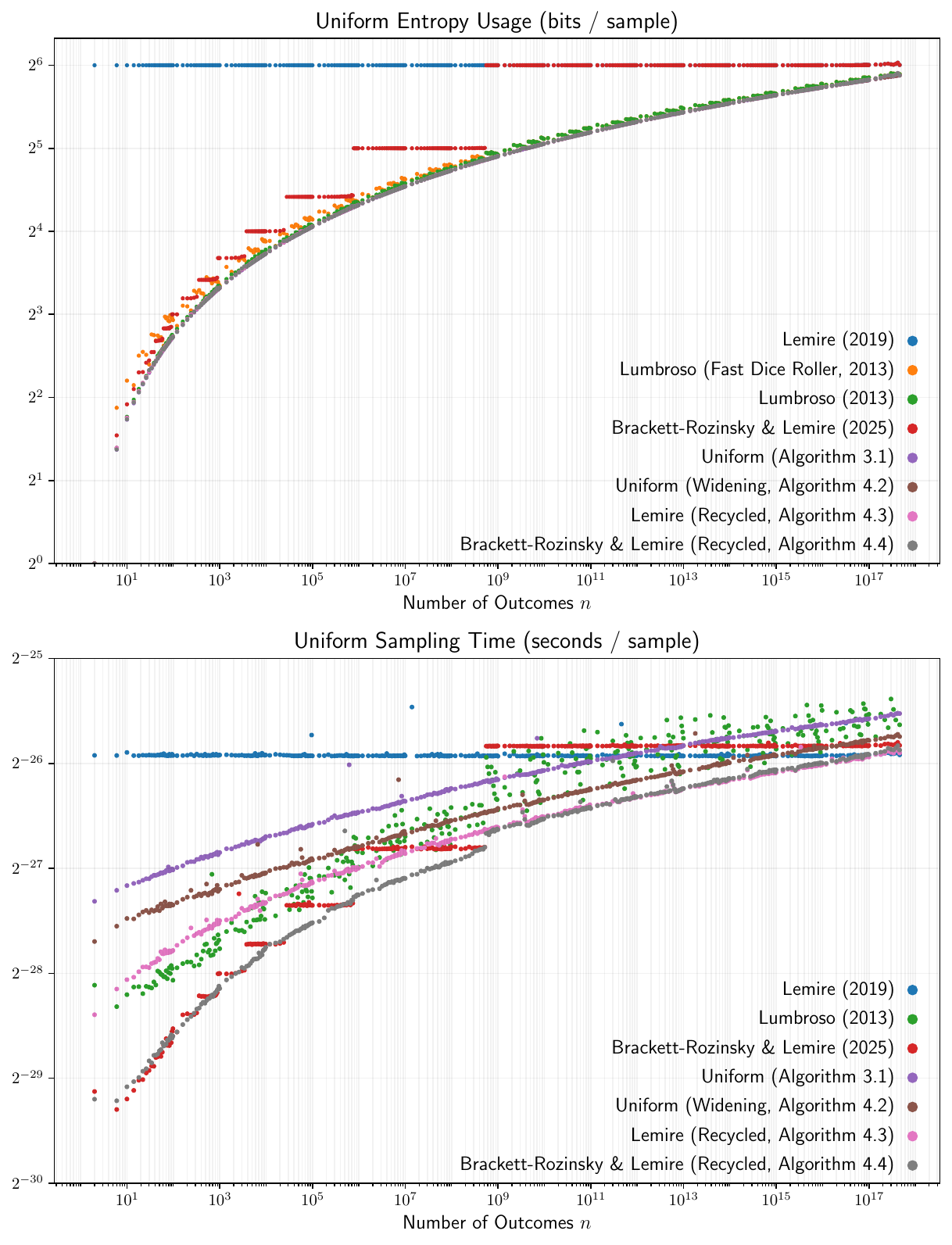}

\captionsetup{skip=8pt}
\caption{Benchmark comparison of entropy consumption and sampling time
using various optimized sampling algorithms for discrete uniforms.
\Crefrange{alg:uniform-widening}{alg:uniform-brackett} are novel to this work.
}
\label{fig:benchmark-uniform-full}
\end{figure}

%% file: fig-shuffle.tex
%!TEX root=./main.tex
\begin{figure}[t]
\includegraphics[width=\linewidth]{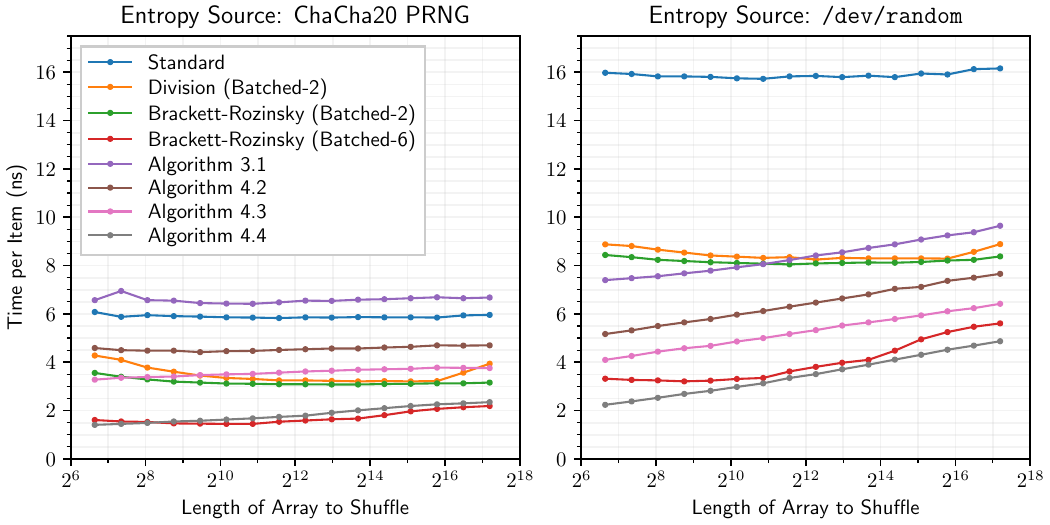}
\caption{Runtime of \citet{fisher1953} shuffle for generating random
permutations, using eight different sampling algorithms for discrete uniforms
and two different entropy sources (ChaCha20 and \texttt{/dev/random}).
The first four samplers in the legend are from the benchmark
set in \citet{brackett2025}.
The remaining samplers are our
\cref{alg:uniform,alg:uniform-widening,alg:uniform-lemire,alg:uniform-brackett};
the last three use randomness recycling techniques that are optimized for
word-based operations.}
\label{fig:shuffle}
\end{figure}

%% file: sec-nonuniform.tex
%!TEX root=./main.tex

\section{Randomness Recylers for General Distributions}
\label{sec:rr-general}

In \cref{sec:analysis-nonuniform} we presented a randomness recycling
method for sampling general (nonuniform) distributions using inversion
sampling.
Randomness recycling can also be used with many other samplers for general
distributions, summarized in \cref{table:algorithms}.
We describe these techniques here.
Whereas the proof of \cref{theorem:entropy-cost} in
\cref{sec:analysis-theorem} used \cref{alg:inversion} in the witness, in
principle any one of the forthcoming algorithms could be used instead.
When used to generate a single sample, the algorithms in this section may
be far from entropy optimal.
When used to generate a random sequence, however, they achieve nearly optimal
amortized entropy in the sense of \cref{theorem:entropy-cost},
by virtue of the randomness recycling rules.

\subsection{Lookup-Table Sampling}
\label{sec:rr-general-lookup}

For a rational discrete distribution $\dist \defeq (a_0, \dots, a_{n-1})/A$
with moderately sized weight sum $A$, the lookup-table method~\citep[p.~770]{devroye1986}
is a practical way to eliminate the logarithmic
cost of binary search in inversion sampling.
A table $T[0 \twodots A-1]$ of size $A$ is constructed such that
each integer $i \in [0,n)$ is stored exactly $a_i$ times.
To generate a sample from $\dist$, a uniform variate
$U \sim \Uniform[0,A)$ is generated and then
$X \gets T[U]$ is returned.
The randomness recycling strategy is analogous to that of the inversion
sampler.

\begin{align*}
\NiceMatrixOptions{custom-line={letter = I, tikz = {line width=2pt, blue }}}
\begin{NiceArray}[first-col,hvlines]{ccccIccccIcIcccc}
U
&
\cA[0] & \cA[1] & \cA[\dots] & \cA[A_0-1]
&
\cB[A_0] & \cB[A_0+1] & \cB[\dots] & \cB[A_1-1]
&
\cC[\dots]
&
\cD[A_{n-2}] & \cD[A_{n-2}+1] & \cD[\dots] & \cD[A_{n-1}-1]
\\
T
&
0 & 0 & \dots & 0
&
1 & 1 & \dots & 1
&
\dots
&
n-1 & n-1 & \dots & n-1
\CodeAfter
  \UnderBrace[shorten,yshift=4pt]{1-1}{2-4}{X=0}
  \UnderBrace[shorten,yshift=4pt]{1-5}{2-8}{X=1}
  \UnderBrace[shorten,yshift=4pt]{1-10}{2-13}{X=n-1}
  \OverBrace[shorten,yshift=4pt]{1-1}{1-4}{a_0}
  \OverBrace[shorten,yshift=4pt]{1-5}{1-8}{a_1}
  \OverBrace[shorten,yshift=4pt]{1-10}{1-13}{a_{n-1}}
\end{NiceArray}\,.
\\[5pt]
\end{align*}
The top array shows all possible values of $U\sim \Uniform[0,A_{n-1})$
and the bottom array shows the lookup table $T$.
Since $T$ has length $A \equiv A_{n-1}$ and must be stored in memory,
lookup-table sampling scales exponentially with the number of bits needed
to encode $\dist$.

\input{alg-lookup}

\subsection{Alias Sampling}
\label{sec:rr-general-alias}

The alias method~\citep{walker1977} is a state-of-the-art sampler that
avoids the exponential space complexity of the lookup-table method while
retaining its extremely fast runtime.
In the preprocessing phase,
the target distribution $\dist = (a_0, \dots, a_{n-1}) / b$
is used to compute \begin{enumerate*}[label=(\roman*)]
\item the alias outcomes $\mathbf{z} \defeq (z_0, \dots, z_{n-1})$, where $z_i \in [0,n-1] \setminus \set{i}$; and
\item the ``no alias'' odds $\mathbf{w} \defeq (w_0, \dots, w_{n-1})$, where $w_i \in [0,b]$
\end{enumerate*}
\citep{vose1991}.
Equipped with these data structures, the generation phase is as follows:
\begin{itemize}[wide,leftmargin=*,noitemsep]
\item Draw a uniform index $I \sim \Uniform[0,n)$.
\item Draw $B \sim \Bernoulli(w_I/b)$;
  if $B=1$ then return $X \gets I$, else return $X \gets z_I$.
\end{itemize}
For any discrete distribution $\dist$, it is always possible to construct an
alias data structure that guarantees the return value $X$ is distributed
according to $\dist$ \citep[Theorem 4.1]{devroye1986}.
That is, for a given outcome $j \in [0,n)$, let $k_1, \dots, k_m$
be the indices for which $j$ is an alias (i.e., $z_{k_i} = j$, for $i=1,\dots,m$).
The alias method guarantees that
\begin{align}
\Prob(X = i) \defeq 1/n\cdot w_j/b + \sum_{i=1}^{m} 1/n \cdot (b-w_{k_i})/b = a_j/b
&& ( 0 \le j < n).
\end{align}

A na\"ive approach to randomness recycling with the alias sampler
is to sample $I \gets \Call{Uniform}{n}$ using \cref{alg:uniform}
and $B \gets \Call{Inversion}{(w_I, b - w_I)}$ using \cref{alg:inversion}.
However, this approach does not optimally recycle randomness.
Instead, the optimal recycling rule is as follows:
\begin{itemize}[wide,leftmargin=*]
  \item Sample $U \sim \Uniform[0,nb)$ using \cref{alg:uniform}.
  \item Write $U = q_U b + r_U$,
    where $q_U \defeq \floor{U/b}$ and $r_U \defeq U \bmod b$.
    Here $I \equiv q_U \sim \Uniform[0,n)$ furnishes a uniformly
    chosen index in $[0,n)$ and $r_U \sim \Uniform[0,b)$ furnishes the uniform
    needed to generate the $\Bernoulli(w_I/b)$ variable.
  \item If $r_U < w_{I}$, then return $X \gets I$,
  else return $X \gets z_{I}$; and recycle a uniform state $(\randS',\randM')$
  with $\randM' = n a_X$ by
  using the auxiliary array $\mathbf{c}$ described in the remainder of this section.
\end{itemize}
Conditional on the event $\set{X=i}$, the selected cell in the alias table
is uniform over all $n a_i$ ways the label $i$ could have been achieved.
Applying \cref{prop:merge-nonuniform} gives the recycling rule, which
is visualized in the following example.

\input{fig-alias}

Let
$\dist \defeq \set{
  \mathrm{A} \mapsto 4 \eqdef a_0,
  \mathrm{B} \mapsto 3 \eqdef a_1,
  \mathrm{C} \mapsto 3 \eqdef a_2,
  \mathrm{D} \mapsto 1 \eqdef a_3
  } / 11$
be the target distribution, where we use the symbols
$\mathrm{A}=0, \mathrm{B}=1, \mathrm{C}=2, \mathrm{D}=3$ to avoid
confusion between the outcome labels and other
integers in the alias data structures.
\Cref{fig:alias} shows the corresponding data structures for alias sampling
with randomness recycling, where $*$ denotes an arbitrary value.
Note that the (exponentially sized) alias table itself is never stored in memory,
but is represented compactly as $(\mathbf{z}, \mathbf{w}, \mathbf{a}, \mathbf{c})$.
In this example, consider two cases:
\begin{itemize}[wide,leftmargin=*]
\item Suppose that $(q_U, r_U) = (1,2)$ specifies the cell $(\mathrm{B}, 2)$.
Then conditioned on $\set{X = \mathrm{B}}$, the value $2$ is uniformly
distributed over $[0,12)$, i.e., the indices of all other cells labeled
$\mathrm{B}$ in the table.
We can directly recycle $(\randS', \randM') = (2, 12)$ into the global
uniform state, where $\randS' = r_U = 2$ and $\randM'=n a_1 = 12$.

\item Suppose that $(q_U,r_U) = (1,7)$ specifies the
cell $(\mathrm{A}, 12)$.
Then conditioned on $\set{X = \mathrm{A}}$, the value $12$ is uniformly
distributed over $[0,16)$, i.e., the indices of all other cells labeled
$\mathrm{A}$ in the table.
We can thus recycle $(\randS', \randM') = (12, 16)$ into the uniform state.
To compute the value $\randS' = 12$ from $(q_U, r_U) = (1,7)$, we first add
the number of $\mathrm{A}$ cells to the left of the current column
($11$), then subtract the number of $\mathrm{B}$ cells in the current
column $(6)$ to obtain the offset $c_1=11-6=5$, and finally add $r_U = 7$ to obtain
$\randS'=c_1+r_U=5+7=12$.
The bound $\randM'=16 = n a_{z_1}$ is available as in the previous case.
\end{itemize}

To enhance the alias sampler with this randomness recycling rule, we
construct an array $\mathbf{c} \defeq (c_0, \dots, c_{n-1})$ of the offsets
needed to perform recycling as described above, which can be done in linear
time.
\Cref{alg:alias} shows the resulting alias sampler with randomness recycling.

\input{alg-alias}

\subsection{Discrete Distribution Generating Tree Sampling}
\label{sec:rr-general-ddg}

\input{fig-recycle-ddg}

A discrete distribution generating (DDG) tree is a universal computational
model introduced by \citet{knuth1976} for describing any computable
sampling algorithm that maps random bits to discrete outcomes.
A DDG tree $G$ is a complete, rooted binary tree where each leaf node
has an outcome label $i \in \Nat$.
DDG tree sampling operates as follows, starting from the root of $G$:
\begin{enumerate}[wide, label=(D\arabic*), leftmargin=*,]
  \item\label{ddg:flip} Obtain a fair coin toss $B \gets \Call{Flip}{1}$. If $B=0$,
    then visit the left child of the current node; else if $B=1$, then
    visit the right child.
  \item\label{ddg:decide} If the visited child node is a leaf node, then
  return its label; else go to \labelcref{ddg:flip}.
\end{enumerate}
The set of leaf nodes in a DDG tree $G$ is denoted $\mathcal{L}(G)$.
The depth and label of any leaf $l \in \mathcal{L}(G)$,
are denoted $d(l) \ge 0$ and $\ell(l) \in \Nat$, respectively.
With these notations, a random variable $X \sim G$ obtained
by DDG tree sampling~\labelcrefrange{ddg:flip}{ddg:decide}
has the following distribution:
\begin{align}
\Prob(X = i) = \sum_{l \in \mathcal{L}(G)} 2^{-d(l)} \cdot \mathbb{I}[\ell(l) = i].
\label{eq:ddg-output-dist}
\end{align}
The expected number of coin tosses used to sample $X$ is equal to the
average depth of a leaf:
\begin{align}
\expect{T_G} = \sum_{l \in \mathcal{L}(G)} 2^{-d(l)} \cdot d(l).
\label{eq:ddg-entropy-cost}
\end{align}

As discussed in \cref{sec:intro-existing}, \citet{knuth1976} show how to
construct an entropy-optimal DDG tree $G^*$ for any distribution distribution
$\dist \defeq (a_0, \dots, a_{n-1}) / A$ whose expected number of coin tosses
\cref{eq:ddg-entropy-cost} is the least possible.
This optimal tree $G^*$ is constructed by placing a leaf labeled $i$
at depth $j$ if and only if the $j$th bit in the binary expansion of
$a_i/A$ is 1.
Explicitly constructing $G^*$ can require exponential space in the
number of bits needed to encode $\dist$ \citep[Theorem 3.6]{saad2020popl}.
An alternative approach is to incrementally traverse $G^*$ without
explicitly constructing the tree
(e.g., \citet[p.~384]{knuth1976}, \citet[Algorithm 1]{saad2025}),
although steps \labelcref{ddg:flip,ddg:decide} become more complicated
as the leaves and labels must be created during sampling.

The Fast and Amplified Loaded Dice Roller~\citep{saad2020fldr,draper2025}
algorithms are near entropy-optimal DDG tree samplers that use rejection
sampling to reduce the space complexity of explicitly stored DDG trees.
These methods can be interpreted as a compression of the lookup-table method
from \cref{sec:rr-general-lookup} using power-of-two block sizes.
The key idea is to build an entropy-optimal DDG tree $G'$ for
$\dist' \defeq (a_0, \dots, a_{n-1}, a_n) / 2^k$
where $2^k$ is a power of two that is larger than the sum of weights $A$ and
$a_n = 2^k - A$ is a ``reject'' outcome.
\Citet{saad2020popl} and \citet{draper2025} show that with a linearithmic
sized tree $G'$ it is possible to achieve expected costs less than $\entropy{\dist}+6$
and $\entropy{\dist}+2$, respectively, where the latter coincides with the entropy-optimal
\citet{knuth1976} range $[\entropy{\dist}, \entropy{\dist}+2)$.
\Cref{fig:recycle-ddg-naive} shows an example DDG tree $G$, where the
colors of the leaf nodes are suggestive of the recycling strategy.

\input{alg-ddg}

\paragraph{Recycling Rule}

The recycling rule for DDG tree sampling extracts
a uniform state $(\randS', \randM')$ from a nonuniform state
using \cref{prop:merge-nonuniform}.
We recycle a draw from the distribution of \textit{depth}
of the leaf nodes conditioned on the
\textit{label} of the visited leaf node, which gives a nonuniform distribution
$\Discrete(\mathbf{w})$ as in \cref{prop:merge-nonuniform} whose weights
$w_i$ are distinct powers of two corresponding to the set bits in the
binary expansion of the label probability.

More specifically, consider any entropy-optimal DDG tree $G$ whose output
distribution is $\dist = (a_0, \dots, a_{n-1}, a_n)/2^k$, where
$a_n$ is a
``reject'' outcome that implicitly denotes a back-edge to the root.
Each leaf node in $G$ has a label $i \in [0,n]$.
Let $c_i \ge 1$ denote the (finite)
number of leaves in $G$ with label $i$.
Further, let $0 \le d_{i1} < d_{i2} < \dots < d_{ic_i}$ denote the
depths of the leaves with label $i$, which are distinct since $G$
is entropy optimal.
Conditioned on $\set{X=i}$, the distribution over
the possible leaves with label $i$ is
\begin{equation}
\mathbf{w}_i = (w_{i1}, w_{i2}, \dots, w_{ic_i})
= (2^{d_{i c_i}-d_{i1}}, 2^{d_{i c_i}-d_{i2}} \dots, 1)
\propto
(2^{-d_{i1}}, 2^{-d_{i2}} \dots, 2^{-d_{ic_i}})
\end{equation}
Conditioned on returning $\set{X=i}$, DDG tree sampling selects one
of these $c_i$ leaves as an exact nonuniform draw $X' \sim \Discrete(\mathbf{w}_i)$.
For example, in \cref{fig:recycle-ddg-naive}, the (unnormalized)
distribution over depths given $\set{X = \mathrm{A}}$ is given by
$\mathbf{w}_{\mathrm{A}} = (4, 1) \propto (2^{-2}, 2^{-4})$.
The randomness recycling rule is immediate from
\cref{prop:merge-nonuniform}, which shows how to recycle
the nonuniform state $X'$ using a fresh draw $U \sim \Uniform[0,w_{X'})$.
The challenge with this approach is that it is difficult to guarantee
the $W$-bit global uniform state will not overflow,
when using \cref{alg:recycle}
(although \cref{alg:recycle-widening} can provide an alternative
solution to this concern).

\paragraph{Efficient Implementation}
To avoid explicitly merging a nonuniform state,
it is more convenient to
implement randomness recycling on a (hypothetical) DDG tree $G'$
(\cref{fig:recycle-ddg})
whose leaves live at the same level, as follows:
\begin{itemize}[noitemsep]
\item Replace each leaf node in $G$ at depth $d$ with a subtree that
  terminates at the maximum depth of $G$; and set all
  the labels of the new leaf nodes to be the same label as the original leaf.
\item Pack all the non-reject outcomes to the left of the resulting tree.
\item Use $U \sim \Uniform[0,A)$ to select one of the $A$ leaves at the final level.
\end{itemize}
\Cref{alg:ddg} shows the corresponding algorithm for sampling a
``left-packed'' DDG tree of this form whose leaves are all at the same level.
Because the leaves in \cref{fig:recycle-ddg} are all at the same level,
the recycling rule reverts to the typical uniform case as in the inversion,
lookup table, and alias sampler implementations from the previous sections, rather than
the nonuniform case using \cref{prop:merge-nonuniform} if the tree were
of the form in \cref{fig:recycle-ddg-naive}.

\Cref{alg:ddg} uses fast bit operations to identify the sampled leaf $X$
and implement the recycling rule without actually
constructing the hypothetical tree $G'$.
In particular, when a leaf node with label $X$ is encountered (\cref{algline:fldr-hit-leaf}),
a uniform $\randS' \sim \Uniform[0,\randM')$ with $\randM' \defeq a_X$ is extracted.
The value of $\randS'$ is computed by first choosing one of the ``hypothetical''
leaf nodes below the current leaf (dashed edges in \cref{fig:recycle-ddg})
using $U \bmod 2^{k-d}$, and then computing the offset
$a_X - (a_X \bmod 2^{k+1-d})$ which is the sum of all hypothetical
leaves with label $X$ to the left.

\input{fig-benchmark}

\begin{remark}
\label{remark:uniform-power-of-two}
For a distribution whose sum $A$ of weights is very close to a power of two,
as for ALDR \citep{draper2025} or any dyadic distribution,
the method on \cref{algline:fldr-uniform} of \cref{alg:ddg} to
sample a uniform over $A$ can be replaced by a method using
rejection sampling via a uniform over $2^{\ceil{\log(A)}}$.
Although this change slightly increases the entropy cost, it may improve the runtime,
because the division operations in \cref{alg:uniform} can be replaced by
bitwise operations when the target uniform range is a power of two.
\end{remark}

\subsection{Benchmark Evaluations}
\label{sec:rr-general-evaluation}

\Cref{fig:benchmark} shows how the randomness recycling strategies
described in this section yield improvements in entropy cost and wall-clock
runtime as compared to the baseline versions without randomness recycling.
Each row in \cref{fig:benchmark} shows the results for a specific sampling
algorithm.
The first column shows the preprocessing time in seconds, the second
column shows the entropy usage in bits/sample, and the third column shows
the sampling time in seconds/sample.
In each panel, every dot shows the measurements for a given
probability distribution over $n$ outcomes (x-axis) with sum of weights $m=1 000 000$,
amortized over an i.i.d.~sequence of one million samples.
In the majority of cases, randomness recycling introduces minimal runtime
overhead in the preprocessing time, while enabling lower entropy
consumption and sampling time.
A software library in the C programming language containing the algorithms described in this paper is available
at \url{https://github.com/probsys/randomness-recycling}.

\subsection{Discrete Gaussian Sampling}
\label{sec:rr-general-discrete-gaussian}

\input{fig-gaussian}

The discrete Gaussian sampler of \citet{canonne2020} is a prominent example
of an exact sampler over the integers.
The algorithm samples a discrete Gaussian by calling two primitive
samplers: discrete uniform and Bernoulli.
We adapted the reference Python implementation by the authors
(available at \url{https://github.com/IBM/discrete-gaussian-differential-privacy})
by replacing all the calls to uniform and Bernoulli with their randomness-recycled
variants from \cref{alg:uniform,alg:inversion}.
\Cref{fig:benchmark-gaussian} shows that randomness recycling reduces the
entropy cost of the algorithm by up to 10x compared to the original
version.

%% file: alg-lookup.tex
%!TEX root=./main.tex
\begin{algorithm}
\caption{Lookup table sampling with randomness recycling}
\label{alg:lookup}
\begin{algorithmic}[1]
\Require{Positive integers $a_0, \dots, a_{n-1}$ with sum $A$}
\Ensure{Random sample $X \sim \Discrete(a_0,\dots,a_{n-1})$}
\Procedure{Lookup}{$a_0$,$\dots$,$a_{n-1}$}
\State Set $A_{-1} \gets 0$
  and $A_i \gets a_0 + \dots + a_i$ for $i=0,\dots,{n-1}$
                                                     \Comment{prefix sums (if not given)}
\State Construct array $T[0\twodots A-1]$
  where $i \in [0,n)$ is stored $a_i$ times.         \Comment{lookup table}
\State $U \sim \Call{Uniform}{A}$                    \Comment{draw uniform variate (\cref{alg:uniform})}
\State Let $X \gets T[U]$                            \Comment{lookup}
\State $(\randS', \randM') \gets (U-A_{i-1}, a_i)$   \Comment{extract uniform state}
\State $\Call{Recycle}{\randS', \randM'}$ \Comment{(\cref{alg:recycle})}
\State \Return $X$
\EndProcedure
\end{algorithmic}
\end{algorithm}

%% file: fig-alias.tex
%!TEX root=./main.tex
\begin{figure}[ht]
\begin{align*}
\tikzset{nicematrix/brace/.append style= {thick, decoration = { brace , raise = -0.5 em }}}
\begin{NiceArray}[last-row,first-col]{c|c|c|c|c|}
\cline{2-5}
% \Vbrace{11}{{\textstyle r_U}\hspace{-0.1cm}}
  & 10 & (\cA,10) & (\cA,15) & (\cB,11) & (\cC,11) \\ \cline{2-5}
~ & 9  & (\cA,9)  & (\cA,14) & (\cB,10) & (\cC,10) \\ \cline{2-5}
~ & 8  & (\cA,8)  & (\cA,13) & (\cB,9)  & (\cC,9) \\ \cline{2-5}
~ & 7  & (\cA,7)  & (\cA,12) & (\cB,8)  & (\cC,8) \\ \cline{2-5}
~ & 6  & (\cA,6)  & (\cA,11) & (\cB,7)  & (\cC,7) \\ \cline{2-5}
~ & 5  & (\cA,5)  & (\cB,5)  & (\cB,6)  & (\cC,6) \\ \cline{2-5}
~ & 4  & (\cA,4)  & (\cB,4)  & (\cC,4)  & (\cC,5) \\ \cline{2-5}
~ & 3  & (\cA,3)  & (\cB,3)  & (\cC,3)  & (\cD,3) \\ \cline{2-5}
~ & 2  & (\cA,2)  & (\cB,2)  & (\cC,2)  & (\cD,2) \\ \cline{2-5}
~ & 1  & (\cA,1)  & (\cB,1)  & (\cC,1)  & (\cD,1) \\ \cline{2-5}
~ & 0  & (\cA,0)  & (\cB,0)  & (\cC,0)  & (\cD,0) \\ \cline{2-5}
~ & ~  & 0        & 1        & 2        & 3
\CodeAfter
  \UnderBrace{12-2}{12-5}{q_U}
  \tikz[thick] \draw[{thick,decorate, decoration={brace,mirror,amplitude=.2cm}}] (1-|1) -- node[pos=0.5,left,xshift=-0.15cm]{$r_U$} (12-|1);
\end{NiceArray}
&&
\begin{aligned}
\mathbf{z} && z_0 &= *,  & z_1 &= \mathrm{A}, & z_2 &= \mathrm{B}, & z_3 &= \mathrm{C} \\
\mathbf{w} && w_0 &= 11, & w_1 &= 6,          & w_2 &= 5,          & w_3 &= 4 \\
\mathbf{a} && a_0 &= 4,  & a_1 &= 3,          & a_2 &= 3,          & a_3 &= 1 \\
\mathbf{c} && c_0 &= *,  & c_1 &= 5,          & c_2 &= 1,          & c_3 &= 1 \\
\end{aligned}
\end{align*}
\caption{Randomness recycling with the alias sampler
for the target distribution
$\dist \defeq \set{
  \mathrm{A} \mapsto 4 \eqdef a_0,
  \mathrm{B} \mapsto 3 \eqdef a_1,
  \mathrm{C} \mapsto 3 \eqdef a_2,
  \mathrm{D} \mapsto 1 \eqdef a_3
  } / 11$.
}
\label{fig:alias}
\end{figure}

%% file: alg-alias.tex
%!TEX root=./main.tex
\begin{algorithm}[ht]
\caption{Alias sampling with randomness recycling}
\label{alg:alias}
\begin{algorithmic}[1]
\Require{%
  Positive integers $a_0, \dots, a_{n-1}$ with sum $A$\\
  Aliases $\mathbf{z} \defeq (z_0, \dots, z_{n-1})$ \\
  No alias odds $\mathbf{w} \defeq (w_0, \dots, w_{n-1})$ \\
  Recycling offsets $\mathbf{c} \defeq (c_0, \dots, c_{n-1})$
}
\Ensure{Random sample $X \sim \Discrete(a_0,\dots,a_{n-1})$}
\Procedure{Alias}{$(a_0,\dots,a_{n-1})$, $\mathbf{z}$, $\mathbf{w}$, $\mathbf{c}$}
\State $U \gets \Call{Uniform}{A\cdot n}$                     \Comment{draw uniform variate (\cref{alg:uniform})}
\State $(q, r) \gets \Call{DivMod}{U, A}$                     \Comment{compute cell $(q,r)$ in alias table}
\If{$r < w_{q}$}                                              \Comment{no alias: select outcome $q$}
  \State $(\randS', \randM') \gets (r, n \cdot a_{q})$        \Comment{extract uniform state}
  \State $\Call{Recycle}{\randS', \randM'}$                   \Comment{(\cref{alg:recycle})}
  \State \Return $q$
\Else                                                         \Comment{alias: select outcome $z_q$}
  \State $(\randS',\randM') \gets (r + c_q, n \cdot a_{z_q})$ \Comment{extract uniform state}
  \State $\Call{Recycle}{\randS',\randM'}$                    \Comment{(\cref{alg:recycle})}
  \State \Return $z_{q}$
\EndIf
\EndProcedure
\end{algorithmic}
\end{algorithm}

%% file: fig-recycle-ddg.tex
%!TEX root=./main.tex
\begin{figure}[t]
\centering
\begin{adjustbox}{max width=.85\textwidth}
\newcommand{\ddd}{\edge[densely dashdotted]}
\tikzset{every tree node/.style={anchor=north}}
\tikzset{level distance=20pt}
\tikzset{sibling distance=2pt}
\begin{tikzpicture}
\Tree[
  [
    \node[label={below:(A,0)},draw=black,circle,fill=\clrA]{};
    \node[label={below:(B,0)},draw=black,circle,fill=\clrB]{};
  ]
  [
    [
      \node[label={below:(C,0)},draw=black,circle,fill=\clrC]{};
      [
        \node[label={below:(A,1)},draw=black,circle,fill=\clrA]{};
        \node[label={below:(D,0)},draw=black,circle,fill=\clrD]{};
      ]
    ]
    \node[label={below:(R,0)},draw=black,circle,fill=gray]{};
  ]
]
\end{tikzpicture}
\end{adjustbox}
\caption{
  DDG tree for the target distribution
  $\dist \defeq \set{
    \mathrm{A} \mapsto 5,
    \mathrm{B} \mapsto 4,
    \mathrm{C} \mapsto 2,
    \mathrm{D} \mapsto 1
    } / 12$ using the rejection-based Fast Loaded Dice Roller \citep{saad2020fldr} method.
}
\label{fig:recycle-ddg-naive}

\bigskip

\begin{adjustbox}{max width=.85\textwidth}
\tikzset{dd/.style={densely dashdotted}}
\tikzset{ll/.style n args={2}{draw=black,circle,fill=#1,label={below:#2}}}
\tikzset{every tree node/.style={anchor=north}}
\tikzset{level distance=20pt}
\tikzset{sibling distance=2pt}
\begin{tikzpicture}
\Tree[
  [
    [.A
      \edge[dd]; [ \edge[dd]; [ \edge[dd]; \node[name=start, ll={\clrA}{(A,0)}]{}; ] \edge[dd]; [ \edge[dd]; \node[ll={\clrA}{(A,1)}]{}; ] ]
      \edge[dd]; [ \edge[dd]; [ \edge[dd]; \node[ll={\clrA}{(A,2)}]{}; ] \edge[dd]; [ \edge[dd]; \node[ll={\clrA}{(A,3)}]{}; ] ]
    ]
    [.B
      \edge[dd]; [ \edge[dd]; [ \edge[dd]; \node[ll={\clrB}{(B,0)}]{}; ] \edge[dd]; [ \edge[dd]; \node[ll={\clrB}{(B,1)}]{}; ] ]
      \edge[dd]; [ \edge[dd]; [ \edge[dd]; \node[ll={\clrB}{(B,2)}]{}; ] \edge[dd]; [ \edge[dd]; \node[ll={\clrB}{(B,3)}]{}; ] ]
    ]
  ]
  [
    [
      [.C \edge[dd]; [ \edge[dd]; \node[ll={\clrC}{(C,0)}]{}; ] \edge[dd]; [ \edge[dd]; \node[ll={\clrC}{(C,1)}]{}; ] ]
      [
        [.A \edge[dd]; \node[ll={\clrA}{(A,4)}]{}; ]
        [.D \edge[dd]; \node[name=end,ll={\clrD}{(D,0)}]{}; ]
      ]
    ]
    R
  ]
]

\draw[{thick,decorate, decoration={brace, amplitude=.5cm, raise=.75cm}}]
  (end.east)
  -- node[pos=0.5,below,yshift=-1.25cm]{$U \sim \Uniform[0,A = 12)$}
  (start.west);
\end{tikzpicture}
\end{adjustbox}
\caption{Randomness recycling for the distribution in
\cref{fig:recycle-ddg-naive} by using a ``left-packed'' DDG tree whose
leaves all live at the same level. This (exponentially sized) tree can be
compactly represented by using the data structures listed in \cref{alg:ddg}.}
\label{fig:recycle-ddg}
\end{figure}

%% file: alg-ddg.tex
%!TEX root=./main.tex
\begin{algorithm}[t]
\caption{DDG tree sampling with randomness recycling}
\label{alg:ddg}
\begin{algorithmic}[1]
\Require{%
Target distribution $(a_0, \ldots, a_{n-1})$ with sum $A$ \\
DDG tree depth $k \ge 0$ \\
Leaf counts per level $\mathbf{L} \defeq (L_0, \ldots, L_k)$ \\
Leaf labels $\mathbf{H} \defeq ((H_{0,0}, \ldots, H_{0,L_0-1}), \ldots (H_{k,0}, \ldots, H_{k,L_k-1}))$
}
\Ensure{Random sample $X \sim \Discrete(a_0,\dots,a_{n-1})$}
\Procedure{DDG}{$(a_0,\dots,a_{n-1})$, $k$, $\mathbf{L}$, $\mathbf{H}$}
\State $U \gets \Call{Uniform}{A}$
  \Comment{draw uniform variate (\cref{alg:uniform})}
  \label{algline:fldr-uniform}
\State $(d,v) \gets (0,0)$        \Comment{initialize depth and value}
\While{\textbf{true}}
\Comment{loop up to $k$ times}
  \If{$v < L_d$}                  \Comment{hit leaf node}
                                  \label{algline:fldr-hit-leaf}
    \State $X \gets H_{d,v}$       \Comment{label at leaf node}
    \State $\randS' \gets a_{X} - (a_{X} \bmod 2^{k+1-d}) + (U \bmod 2^{k-d})$
                                  \Comment{extract uniform value}
                                  \label{algline:fldr-extract-uniform}
    \State $\randM' \gets a_{X}$    \Comment{extract uniform upper bound}
    \State $\Call{Recycle}{\randS',\randM'}$ \Comment{(\cref{alg:recycle})}
    \State \Return $X$
    \Comment{return the label}
  \EndIf
  \State $v \gets 2 \cdot (v - L_d) + (\floor{U / 2^{k-1-d}} \bmod 2)$
  \Comment{visit random child}
  \State $d \gets d + 1$
  \Comment{increment depth}
\EndWhile
\EndProcedure
\end{algorithmic}
\end{algorithm}

%% file: fig-benchmark.tex
%!TEX root=./main.tex
\begin{figure}[p]
\centering

\captionsetup[subfigure]{aboveskip=0pt,belowskip=0pt,font=bf}
\begin{subfigure}{\linewidth}
\centering
\caption*{Binary Search (\cref{alg:inversion})}
\includegraphics[width=.275\linewidth]{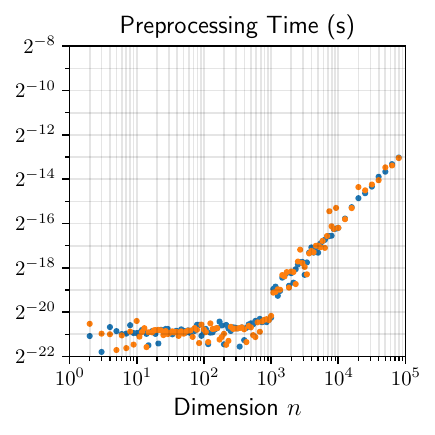}%
\includegraphics[width=.275\linewidth]{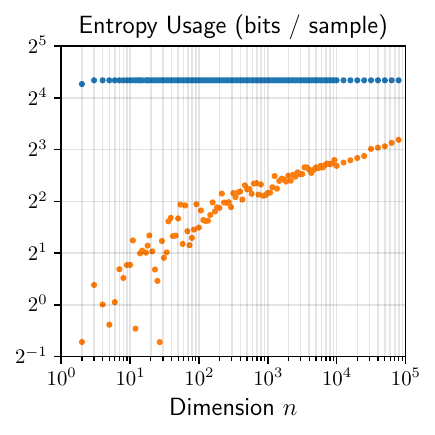}%
\includegraphics[width=.275\linewidth]{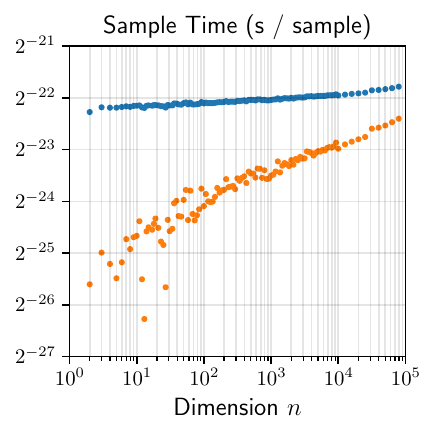}%
\end{subfigure}

\begin{subfigure}{\linewidth}
\centering
\caption*{Lookup Table (\cref{alg:lookup})}
\includegraphics[width=.275\linewidth]{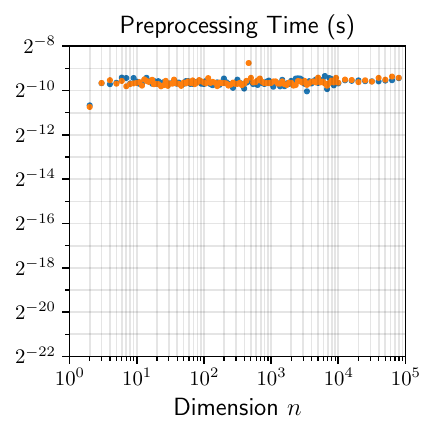}%
\includegraphics[width=.275\linewidth]{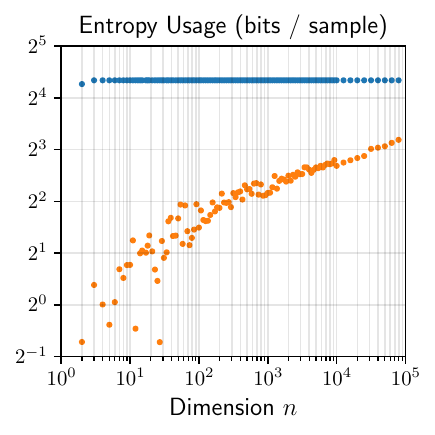}%
\includegraphics[width=.275\linewidth]{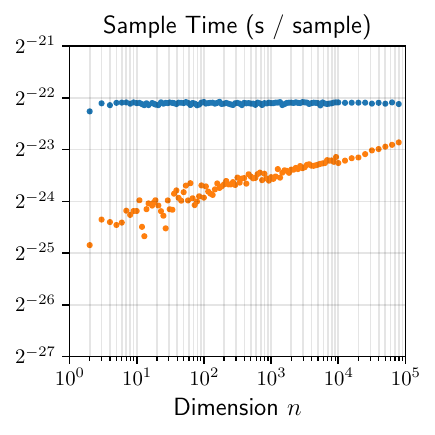}%
\end{subfigure}

\begin{subfigure}{\linewidth}
\centering
\caption*{Alias Method (\cref{alg:alias})}
\includegraphics[width=.275\linewidth]{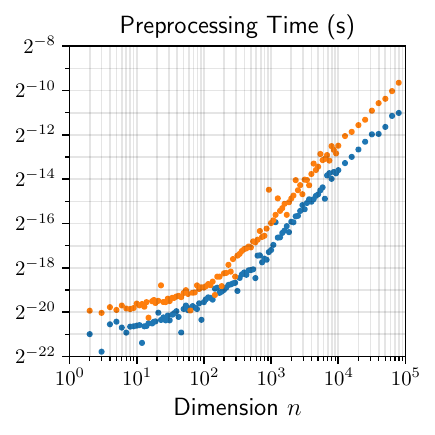}%
\includegraphics[width=.275\linewidth]{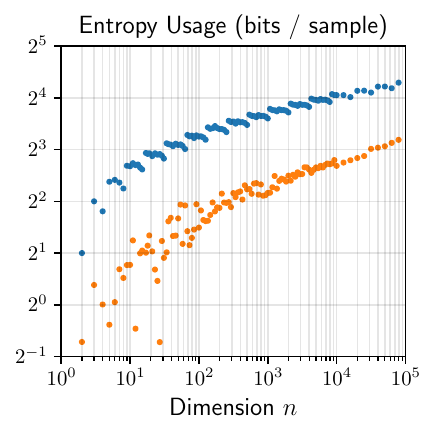}%
\includegraphics[width=.275\linewidth]{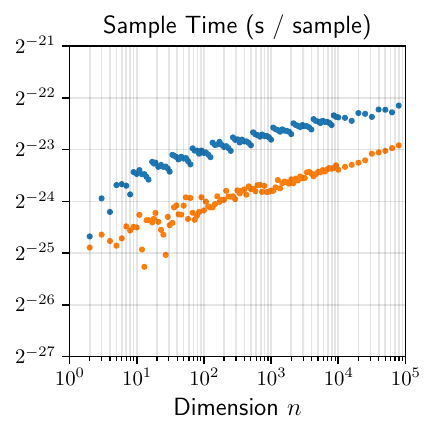}%
\end{subfigure}

\begin{subfigure}{\linewidth}
\centering
\caption*{DDG Sampling (\cref{alg:ddg})}
\includegraphics[width=.275\linewidth]{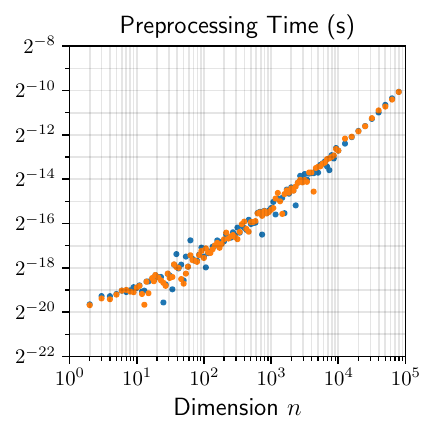}%
\includegraphics[width=.275\linewidth]{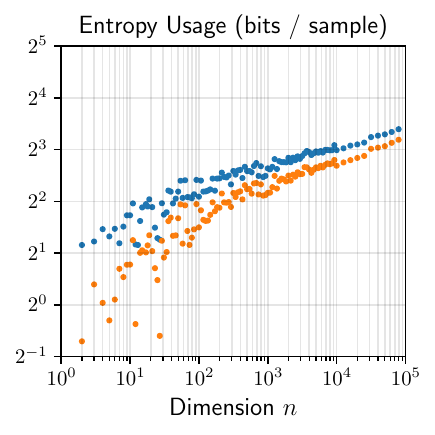}%
\includegraphics[width=.275\linewidth]{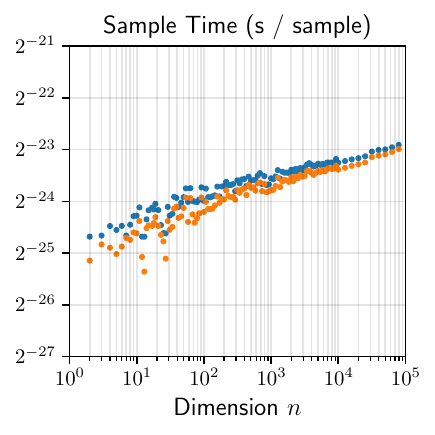}%
\end{subfigure}

\definecolor{tabblue}{HTML}{1f77b4}
\captionsetup{skip=5pt}
\caption{Benchmark comparison of preprocessing time, entropy consumption, and sampling time
for the binary search, lookup table, alias, and DDG tree samplers;
with randomness recycling
(orange {\protect\tikz{\protect\draw[fill=orange] circle[radius=.75ex];}})
and without
(blue {\protect\tikz{\protect\draw[fill=tabblue] circle[radius=.75ex];}}),
on a range of rational discrete distributions whose probabilities
have common denominator equal to $10^6$.}
\label{fig:benchmark}
\end{figure}

%% file: fig-gaussian.tex
%!TEX root=./main.tex
\begin{figure}[t]
\centering
\includegraphics[width=\linewidth]{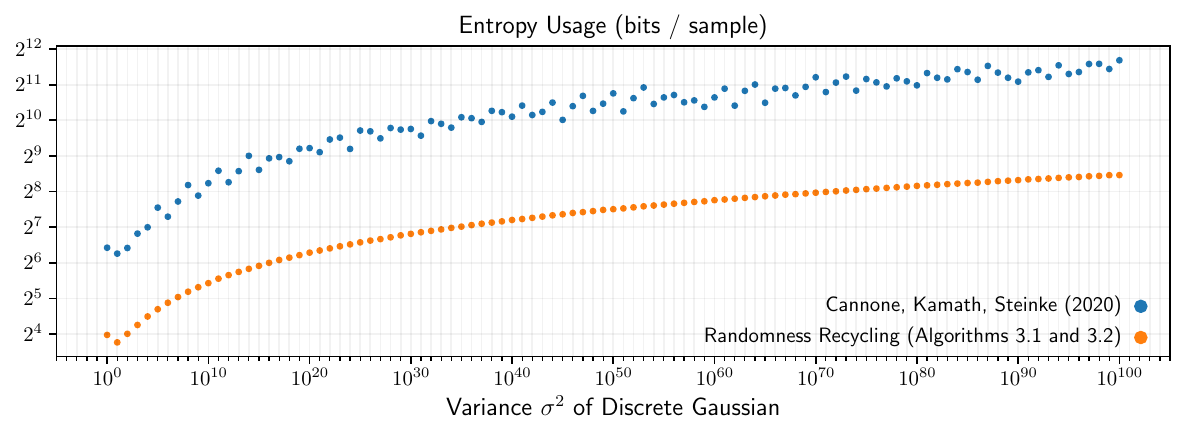}%

\caption{Applying randomness recycling to sample from discrete Gaussian distributions.}
\label{fig:benchmark-gaussian}
\end{figure}

%% file: sec-related.tex
%!TEX root=./main.tex

\section{Related Work}
\label{sec:related}

Existing approaches to online random sampling that are most closely related
to the problem setting in this article are surveyed in
\Cref{sec:intro-existing}.
Here we describe additional literature that studies variants of the sampling
problem and literature on randomness recycling in other settings.

\paragraph{Converting Entropy Streams}

Converting entropy streams from a given input distribution to a desired
output distribution is widely studied in computer science.
This article focuses on the case of online random sampling using an
entropy source that provides an i.i.d.~sequence of fair coin tosses,
to generate output variables with any (nonuniform) rational distribution.
\Citet{knuth1976} introduce a complexity theory for studying nonuniform
sampling algorithms given i.i.d.~fair coins.
\Citet{saad2020fldr} and \citet{draper2025} develop algorithms that are close
to the \citeauthor{knuth1976} entropy-optimal rate using linearithmic,
instead of exponential, memory; but only consider the single-sample case.
\Citet{saad2025} give a deterministic implementation of the
nondeterministic online sampling algorithm from \citet{knuth1976} (see also
\cref{alg:ky-dg}), and efficient algorithms that guarantee zero numerical
error when the probabilities are floating-point numbers.
These works do not use randomness recycling and focus on
achieving near-optimal entropy cost for generating a \textit{single} random
sample,
whereas our method achieves a near-optimal \textit{amortized} entropy cost
for sampling an infinite sequence.

Several articles have studied variants of the entropy conversion problem
under different assumptions.
One popular variant is the problem of extracting unbiased coin tosses
from an i.i.d.\ source with an arbitrary but known
distribution~\citep{elias1972,abrahams1996,roche1991,pae2005a,kozen2014,pae2015,pae2020}.
\Citet{han1997} and \citet{Kozen2022} explore very general reductions
for converting $k$-sided dice dice rolls into $n$-sided dice rolls; the
interval method in the former article is a particularly elegant technique
that allows both the input and output sequences that are be i.i.d., Markov,
or arbitrary stochastic processes.
\Citet{neumann1951} describes a simple method that extracts i.i.d.~fair coins using
a source that emits i.i.d.~dice rolls with unknown bias, which was
explored by several additional authors \citep{hoeffding1970,stout1984,cohen1985,peres1992,pae2006}.
\Citet{elias1972} and \citet{blum1986} consider an unknown entropy source whose stream
is subject to a stationary Markov chain.
Some methods~\citep{elias1972,peres1992,cicalese2006} produce a
variable-length output instead of a single (fixed-length) output at each
invocation, depending on the specific bit pattern from the source.
\Citet{han1993} and \cite{vembu1995} allow the sampler to produce approximate
samples from the target distribution up to a given statistical error
tolerance, providing information-theoretic asymptotic rates.
\Citet{saad2020popl} provide an efficient algorithm that finds
an optimal (in terms of minimizing any $f$-divergence) rational
approximation to an arbitrary distribution, generalizing the results of
\citet{Bocherer2016} who considered the total variation and
Kullback-Leibler divergence.

\paragraph{Specialized Online Random Samplers}

\Citet{Kozen2022} study the problem of exactly transforming i.i.d.~random
streams in an entropy-efficient manner while using limited memory, with a
focus on the composition of these transformation protocols.
The input and output streams may consist of discrete uniforms
or arbitrary discrete distributions with finite support.
For the case of transforming uniform inputs into i.i.d.~uniform outputs
or to i.i.d.~samples from a rational discrete distribution,
they give protocols that use $O(1/\varepsilon)$ space to achieve an expected entropy
inefficiency of at most $\varepsilon > 0$ bits per sample.
They also note that no finite-memory algorithm can produce exact samples
from an irrational target distribution, but such irrational distributions
can be sampled via an infinite sequence of rational distributions;
an overall entropy inefficiency of $\varepsilon$ bits per sample can be
achieved using $O(1/\varepsilon)$ space for the state of each of these
rational samplers.
Our results improve on these space bounds exponentially, requiring only
$O(\log(1/\varepsilon))$ space to achieve the same entropy inefficiency.
The difference arises in our use of a uniform random state that is
maintained across rounds, whereas \citeauthor{Kozen2022} exclusively use
``restart protocols'' that fully reset the state after emitting samples.
In an alternative setting where input entropy stream
is allowed to be nonuniform, \citet{Kozen2022} achieve a space bound of $O(\log(1/\varepsilon)/\varepsilon)$
in the fully arbitrary case, and $O(1/\varepsilon)$ using a cleverly
optimized construction when the entropy source emits a biased coin with weight $1/r$
for some integer $r \geq 3$.

In concurrent work, \citet{shao2025} address the more specific problem of
generating exact i.i.d.~samples from a fixed discrete distribution using
bounded memory and with near-optimal entropy.
Whereas they consider only repeated i.i.d.~samples from a prespecified
distribution, we show how to sample from an arbitrary sequence of
distributions specified in an online setting.
They also restrict attention to the inversion method,
whereas we apply randomness recycling to a variety of discrete sampling methods
(\cref{table:algorithms}).
Their method bears some similarity to \cref{alg:inversion} in our work,
using amplified weights in the sense of \citet{draper2025} to avoid the
need for division when the randomness source is a stream of bits.
The recycling method used by \citet{shao2025}
requires $O(1/\varepsilon)$ space and computation overhead to achieve
a desired $\varepsilon$ entropy bound,
compared to our $O(\log(1/\varepsilon))$ from \cref{theorem:entropy-cost}.
Namely, they convert a large uniform state
of size $2\log(d)/\varepsilon$ bits into i.i.d.~bits
(which wastes up to 2 bits) after every $2/\varepsilon$ samples, to achieve
an amortized inefficiency of $\varepsilon$ bits per sample.
Notably, the $2/\varepsilon$ steps in \citet{shao2025}
coincide exactly with the batch size
$n$ in \cref{eq:ky-batched}
needed to achieve a given efficiency $\varepsilon$ when performing batched
sampling without carrying over any state between batches, this time to
amortize the optimal \citeauthor{knuth1976} toll of $2$ bits in converting
from fair coin tosses to a general random state (instead of vice versa).

\paragraph{Uniform Sampling}
The uniform sampler in \cref{sec:analysis-uniform}
is closely related to a method described by \citet{jacques2004} and
rediscovered by \citet{omer2014}.
\Citet{mennucci2010} empirically analyzes this method for scaling uniform random
number generators, and performs extensive numerical tests with specialized
optimization using bitwise operations on different machine architectures.
Our contributions extend these previous works by
\begin{itemize}[wide,leftmargin=*]
\item using randomness recycling to efficiently generate a random
  sequence whose expected amortized entropy cost is arbitrarily close to
  the information-theoretic lower bound (\cref{theorem:entropy-cost});

\item theoretically analyzing the entropy loss of randomness recycling
  samplers (\cref{sec:analysis});

\item developing randomness recycling for more specialized
  uniform samplers that incorporate widening and batching optimizations
  (\cref{sec:rr-uniform}), which improve performance
  (\cref{fig:benchmark-uniform-full,fig:shuffle}); and

\item leveraging randomness recycling for sampling general nonuniform
  distributions (\cref{sec:rr-general}), which improves the
  runtime and entropy characteristics of diverse algorithms (\cref{fig:benchmark}).
\end{itemize}

\paragraph{Randomness Recycling}

The term ``randomness recycler'' appears to have originated in
\citet{fill2000}, who introduce the concept as an exact (perfect) sampling
technique that stands in contrast to approximate sampling using Markov
chains.
\Citeauthor{fill2000} show how to apply randomness recycling in challenging
combinatorial settings such as generating random independent sets, random
graph colorings, Ising models, random cluster models, and self-organizing
lists in expected linear time.
These randomness recycling algorithms have better time complexity
than exact rejection sampling and are more accurate than approximate
sampling using Markov chains.
The authors note that randomness recycling is not universally applicable
in all scenarios where Markov chains are used, but it can efficiently generate
perfect samples in linear time for challenging combinatorial problems.
Our work develops randomness recycling techniques for sampling
an online random sequence, to obtain improvements in space, time, and
entropy compared to the best-known existing approaches described in
\cref{sec:intro-existing}.

\Citet{impagliazzo1989} describe randomness recycling techniques
in the context of bounded-error probabilistic polynomial time (BPP)
algorithms.
The recycled random state does not satisfy the independence invariant
\labelcref{invariant} and the sampler is allowed to be approximate.
The authors show how to run a BPP algorithm multiple times while using
(approximately independent) recycled coin tosses from applying hash
functions to more efficiently amplify the correctness probability as
compared to using fresh coin tosses on each trial.

%% file: sec-remarks.tex
%!TEX root=./main.tex

\section{Remarks}
\label{sec:remarks}

\paragraph*{Tightness of Space-Entropy Tradeoff Bound}

\Cref{alg:uniform} shows that it is possible to exactly sample
from discrete distributions with weight sums bounded by $d$,
in an online manner,
within $\varepsilon > 0$ of the optimal entropy rate,
using only $O(\log(d/\varepsilon))$ bits of space.
We conjectured that this bound is tight.

\EntropyCostTight*

It is difficult to reason about the space usage of all possible algorithms,
but some examples can demonstrate the plausibility of this conjecture.
The interval method \citep{han1997} achieves the entropy rate exactly,
but it does not have bounded space, even after generating just a single sample
(although the expected space usage can be bounded after one sample).
For a DDG tree method using rejection back edges to represent the tree in finite
space \citep{saad2020fldr}, the recyclable state comes from the distribution of
leaf given label, and if the tree has depth $K$, then there are at least $2^K$
possible denominators.
Representing the recyclable state then requires at least $K$ bits, even if it
is, for example, converted to a uniform using \cref{prop:merge-nonuniform}.
Further, the relationship between the DDG tree's rejection probability and its
entropy inefficiency is analogous to the analysis in \cref{sec:analysis},
so to achieve a rate of $\varepsilon$, the depth $K$ must be roughly
$\log(d/\varepsilon)$ (cf.~\citet[Proposition~3]{draper2025}).

In general, data related to the distribution, such as the weight sum, or
the weight of the sampled outcome, requires roughly $\log(d)$ bits of space.
However, we need not store exact data about the distribution;
$\varepsilon$-approximate data can suffice to achieve an entropy
rate within $\varepsilon$ of optimal---an observation due to David G.~Harris
(pers.~comm.).
Generally, rejection-based methods for exact sampling with finite space seem
to require $\log(1/\varepsilon)$ space to achieve a rate of $\varepsilon$.
Our method's space usage is asymptotic to $2 \log(d/\varepsilon)$ bits,
and space as small as $\log(1/\varepsilon)$ is conceivable,
but anything smaller seemingly could not fit the relevant recycled information.

%% file: appx-baselines.tex
%!TEX root=./main.tex

\section{Baseline Online Random Sampling Algorithms}
\label{appx:baselines}

This appendix provides concrete implementations of several baseline
online samplers discussed in \cref{sec:intro-existing},
when the target distributions $(\dist_i)_{i \ge 1}$ are given
as arrays of integer weights at each round.
The common denominator of the probabilities in $\dist_i$ are most $d \ge 1$.
These implementations admit a precise analysis of their space and time
complexities.

\paragraph{Entropy-Optimal Sampling for One Distribution}
\label{appx:baselines-ky-ddg}

\Cref{alg:ky-ddg} shows an optimized implementation of the
algorithm described in \cref{sec:intro-existing-single}.
For a single sample from a particular target distribution,
explicitly building a DDG tree as in \citet{knuth1976} can require exponential
space in the worst case, so instead we implicitly traverse the tree level-by-level
as proposed in \citet{roy2013}.
Additionally, we compute the binary expansions of the weights during traversal,
instead of precomputing them, to avoid preprocessing time.

\input{alg-ky-ddg}

The loop over levels in \cref{alg:ky-ddg} occurs less than
$H(\dist)+2 = O(\log(k))$ times in expectation.
The \textbf{for} loop over $k$ leaves performs $O(\log(A))$ work per leaf, for a total of
$O(k \log(A))$.
Thus, the overall expected time complexity is bounded as $O(k \log(k) \log(A))$.
The space used is $O(k \log(A))$ bits to store the copied weights $b_i$ used to
compute the binary expansions.
(Alternatively, one could recompute $b_i$ from $a_i$ at each level to save space
at the expense of runtime, but we omit this analysis.)
Over $n$ samples with each $k \leq K$ and $A \leq d$, the expected time complexity
is $O(n K \log(K) \log(d))$ and the space complexity is $O(K \log(d))$ bits.
This algorithm does not maintain any auxiliary state between samples.

\paragraph{Online Entropy-Optimal Sampling}
\label{appx:baselines-ky-dg}

\Cref{alg:ky-dg} shows an optimized implementation of the
algorithm described in \cref{sec:intro-existing-optimal}.
The online entropy-optimal algorithm of \citet{knuth1976} can also be
implemented efficiently using a level-by-level approach as in
\cref{alg:ky-ddg},
but several details require extra care.
The optimum refinement algorithms given in \citet[p.~384]{knuth1976}
and \citet[Algorithm 1]{saad2025}
are specialized to the case of binary-coded probability distributions,
so we replace the two explicit leaf creation steps (for their $p'$ and $p''$)
by an iteration over all new refined labels.
Further, the new leaf order can be fixed to match the input order
for compatibility with the method of \cref{alg:ky-ddg},
which is equivalent to replacing the set $S$ with a queue
in the algorithm of \citet{knuth1976}.
Lastly, the uniform state over the level nodes in the refined subtrees
requires a more complex initialization, based on the number of live nodes
at the current level as computed in \cref{algline:ky-dg-uniform-initialization}.

\input{alg-ky-dg}

In \cref{alg:ky-dg}, after $n$ samples with each $k \leq K$ and $A \leq d$,
the parameters are bounded as $C \leq B \leq d^n$ and $\expect{D} < n\log(K)+2$,
so the (expected) space complexity of the auxiliary state is $O(n \log(d))$,
although it is unbounded in the worst case, even after just the first sample,
because the depth $D$ can be arbitrarily large.
The computation of $2^D \bmod 2B$ in \cref{algline:ky-dg-binary-initialization}
requires time $O(\log(D) \log^2(B))$ when using long multiplication,
and the computation of each $b_i$ requires time
$O(\log^2(B))$ given the value $2^D \bmod 2B$.
Additionally, the computation of $z$ requires time $O(k\log(B) + \log^2(B))$,
so the expected combined
time complexity before the loop is $O((k + \log(n \log(K))) n^2 \log^2(d))$.
The space complexity of the $b_i$ is $O(k n \log(d))$ bits, which is the
same as the overall expected space complexity.
As for \cref{alg:ky-ddg}, the expected time complexity of the loop is
$O(k n \log(k) \log(d))$.
Thus, the overall expected time complexity is
$O(n \log(d) (k \log(k) + (k + \log(n \log(K))) n\log(d)))$.
Over $n$ samples with each $k \leq K$ and $A \leq d$, the expected time complexity
is $O(n^2 \log(d) (K \log(K) + K n \log(d) + n \log(n) \log(d)))$
and the expected space complexity is $O(K n \log(d))$ bits.
The expected size of the auxiliary state after $n$ samples is $O(n \log(d))$ bits.

\paragraph{Interval Method (Arithmetic Coding)}
\label{appx:baselines-hh-interval}

\Cref{alg:hh-interval} shows an optimized implementation of the
algorithm described in \cref{sec:intro-existing-interval}.
Rather than explicitly maintaining the interval state
$[\alpha,\beta) \subseteq [\gamma,\delta) \subseteq [0,1)$
described in \citet{han1997}, we normalize the state to
$[(\alpha-\gamma)/(\delta-\gamma), (\beta-\gamma)/(\delta-\gamma)) \subseteq [0,1)$.
This representation requires only three integers $L, R, B$ to represent
this interval $[L/B, R/B)$ by the endpoint numerators and their common denominator.
Further, we compute the prefix sums using linear preprocessing,
in order to allow binary search for the interval refinement step,
which avoids a linear number of big integer operations per sample.

\input{alg-hh-interval}

\Cref{alg:hh-interval} requires $O(\log(k))$ interval comparisons in expectation,
so the expected time complexity of the loop is $O(\log(k) \log(A)\log(B))$ when using
long multiplication.
The prefix sums require $O(k \log(A))$ time to compute
and $O(k \log(A))$ space to store.
Over $n$ samples with each $k \leq K$ and $A \leq d$,
the interval state grows as $\expect{\log(B)} = O(n \log(d))$,
so the expected total space complexity is $O((n+K) \log(d))$,
and the expected total time complexity is $O(n \log(d) (n \log(K) \log(d) + K))$.
The expected size of the auxiliary state after $n$ samples is $O(n \log(d))$ bits.

%% file: alg-ky-ddg.tex
%!TEX root=./paper.tex
\begin{algorithm}
\caption{Single-sample entropy-optimal sampling}
\label{alg:ky-ddg}
\begin{algorithmic}[1]
\Require{Positive integers $a_0, \dots, a_{k-1}$ with sum $A = a_0 + \dots + a_{k-1}$}
\Ensure{Random sample $X \sim \Discrete(a_0,\dots,a_{k-1})$}
\Procedure{KY-DDG}{$a_0$,$\dots$,$a_{k-1}$}
\State $z \gets 0$   \Comment{initialize uniform state (upper bound is implicit)}
\State $b_0 \gets a_0; \dots; b_{k-1} \gets a_{k-1}$ \Comment{copy weights for computing binary expansions}
\While{\textbf{true}} \Comment{iterate over levels of tree}
  \For{$i \gets 0 ~\mathbf{to}~ k-1$} \Comment{iterate over leaves at current level}
    \If{$b_i \geq A$} \Comment{leaf $i$ exists at this level}
      \If{$z = 0$} \Comment{hit leaf node}
        \State \Return $i$ \Comment{return leaf label}
      \Else \Comment{pass leaf node}
        \State $z \gets z - 1$ \Comment{reduce space of live nodes}
      \EndIf
      \State $b_i \gets b_i - A$ \Comment{remove leaf weight}
    \EndIf
    \State $b_i \gets 2 b_i$ \Comment{double weight for next level}
  \EndFor
  \State $z \gets 2 z + \Call{Flip}{ }$ \Comment{refine uniform state}
\EndWhile
\EndProcedure
\end{algorithmic}
\end{algorithm}

%% file: alg-ky-dg.tex
%!TEX root=./paper.tex
\begin{algorithm}
\caption{Online entropy-optimal sampling}
\label{alg:ky-dg}
\begin{algorithmic}[1]
\Require{Positive integers $a_0, \dots, a_{k-1}$ with sum $A = a_0 + \dots + a_{k-1}$}
\Ensure{Random sample $X \sim \Discrete(a_0,\dots,a_{k-1})$}
\LComment {Auxiliary State Variables $D, B, C$}
\State \textbf{mutable int} $D \gets 0$ \Comment{depth in tree}
\State \textbf{mutable int} $C \gets 1$ \Comment{numerator weight so far}
\State \textbf{mutable int} $B \gets 1$ \Comment{denominator weight so far}
\Procedure{KY-DG}{$a_0$,$\dots$,$a_{k-1}$}
\State $\textbf{update}~B \setp A B$ \Comment{current total weight}
\State $b_0 \gets 2^D C a_0 \bmod 2 B; \dots; b_{k-1} \gets 2^D C a_{k-1} \bmod 2 B$ \Comment{weights for computing binary expansions} \label{algline:ky-dg-binary-initialization}
\State $z \gets \floor{(b_0+\dots+b_{k-1})/B} - 1$ \Comment{reinitialize uniform state (upper bound is implicit)} \label{algline:ky-dg-uniform-initialization}
\While{\textbf{true}} \Comment{iterate over levels of tree}
  \For{$i \gets 0 ~\mathbf{to}~ k-1$} \Comment{iterate over leaves at current level}
    \If{$b_i \geq B$} \Comment{leaf $i$ exists at this level}
      \If{$z = 0$} \Comment{hit leaf node}
        \State $\textbf{update}~C \setp a_i C$ \Comment{update numerator weight}
        \State \Return $i$ \Comment{return leaf label}
      \Else \Comment{pass leaf node}
        \State $z \gets z - 1$ \Comment{reduce space of live nodes}
      \EndIf
      \State $b_i \gets b_i - B$ \Comment{remove leaf weight}
    \EndIf
    \State $b_i \gets 2 b_i$ \Comment{double weight for next level}
  \EndFor
  \State $z \gets 2 z + \Call{Flip}{ }$ \Comment{refine uniform state}
  \State $\textbf{update}~D \setp D + 1$ \Comment{increment depth}
\EndWhile
\EndProcedure
\end{algorithmic}
\end{algorithm}

%% file: alg-hh-interval.tex
%!TEX root=./paper.tex
\begin{algorithm}
\caption{Online interval method sampling}
\label{alg:hh-interval}
\begin{algorithmic}[1]
\Require{Positive integers $a_0, \dots, a_{k-1}$ with sum $A = a_0 + \dots + a_{k-1}$}
\Ensure{Random sample $X \sim \Discrete(a_0,\dots,a_{k-1})$}
\LComment {Auxiliary State Variables $L, R, B$}
\State \textbf{mutable int} $L \gets 0$ \Comment{left numerator so far}
\State \textbf{mutable int} $R \gets 1$ \Comment{right numerator so far}
\State \textbf{mutable int} $B \gets 1$ \Comment{denominator so far}
\Procedure{HH-Interval}{$a_0$,$\dots$,$a_{k-1}$}
\State $A_0 \gets 0; A_1 \gets a_0; \dots; A_n \gets a_0 + \dots + a_{k-1}$ \Comment{prefix sums of weights}
\State $i \gets 0; j \gets k$ \Comment{initialize comparison interval $[A_i/A,A_j/A)$}
\While{$j-i>1$} \Comment{iteratively refine $[L/B,R/B) \subseteq [A_i/A,A_j/A)$}
  \State $m \gets \floor{(i+j)/2}$ \Comment{midpoint of comparison interval}
  \If{$A_m B \leq LA$}\Comment{$A_m/A \leq L/B$}
    \State $i \gets m$ \Comment{narrow comparison interval from left}
  \ElsIf{$A_m B \geq RA$} \Comment{$A_m/A \geq R/B$}
    \State $j \gets m$ \Comment{narrow comparison interval from right}
  \Else \Comment{$[L/B,R/B)$ must be refined}
    \State $D \gets R-L$ \Comment{compute interval width numerator}
    \State $\textbf{update}~(L, R, B) \setp (2L, 2R, 2B)$ \Comment{increase precision}
    \State \IfThenElse{$\Call{Flip}{ } = 1$}{$L \setp L + D$}{$R \setp R-D$} \Comment{refine interval state}
  \EndIf
\EndWhile
\State $\textbf{update}~(L, R, B) \setp (L A - A_i B, R A - A_i B, a_i B)$ \Comment{update interval state}
\State \Return $i$ \Comment{return leaf label}
\EndProcedure
\end{algorithmic}
\end{algorithm}